\documentclass[11pt,a4paper]{article}
\usepackage{amssymb}
\usepackage{amsmath}
\usepackage{amsfonts}
\usepackage{dsfont}
\usepackage{amsthm}
\usepackage{mathrsfs}
\usepackage{hyperref}
\usepackage{color}
\usepackage[margin=2.41cm]{geometry}
\usepackage[all,cmtip]{xy}
\usepackage[utf8]{inputenc}
\usepackage{graphicx}
\usepackage{varwidth}
\usepackage{comment}

\usepackage{upgreek}
\usepackage{rotating}

\usepackage{tikz}
\usetikzlibrary{shapes.geometric}

\usepackage[shortlabels]{enumitem}

\definecolor{darkred}{rgb}{0.8,0.1,0.1}
\hypersetup{
     colorlinks=false,         % false: boxed links; true: colored links,false is default
     linkcolor=darkred,
     citecolor=blue,
}

\theoremstyle{plain}
\newtheorem{theo}{Theorem}[section]
\newtheorem{lem}[theo]{Lemma}
\newtheorem{propo}[theo]{Proposition}
\newtheorem{cor}[theo]{Corollary}

\theoremstyle{definition}
\newtheorem{defi}[theo]{Definition}

\newtheorem{exx}[theo]{Example}
\newenvironment{ex}
  {\pushQED{\qed}\exx}
  {\popQED\endexx}

\newtheorem{remm}[theo]{Remark}
\newenvironment{rem}
  {\pushQED{\qed}\remm}
  {\popQED\endremm}

\numberwithin{equation}{section}

\def\nn{\nonumber}

\def\bbK{\mathbb{K}}
\def\bbR{\mathbb{R}}
\def\bbC{\mathbb{C}}

\def\bbZ{\mathbb{Z}}

\def\bbS{\mathbb{S}}

\def\Hom{\mathrm{Hom}}

\def\End{\mathrm{End}}

\def\Aut{\mathrm{Aut}}

\def\id{\mathrm{id}}
\def\Id{\mathrm{Id}}

\def\dd{\mathrm{d}}

\def\cc{\mathrm{c}}

\def\1{\mathbf{1}}
\def\oone{\mathds{1}}
\def\op{\mathrm{op}}

\def\Loc{\mathbf{Loc}}

\def\Locc{\mathbf{Loc}_{\text{\large $\diamond$}}^{}}
\def\Man{\mathbf{Man}}

\def\AQFT{\mathbf{AQFT}}
\def\2AQFT{\mathbf{2AQFT}}
\def\Fun{\mathbf{Fun}}
\def\Open{\mathbf{Open}}
\def\Disk{\mathbf{Disk}}

\def\Set{\mathbf{Set}}
\def\Alg{\mathbf{Alg}}

\def\Vec{\mathbf{Vec}}

\def\flip{\mathrm{flip}}

\def\CC{\mathbf{C}}
\def\DD{\mathbf{D}}

\def\EE{\mathbf{E}}

\def\TT{\mathbf{T}}
\def\Cat{\mathbf{Cat}}

\def\Pr{\mathbf{Pr}}

\def\QCoh{\mathbf{QCoh}}
\def\Rep{\mathbf{Rep}}
\def\Mod{\mathbf{Mod}}
\def\catEnd{\mathbf{B}\mathrm{End}}

\def\AAA{\mathfrak{A}}
\def\BBB{\mathfrak{B}}

\def\BBB{\mathfrak{B}}

\def\FFF{\mathfrak{F}}
\def\GGG{\mathfrak{G}}

\def\CCR{\mathfrak{CCR}}

\def\AAAA{\boldsymbol{\mathfrak{A}}}
\def\BBBB{\boldsymbol{\mathfrak{B}}}

\def\a{\mathbf{a}}
\def\b{\mathbf{b}}

\def\O{\mathcal{O}}
\def\P{\mathcal{P}}

\def\As{\mathsf{As}}

\def\colim{\mathrm{colim}}

\def\bicolim{\mathrm{bicolim}}

\newcommand\und[1]{\underline{#1}}
\newcommand\ovr[1]{\overline{#1}}

\DeclareMathOperator*{\Motimes}{\text{\raisebox{0.25ex}{\scalebox{0.8}{$\bigotimes$}}}}

\DeclareMathOperator*{\bigboxtimes}{\text{\raisebox{-0.5ex}{\scalebox{1.6}{$\boxtimes$}}}}
\def\sk{\vspace{1mm}}

\makeatletter
\let\@fnsymbol\@alph
\makeatother

\title{%
Categorification of algebraic quantum field theories
}

\author{%
Marco Benini$^{1,2,a}$, 
Marco Perin$^{3,b}$,
Alexander Schenkel$^{3,c\,\ast}$\ and\ 
Lukas Woike$^{4,d}$\vspace{4mm}\\
{\small ${}^1$ Dipartimento di Matematica, Universit\`a di Genova,}\\
{\small Via Dodecaneso 35, 16146 Genova, Italy.}\vspace{2mm}\\
{\small ${}^2$ INFN, Sezione di Genova,}\\
{\small Via Dodecaneso 33, 16146 Genova, Italy.}\vspace{2mm}\\
{\small ${}^3$ School of Mathematical Sciences, University of Nottingham,}\\
{\small University Park, Nottingham NG7 2RD, United Kingdom.}\vspace{2mm}\\
{\small ${}^4$ Fachbereich Mathematik, Universit\"at Hamburg,}\\
{\small Bundesstr.~55, 20146 Hamburg, Germany.}\vspace{4mm}\\
{\small \begin{tabular}{ll}
Email: & ${}^a$~\texttt{benini@dima.unige.it}\\
& ${}^b$~\texttt{marco.perin@nottingham.ac.uk}\\
& ${}^c$~\texttt{alexander.schenkel@nottingham.ac.uk}\\
& ${}^d$~\texttt{lukas.jannik.woike@uni-hamburg.de}\vspace{2mm}
\end{tabular}
}
}

\date{February 2021}

\begin{document}

\maketitle

\vspace{-5mm}

\begin{abstract}
\noindent This paper develops a concept of $2$-categorical algebraic quantum field theories (2AQFTs) that assign locally presentable linear categories to spacetimes. It is proven that ordinary AQFTs embed as a coreflective full $2$-subcategory into the $2$-category of 2AQFTs. Examples of 2AQFTs that do not come from ordinary AQFTs via this embedding are constructed by a local gauging construction for finite groups, which admits a physical interpretation in terms of orbifold theories. A categorification of Fredenhagen's universal algebra is developed and also computed for simple examples of 2AQFTs.
\end{abstract}

\vspace{-1mm}

\paragraph*{Report no.:} ZMP-HH/20-8, Hamburger Beitr\"age zur Mathematik Nr.\ 830
\vspace{-2mm}

\vspace{-2mm}

\vspace{-1mm}

\renewcommand{\baselinestretch}{0.8}\normalsize
\tableofcontents
\renewcommand{\baselinestretch}{1.0}\normalsize

%\newpage

%%%%%%%%%%%%%%%%%%%%%%%%%%%%%%%%%%%%%%%%%%%%%%%%
%%%%%%%%%%%%%%%%%%%%%%%%%%%%%%%%%%%%%%%%%%%%%%%%

\section{\label{sec:intro}Introduction and summary}
An algebraic quantum field  theory (AQFT) is a functor $\AAA : \CC\to\Alg_\bbK$
from a category of spacetimes $\CC$ to the category of associative and 
unital algebras over a field $\bbK$, which satisfies certain physically motivated axioms
such as Einstein causality and the time-slice axiom \cite{HK,BFV,FewsterVerch}.
The (non-commutative) algebra $\AAA(c)$ that is assigned to an object 
$c\in\CC$ is interpreted as the algebra of quantum observables of the theory on the spacetime $c$.
\sk

Describing quantum observables in terms of ordinary algebras in $\Alg_\bbK$
is however insufficient to capture the important, but rather subtle,
higher categorical structures that feature in gauge theories.
For instance, in the context of the BRST/BV formalism \cite{FredenhagenRejzner,FredenhagenRejzner2},
the quantum observables of a gauge theory are described by differential graded algebras (dg-algebras) 
and the latter contain in general more information than their $0$-th cohomology, which is the
ordinary algebra of gauge invariant quantum observables. 
An axiomatic framework for homotopy-coherent AQFTs with values in dg-algebras
was developed in \cite{BSWhomotopy}, see also \cite{BSreview} for a review and \cite{LinearYM} for concrete examples. 
In these works it was also shown that the higher structures encoded by 
dg-algebras are crucial for formalizing descent (i.e.\ local-to-global)
properties of quantum gauge theories.
\sk

The main aim of this paper is to develop another concept of higher categorical AQFTs
which describe quantum observables in terms of locally presentable
$\bbK$-linear categories. As we explain in detail below, 
such AQFTs are more sensitive to global aspects of quantum gauge theories
than the homotopy AQFTs based on dg-algebras studied previously in 
\cite{BSWhomotopy,FredenhagenRejzner,FredenhagenRejzner2}. For example, 
they properly capture finite gauge transformations in contrast to only
infinitesimal ones.
\sk

In order to motivate why it is reasonable to describe quantum observables of gauge theories 
by locally presentable $\bbK$-linear categories, let us first recall why ordinary AQFTs
assign associative and unital $\bbK$-algebras. From the point of view of quantum theory,
a (non-commutative) algebra $A\in \Alg_\bbK$ is interpreted as a
quantized function algebra on the phase space $X$ of a physical system,
i.e.\ $A$ arises as a (deformation) quantization of the commutative algebra
$\O(X)$ of $\bbK$-valued functions on $X$. If $X$ is a sufficiently ``nice''
space (in technical terms, $X$ is an affine scheme over $\bbK$), 
there is no loss of information when passing from $X$ to its function
algebra $\O(X)$. This explains why it is justified to quantize the space
$X$ by quantizing its function algebra $\O(X)$. 
\sk

However, many important examples of
phase spaces that feature in physics are {\em not} of this ``nice'' kind. For instance,
if the phase space $X$ is a {\em stack}, as it happens to be in a gauge theory,
it is in general not true that $X$ is faithfully encoded by its function algebra
$\O(X)$, which in this case is a dg-algebra.  To understand the example below,
let us recall that stacks are higher categorical spaces that 
are compared by a suitable kind of weak equivalences (in contrast to isomorphisms),
see e.g.\ \cite{ToenAffine} for the relevant model category structure.
To ensure that the assignment of function dg-algebras $X\mapsto \O(X)$ is well-defined,  i.e.\
compatible with the weak equivalences of stacks,  one has to study dg-algebras
up to weak equivalences too,  which in this case are given by dg-algebra maps
that induce quasi-isomorphisms on the underlying chain complexes.
In particular, this implies that the function dg-algebra on a stack is just determined up to weak equivalence.
(The precise statement here is that the assignment of function dg-algebras
to stacks is a Quillen functor between certain model categories,
see \cite{ToenAffine} and \cite{BSWhomotopy} for the details.)
As an illustrative example of a stack that is not faithfully encoded by its function dg-algebra, 
let $G$ be a finite group and consider the quotient stack $\mathbf{B}G := \{\ast\}//G$.  The latter 
is a non-trivial stack, namely the classifying stack of principal $G$-bundles. The corresponding function dg-algebra 
$\O(\mathbf{B}G) = C^\bullet(G,\bbK)$ is then given by 
the group cochains with values in the 
trivial $G$-representation $\bbK$. Taking for example $G=\bbZ_2$, 
the cyclic group of order $2$, all cohomology groups $H^n(\bbZ_2,\bbK)=0$,
for $n\neq 0$, are trivial if $\bbK$ has characteristic zero, and $H^0(\bbZ_2,\bbK)=\bbK$. 
It then follows that
$\O(\mathbf{B}\bbZ_2) \simeq  \bbK = \O(\{\ast\})$ is weakly equivalent in the model category
of dg-algebras to the function algebra of the point $\{\ast\}$, i.e.\ all information about the group $G=\bbZ_2$
is lost when passing from the stack $\mathbf{B}G$ to its function dg-algebra.
As a consequence, it is in general {\em not} reasonable to quantize a stack $X$ by quantizing its function dg-algebra $\O(X)$. 
We would like to emphasize that these issues arise for {\em finite gauge transformations}
and not for infinitesimal gauge transformations. In particular, every Lie algebroid $X=Y//\mathfrak{g}$
is completely determined by its function dg-algebra,
which in this case is given by the Chevalley-Eilenberg cochains
$\O(X) = \mathrm{CE}^\bullet(\mathfrak{g}, \O(Y))$ on the Lie algebra
$\mathfrak{g}$ with values in the ordinary function algebra $\O(Y)$
of the affine scheme $Y$.  As a consequence,  such models
with infinitesimal gauge transformations may be described within
the BRST/BV formalism \cite{FredenhagenRejzner,FredenhagenRejzner2,LinearYM}
and its axiomatization in the framework of homotopy AQFT \cite{BSWhomotopy}.
\sk

The feature that stacks are in general not completely determined
by their function dg-algebras is well-known to algebraic geometers,
see e.g.\ \cite{Brandenburg} for an excellent overview, 
who have also proposed the following interesting solution: Instead of assigning
a function dg-algebra $\O(X)$  to a space or stack $X$, it is better
to assign the category $\QCoh(X)$ of quasi-coherent sheaves on $X$. 
The latter is a locally presentable symmetric monoidal $\bbK$-linear category
that should be interpreted roughly as the category of vector bundles over $X$.
This is indeed a better choice because, by a theorem of Lurie \cite{LurieTannaka}, 
every geometric stack $X$ can be reconstructed from its quasi-coherent sheaf category
$\QCoh(X)$. This fact becomes evident in our illustrative example $\mathbf{B}G$ from above:
We observe that $\QCoh(\mathbf{B}G)= \Rep_{\bbK}(G)$ is the symmetric monoidal 
category of $\bbK$-linear representations of $G$. By Tannakian reconstruction,
$\Rep_{\bbK}(G)$ encodes the full information about the group $G$, hence
$\QCoh(\mathbf{B}G)$ is indeed much richer than the function dg-algebra $\O(\mathbf{B}G)$
considered in the previous paragraph (think of $G=\bbZ_2$ for example). Furthermore, for a ``nice'' space $X$,
i.e.\ an affine scheme over $\bbK$, the usual function algebra $\O(X)$ can be recovered as follows: 
One finds that in this case $\QCoh(X) \simeq \Mod_{\O(X)}$ is
the symmetric monoidal $\bbK$-linear category of (right) modules over the commutative
algebra $\O(X)$, hence $\O(X)$ can be reconstructed from $\QCoh(X)$
as the endomorphism algebra $\End(\O(X))\cong \O(X)$ of the rank $1$
free module $\O(X)\in \Mod_{\O(X)}$, i.e.\ as the endomorphism algebra of the monoidal unit of 
$\QCoh(X)$. This means that for ``nice'' spaces the function algebra perspective
and the quasi-coherent sheaf category perspective are compatible 
and carry the same information.
\sk

The previous paragraph explained the need to move 
from the function algebra perspective to the quasi-coherent sheaf 
category perspective. This, however, raises another question:
What does it mean to quantize a quasi-coherent sheaf category? As an illustrative example,
let us start with the case where the space $X$ is ``nice'', i.e.\ affine, and assume that we already have
a non-commutative algebra $A\in\Alg_\bbK$ that quantizes the function algebra $\O(X)$.
We may then form the locally presentable $\bbK$-linear category $\Mod_A$ of right $A$-modules
and interpret it as a quantization of the quasi-coherent sheaf category $\QCoh(X) \simeq \Mod_{\O(X)}$.
It is important to observe the following structural difference between $\Mod_A$ and $\Mod_{\O(X)}$:
The tensor product $\otimes_{\O(X)}^{}$ on $\Mod_{\O(X)}$ is only well-defined because
$\O(X)$ is a {\em commutative} algebra and hence every right $\O(X)$-module is automatically
an $\O(X)$-bimodule. Since the quantized algebra $A$ is  non-commutative,
there is no counterpart on $\Mod_A$ of the tensor product structure on 
$\Mod_{\O(X)}$. However, there is a counterpart on $\Mod_A$ of 
the monoidal unit $\O(X)\in\Mod_{\O(X)}$, which is given by the 
object $A\in\Mod_A$ of the $\bbK$-linear category $\Mod_A$. This suggests
that the quantization of the {\em symmetric monoidal} $\bbK$-linear category
$\QCoh(X)$ should be a {\em pointed} $\bbK$-linear category, 
i.e.\ a $\bbK$-linear category together with the choice of an object in it,  see e.g.\  \cite{JohnsonFreyd} and \cite{BZBJ}.
We would like to emphasize that this idea was made precise in the 
framework of derived algebraic geometry, see \cite{DAG,DAG2} and also To\"en's ICM 
2014 contribution \cite{ToenICM}. In this context, the quantization of a derived stack $X$
endowed with an $n$-shifted symplectic structure  is described by a quantization
of the symmetric monoidal $\infty$-category $\QCoh(X)$ of quasi-coherent sheaves
as an $E_n$-monoidal $\infty$-category. Since the phase space of a physical system
is $0$-shifted symplectic, we recover our intuition that one should quantize 
$\QCoh(X)$ as an $E_0$-monoidal, i.e.\ pointed, $\bbK$-linear category. 
\sk

Let us explain in more detail the content of the present paper and the results we obtain: 
In Section \ref{sec:1AQFT} we introduce an equivalent perspective on ordinary AQFTs as prefactorization
algebras \cite{CostelloGwilliam} with values in the symmetric monoidal
category $\Alg_\bbK$ of associative and unital $\bbK$-algebras.
This perspective will be used in Section \ref{sec:2AQFT} to introduce
our concept of categorified AQFTs (called {\em 2AQFTs}) 
that describe quantum observables by locally presentable $\bbK$-linear
categories, in contrast to associative and unital algebras. 
In more detail, we define a 2AQFT as a (weak) prefactorization algebra
on an orthogonal category $\ovr{\CC}$  (cf.\ Definition \ref{def:orthogonalcategory}) with 
values in the symmetric monoidal  $2$-category $\Pr_\bbK$ of locally 
presentable $\bbK$-linear categories. In Section \ref{sec:trunc}
we explore the relationship between ordinary AQFTs and our concept of 2AQFTs.
We construct a biadjunction (cf.\ Theorem \ref{theo:biadjunction})
that exhibits the $1$-category of ordinary AQFTs as a coreflective full
$2$-subcategory of the $2$-category of 2AQFTs.  This implies in particular that our framework
for 2AQFTs includes ordinary AQFTs faithfully.  Furthermore,  every 2AQFT has an 
underlying ordinary AQFT (cf.\ Subsection \ref{subsec:truncation}) that is obtained by 
truncating its higher categorical structures and is interpreted as the underlying
gauge invariant quantum observables.  We shall illustrate through simple examples 
(see e.g.\ Example \ref{ex:trivialAQFTgauging}) that this truncation does not in general
capture the 2AQFT fully,  which means that the latter has 
additional higher structures that are invisible at the level of gauge invariant observables.
Even though we currently do not know of any {\em direct} ways to access these
higher structures through measurements,  they appear to be physically relevant
in more {\em indirect} ways.  For instance,  we shall illustrate in Example \ref{ex:RepGFredCat}
that the higher categorical structures of 2AQFTs are essential for capturing
the desired descent (i.e.\ local-to-global) properties of gauge theories.
\sk

In Section \ref{sec:gauging} we develop a local gauging construction
for AQFTs with finite group actions, which allows us to construct concrete 
examples of 2AQFTs that admit an interpretation as categorified orbifold theories. 
The main result of this section is Theorem \ref{theo:gaugingtruncated}:
We prove that a categorified orbifold theory is truncated, i.e.\ equivalent
to an ordinary AQFT, if and only if a suitable Hopf-Galois condition is fulfilled, which can be interpreted as a non-commutative analog 
of the condition that a $G$-action on a space is {\it free}. 
This matches with the intuition that an orbifold $\sigma$-model 
with a {\it quotient stack} $X//G$ as target boils down 
to an ordinary $\sigma$-model with the {\it quotient space} $X/G$ 
as target whenever $G$ acts freely on the space $X$ (cf.\ Remark \ref{rem:HopfGalois}). 
In Section \ref{sec:localtoglobal} we study a local-to-global extension
for 2AQFTs (called {\em Fredenhagen's universal category}), which is a higher categorical analog 
of {\em Fredenhagen's universal algebra} \cite{Fre1,Fre2,Fre3}. This is related to
factorization homology \cite{AyalaFrancis,BZBJ,BZBJ2}, however we do not restrict ourselves to topological QFTs.  
We develop concrete models for computing Fredenhagen's universal category and provide simple examples
for extensions of 2AQFTs from intervals to the circle $\bbS^1$.
Appendix \ref{app:2operads} introduces the relevant formalism for $2$-categorical
operad theory that we use throughout this paper.

%%%%%%%%%%%%%%%%%%%%%%%%%%%%%%%%%%%%%%%%%%%%%%%%
%%%%%%%%%%%%%%%%%%%%%%%%%%%%%%%%%%%%%%%%%%%%%%%%

\section{\label{sec:1AQFT}AQFTs and algebra-valued prefactorization algebras}
Let us fix once and for all a field $\bbK$ of characteristic zero.
We briefly review the definition of algebraic quantum field theories (AQFTs)
on an orthogonal category $\ovr{\CC} $ in the sense of \cite{BSWoperad}.
We then prove that such theories admit an equivalent description in terms of
prefactorization algebras on $\ovr{\CC}$ with values in the symmetric monoidal
category $\Alg_\bbK$ of associative and unital $\bbK$-algebras.
The latter perspective will be particularly useful for developing a categorification
of AQFTs in Section \ref{sec:2AQFT}.
\begin{defi}\label{def:orthogonalcategory}
An {\em orthogonal category} $\ovr{\CC} := (\CC,\perp)$ is a small category $\CC$ together with 
a subset ${\perp} \subseteq \mathrm{Mor}\CC\, {}_{\mathrm{t}\!}^{}{\times}_{\mathrm{t}}^{} \,\mathrm{Mor}\CC$
of the set of pairs of morphisms with a common target, such that:
\begin{itemize}
\item[(i)] If $(f_1,f_2)\in {\perp}$, then $(f_2,f_1)\in {\perp}$.
\item[(ii)] If $(f_1,f_2)\in {\perp}$, then $(g\,f_1\,h_1,g\,f_2\,h_2)\in {\perp}$,
for all composable $\CC$-morphisms $g$, $h_1$ and $h_2$.
\end{itemize}
We shall denote orthogonal pairs $(f_1,f_2)\in {\perp}$ also by $f_1\perp f_2$.
\end{defi}

\begin{ex}\label{ex:opens}
Let $\Open(M)$ be the category of non-empty open subsets $U \subseteq M$ 
of a manifold $M$ with morphisms $U \to V$ given by subset inclusions $U \subseteq V \subseteq M$. 
We introduce an orthogonality relation $\perp_{M}$ by declaring two morphisms 
$U_1,U_2\subseteq V\subseteq M$ to be orthogonal if and only if $U_1\cap U_2 = \emptyset$.
The orthogonal category $\overline{\Open(M)}:=(\Open(M),\perp_M)$ may be used to describe factorization algebras 
\cite{CostelloGwilliam} and,  for $M=\bbS^1$ the circle,  also chiral conformal 
AQFT \cite{Kawahigashi}.
\end{ex}

\begin{ex}\label{ex:Loc}
The orthogonal category $\ovr{\Loc}$ relevant for locally covariant 
AQFT \cite{BFV,FewsterVerch} is given by the category $\Loc$
of oriented and time-oriented globally hyperbolic Lorentzian manifolds,
with orthogonality relation $\perp_{\Loc}$ determined by causal disjointness, see
e.g.\ \cite{BSWoperad} for the details. 
For a fixed $M\in\Loc$, consider the slice category $\Loc/M$ and pull back
the orthogonality relation $\perp_\Loc$ along the canonical functor $\Loc/M\to \Loc$.
The resulting orthogonal category $\ovr{\Loc/M}$ may be used to describe
Haag-Kaster style AQFTs \cite{HK} on a fixed $M\in\Loc$.
\end{ex}

\begin{defi}\label{def:AQFT}
An {\em algebraic quantum field theory} (AQFT) on an orthogonal category
$\ovr{\CC}$ is a functor $\AAA : \CC\to \Alg_\bbK $ to the category 
of associative and unital $\bbK$-algebras
that satisfies the $\perp$-commutativity property: For all orthogonal pairs 
$(f_1:c_1\to t)\perp (f_2 : c_2\to t)$, the induced commutator
\begin{flalign}
\big[\AAA(f_1)(-),\AAA(f_2)(-)\big]_{\AAA(t)}^{~}\,:\, \AAA(c_1)\otimes \AAA(c_2)~\longrightarrow~\AAA(t)
\end{flalign}
is zero. The category of AQFTs on $\ovr{\CC}$ is the full subcategory
\begin{flalign}
\AQFT(\ovr{\CC})\,\subseteq\,\Fun(\CC,\Alg_\bbK)
\end{flalign}
of the functor category that consists of all $\perp$-commutative functors.
\end{defi}

In preparation for our definition of categorified AQFTs in Section \ref{sec:2AQFT},
we prove that the category $\AQFT(\ovr{\CC})$ is equivalent to the category of $\Alg_\bbK$-valued
prefactorization algebras on $\ovr{\CC}$. The following definition introduces a
colored operad $\P_{\ovr{\CC}}$ that generalizes the prefactorization 
operad of Costello and Gwilliam \cite{CostelloGwilliam} to an arbitrary orthogonal category
$\ovr{\CC}$. The operad of Costello and Gwilliam is recovered by taking 
$\ovr{\CC} = \overline{\Open(M)}$ for a manifold $M$, see Example \ref{ex:opens}. 
For the relevant background and notations for 
colored operads we refer the reader to \cite{Yau,BSWoperad} and also to Appendix \ref{app:2operads}.
\begin{defi}\label{def:PCoperad}
The {\em prefactorization operad} $\P_{\ovr{\CC}}$ associated to an
orthogonal category $\ovr{\CC}$ is the $\Set$-valued colored operad defined
by the following data:
\begin{enumerate}[(1)]
\item The objects of $\P_{\ovr{\CC}}$ are the objects of the category $\CC$.

\item The sets of operations are
\begin{flalign}
\P_{\ovr{\CC}}\big(\substack{t \\ \und{c}}\big)\,:=\, \Big\{ \und{f} := (f_1,\dots,f_n)\in \prod_{i=1}^n \CC(c_i,t)\,:\, f_i\perp f_j \text{~for all~}i\neq j\,\Big\}\quad,
\end{flalign}
for each object $t\in\CC$ and each tuple of objects $\und{c} :=(c_1,\dots,c_n)\in \CC^n$. For the empty tuple $\und{c}=\emptyset$,
we set $\P_{\ovr{\CC}}\big(\substack{t \\ \emptyset}\big) := \{\ast_t\}$ to be a singleton.

\item The composition maps $\gamma : \P_{\ovr{\CC}}\big(\substack{t \\ \und{a}}\big)\times\prod_{i=1}^n 
\P_{\ovr{\CC}}\big(\substack{a_i \\ \und{b}_i}\big)\to \P_{\ovr{\CC}}\big(\substack{t \\ \und{\und{b}}}\big)$,
where $\und{\und{b}} := (\und{b}_1,\dots,\und{b}_n)$ denotes the concatenation of tuples, 
are given by composition in the category $\CC$, i.e.\
\begin{flalign}
 \gamma \big(\und{f},(\und{g}_1,\dots,\und{g}_n) \big)\,:=\, \und{f}\,\und{\und{g}} \,:=\, \big(f_1\,g_{11},\dots, f_1\, g_{1 k_1},\dots,f_n\, g_{n1},\dots,f_n \, g_{n k_n}\big)\quad.
\end{flalign}

\item The identity operations are $\oone := \id_t \in \P_{\ovr{\CC}}\big(\substack{t \\ t}\big)$.

\item The permutation actions $\P_{\ovr{\CC}}(\sigma) : \P_{\ovr{\CC}}\big(\substack{t \\ \und{c}}\big)\to \P_{\ovr{\CC}}\big(\substack{t \\ \und{c}\sigma}\big)$ are given by
\begin{flalign}
\P_{\ovr{\CC}}(\sigma)(\und{f}) \,:=\, \und{f}\sigma := (f_{\sigma(1)},\dots, f_{\sigma(n)})\quad.
\end{flalign}
\end{enumerate}
\end{defi}

Let us endow the category $\Alg_\bbK$ of associative and unital $\bbK$-algebras
with its standard symmetric monoidal structure. The tensor product of two algebras $A,B\in\Alg_\bbK$
is given by the tensor product algebra $A\otimes B$. Concretely, that is the tensor product
of vector spaces with multiplication given by 
$(a\otimes b)\,(a^\prime\otimes b^\prime):=(a\,a^\prime)\otimes (b\,b^\prime)$
and unit element $1_A\otimes 1_B\in A\otimes B$. The monoidal unit is $\bbK\in\Alg_\bbK$
and the symmetric braiding is given by the $\Alg_\bbK$-morphisms 
$\tau : A\otimes B\to B\otimes A\,,~a\otimes b\mapsto b\otimes a$.
The symmetric monoidal category $\Alg_\bbK$ has an associated
$\Set$-valued colored operad (see e.g.\ \cite{Mandell}) 
that we denote by the same symbol $\Alg_\bbK$.  Concretely,
the objects are the objects of $\Alg_\bbK$ and the sets of operations are given by
\begin{flalign}
\Alg_{\bbK}\big(\substack{B \\ \und{A}}\big) \,:=\, \Alg_{\bbK}\Big(\Motimes_{i=1}^n A_i , B\Big)\quad.
\end{flalign}
The composition maps are determined by the monoidal structure,
the identity operations are the identity morphisms, and the permutation
actions are obtained from the symmetric braiding.
\begin{defi}\label{def:PFA}
The category of {\em $\Alg_\bbK$-valued prefactorization algebras on $\ovr{\CC}$}
is defined by
\begin{flalign}
\Alg_{\P_{\ovr{\CC}}}(\Alg_\bbK)\,:=\, [\P_{\ovr{\CC}},\Alg_{\bbK}]\quad,
\end{flalign}
where $\Alg_\bbK$ is regarded as a colored operad (as explained above)
and $[-,-]$ denotes the $\Hom$-category from Remark \ref{rem:HomCats}.
\end{defi}

\begin{rem}\label{rem:PFAunpacked}
Let us unpack this definition by using the definitions from Appendix \ref{app:2operads}.
(These definitions simplify drastically in the present case because
both $\P_{\ovr{\CC}}$ and $\Alg_{\bbK}$ are $\Set$-valued colored operads.
Hence, all coherence data are necessarily trivial. Non-trivial coherence data
will be needed to describe categorified AQFTs in Remark \ref{rem:2AQFTunpacked}.)
\sk

An $\Alg_\bbK$-valued prefactorization algebra
$\FFF\in \Alg_{\P_{\ovr{\CC}}}(\Alg_\bbK)$ 
is given by the following data:
\begin{enumerate}[(1)]
\item For each $c\in \CC$, an associative and unital $\bbK$-algebra $\FFF(c)\in\Alg_\bbK$.

\item For each tuple $\und{f} = (f_1,\dots,f_n)\in \P_{\ovr{\CC}}\big(\substack{t\\ \und{c}}\big)$ of mutually orthogonal
$\CC$-morphisms, an $\Alg_\bbK$-morphism (called {\em factorization product})
\begin{flalign}
\FFF(\und{f}) \,:\, \bigotimes_{i=1}^n \FFF(c_i)~\longrightarrow~\FFF(t)
\end{flalign}
from the tensor product algebra $\bigotimes_{i=1}^n \FFF(c_i)$. For the empty tuple $\und{c} = \emptyset$,
the $\Alg_\bbK$-morphism $\FFF(\ast_t) : \bbK\to \FFF(t)$ associated to the only element 
$\ast_t\in  \P_{\ovr{\CC}}\big(\substack{t\\ \emptyset}\big)$ is necessarily the unit of $\FFF(t)$, 
because $\bbK$ is the initial object in $\Alg_\bbK$.

\end{enumerate}
These data are required to satisfy the following axioms:
\begin{subequations}
\begin{flalign}\label{eqn:PFAaxiomcomposition}
\xymatrix@C=5em{
\ar[rd]_-{\FFF(\und{f}\,\und{\und{g}})~~~}\bigotimes\limits_{i=1}^n\bigotimes\limits_{j=1}^{k_i} \FFF(b_{ij}) \ar[r]^-{\FFF(\und{\und{g}}) := \Motimes_i \FFF(\und{g}_i)} ~&~ \bigotimes\limits_{i=1}^n \FFF(a_i)\ar[d]^-{\FFF(\und{f})}\\
~&~\FFF(t)
}
\end{flalign}
\begin{flalign}\label{eqn:PFAaxiomunit}
\FFF(\id_t)\,=\,\id_{\FFF(t)}\,:\, \FFF(t)~\longrightarrow ~\FFF(t)
\end{flalign}
\begin{flalign}\label{eqn:PFAaxiompermutation}
\xymatrix@C=5em{
\ar[dr]_-{\FFF(\und{f}\sigma)}\bigotimes\limits_{i=1}^n \FFF(c_{\sigma(i)}) \ar[r]^-{\tau_\sigma} ~&~ \bigotimes\limits_{i=1}^n \FFF(c_i)\ar[d]^-{\FFF(\und{f})}\\
~&~\FFF(t)
}
\end{flalign}
\end{subequations}
In the last diagram, we have denoted by $\tau_\sigma :\bigotimes_{i=1}^n \FFF(c_{\sigma(i)}) 
\to \bigotimes_{i=1}^n \FFF(c_i)\,,~a_{\sigma(1)}\otimes\cdots\otimes a_{\sigma(n)}\mapsto 
a_1\otimes\cdots\otimes a_n$  the $\Alg_\bbK$-morphism that permutes the factors 
of the tensor product algebra.
\sk

A morphism $\zeta : \FFF\to\GGG$ in $\Alg_{\P_{\ovr{\CC}}}(\Alg_\bbK)$ is a family of
$\Alg_\bbK$-morphisms $\zeta_c : \FFF(c)\to\GGG(c)$, for all $c\in\CC$,
such that the diagrams
\begin{flalign}\label{eqn:PFAmapaxiom}
\xymatrix@C=5em{
\ar[d]_-{\FFF(\und{f})}\bigotimes\limits_{i=1}^n \FFF(c_i) \ar[r]^-{\Motimes_i \zeta_{c_i}}~&~\bigotimes\limits_{i=1}^n \GGG(c_i)\ar[d]^-{\GGG(\und{f})}\\
\FFF(t)\ar[r]_-{\zeta_t}~&~\GGG(t)
}
\end{flalign}
commute, for all $\und{f}\in  \P_{\ovr{\CC}}\big(\substack{t\\ \und{c}}\big)$.
\end{rem}

It is easy to see that every $\AAA\in\AQFT(\ovr{\CC})$ defines
an $\Alg_\bbK$-valued prefactorization algebra on $\ovr{\CC}$ by introducing
the factorization products
\begin{flalign}
\AAA(\und{f}) \,:\,\xymatrix@C=3em{
\bigotimes\limits_{i=1}^n \AAA(c_i)  \ar[r]^-{\Motimes_i \AAA(f_i)} ~& ~\AAA(t)^{\otimes n} \ar[r]^-{\mu_{\AAA(t)}^{n}}~ &~ \AAA(t)
}\quad,
\end{flalign}
where $\mu_{\AAA(t)}^n : \AAA(t)^{\otimes n}\to \AAA(t)\,,~a_1\otimes\cdots\otimes a_n \mapsto a_1\,\cdots\,a_n$
denotes the $n$-ary multiplication in the associative and unital algebra $\AAA(t)\in\Alg_\bbK$.
Using $\perp$-commutativity and $f_i\perp f_j$, for all $i\neq j$, 
one shows that $\AAA(\und{f})$ is indeed an $\Alg_\bbK$-morphism on the tensor product algebra.
Furthermore, every $\AQFT(\ovr{\CC})$-morphism $\zeta : \AAA\to\BBB$ defines
an $\Alg_{\P_{\ovr{\CC}}}(\Alg_\bbK)$-morphism between the corresponding
prefactorization algebras, hence we obtain a functor $\AQFT(\ovr{\CC})\to \Alg_{\P_{\ovr{\CC}}}(\Alg_\bbK)$.
Note that this functor is fully faithful. 
\sk

Conversely, we have the following
lemma showing that every $\Alg_\bbK$-valued prefactorization algebra
is completely determined by an underlying AQFT.
\begin{lem}
For every $\FFF\in \Alg_{\P_{\ovr{\CC}}}(\Alg_\bbK)$, the factorization product
$\FFF(\und{f})$ for $\und{f} = (f_1,\dots,f_n)\in \P_{\ovr{\CC}}\big(\substack{ t \\ \und{c}}\big)$
factorizes as
\begin{flalign}\label{eqn:FFFfactorizationTMP}
\FFF(\und{f}) \,:\,\xymatrix@C=3em{
\bigotimes\limits_{i=1}^n \FFF(c_i)  \ar[r]^-{\Motimes_i \FFF(f_i)} ~& ~\FFF(t)^{\otimes n} \ar[r]^-{\mu_{\FFF(t)}^{n}}~ &~ \FFF(t)
}\quad,
\end{flalign}
where $\mu_{\FFF(t)}^n$ denotes the $n$-ary multiplication in the associative and unital algebra 
$\FFF(t)\in\Alg_\bbK$. In particular, $\FFF$ is completely determined by its underlying
functor $\FFF : \CC\to\Alg_{\bbK}$, which satisfies the $\perp$-commutativity
property from Definition \ref{def:AQFT} and hence defines an AQFT.
\end{lem}
\begin{proof}
Using the composition maps from Definition \ref{def:PCoperad}, we compute
\begin{flalign}
\gamma\big(\und{f},(\ast_{c_1},\dots,\ast_{c_{k-1}},\id_{c_k},\ast_{c_{k+1}},\dots, \ast_{c_n})\big) \,=\, f_k\quad, 
\end{flalign}
for all $k=1,\dots,n$, where $\ast_{t}\in \P_{\ovr{\CC}}\big(\substack{ t \\ \emptyset}\big)$
denotes the unique arity zero operation. The corresponding commutative diagram in 
\eqref{eqn:PFAaxiomcomposition} then reads as
\begin{flalign}\label{eqn:FFFfactorizationTMPdiagram}
\xymatrix@C=5em{
\bigotimes\limits_{i=1}^n \FFF(c_i) \ar[r]^-{\FFF(\und{f})} ~&~\FFF(t)\\
\ar[u]^-{\Motimes_i \FFF(\ast_{c_i}) \otimes \FFF(\id_{c_k}) \otimes \Motimes_i \FFF(\ast_{c_i})} \bigotimes\limits_{i=1}^{k-1}\bbK \otimes \FFF(c_k) \otimes \bigotimes\limits_{i=k+1}^{n}\bbK \ar[r]_-{\cong}~&~ \FFF(c_k)\ar[u]_-{\FFF(f_k)}
}
\end{flalign}
Using further that $\FFF(\id_{c_k})=\id_{\FFF(c_k)}$ (cf.\ \eqref{eqn:PFAaxiomunit})
and that $\FFF(\ast_{c_i}) : \bbK\to \FFF(c_i)$ is the unit of $\FFF(c_i)\in\Alg_\bbK$ 
(cf.\ Remark \ref{rem:PFAunpacked}), the commutative diagram \eqref{eqn:FFFfactorizationTMPdiagram}
implies that
\begin{flalign}
\FFF(\und{f})\big(1_{\FFF(c_1)}\otimes \cdots\otimes 1_{\FFF(c_{k-1})} \otimes a_k \otimes 1_{\FFF(c_{k+1})}\otimes\cdots \otimes 1_{\FFF(c_n)}\big) \,=\, \FFF(f_k)(a_k)\quad,
\end{flalign}
for all $a_k\in \FFF(c_k)$. By definition of the product of a tensor product algebra, it then follows that
\begin{flalign}
\FFF(\und{f})(a_1\otimes \cdots\otimes a_n)\,=\,\FFF(f_1)(a_1)\,\cdots\,\FFF(f_n)(a_n)\quad,
\end{flalign}
for all $a_1\otimes \cdots\otimes a_n\in \bigotimes_{i=1}^n \FFF(c_i)$,
which proves \eqref{eqn:FFFfactorizationTMP}. Using further that every two elements of the form
$a\otimes 1_B$ and $1_A\otimes b$ commute in a tensor product algebra $A\otimes B$,
it follows that the underlying functor $\FFF : \CC\to \Alg_{\bbK}$ is $\perp$-commutative.
\end{proof}

Summing up, we have proven the following
\begin{theo}\label{theo:AQFTPFAequivalence}
For every orthogonal category $\ovr{\CC}$, there exists a canonical isomorphism
\begin{flalign}
\AQFT(\ovr{\CC})\, \cong\, \Alg_{\P_{\ovr{\CC}}}(\Alg_\bbK)
\end{flalign}
between the category of AQFTs on $\ovr{\CC}$ and the 
category of $\Alg_\bbK$-valued prefactorization algebras on $\ovr{\CC}$.
\end{theo}

\begin{rem}
The equivalent description of AQFTs in terms of $\Alg_\bbK$-valued
prefactorization algebras
provides an interesting conceptual interpretation of the 
$\perp$-commutativity property  from  Definition \ref{def:AQFT}.
From the prefactorization algebra point of view, every quantum field theory
comes with two different kinds of ``multiplications'',  namely the object-wise
products $\mu_{\FFF(c)} :\FFF(c)\otimes\FFF(c)\to \FFF(c)$, for every $c\in \CC$,
and the factorization products $\FFF(\und{f}): \bigotimes_{i=1}^n \FFF(c_i)\to \FFF(t)$,
for every tuple $\und{f}$ of mutually orthogonal $\CC$-morphisms. These two kinds
of ``multiplications'' are compatible with each other because the factorization products
$\FFF(\und{f})$ are $\Alg_\bbK$-morphisms. The $\perp$-commutativity property
is thus a consequence of an Eckmann-Hilton argument.
\end{rem}

%%%%%%%%%%%%%%%%%%%%%%%%%%%%%%%%%%%%%%%%%%%%%%%%
%%%%%%%%%%%%%%%%%%%%%%%%%%%%%%%%%%%%%%%%%%%%%%%%

\section{\label{sec:2AQFT}Definition of 2AQFTs}
The aim of this section is to introduce a categorification of the concept of AQFTs, 
which we shall call {\em $2$-categorical algebraic quantum field theories} (2AQFTs). While ordinary
AQFTs assign associative and unital $\bbK$-algebras to the objects
of an orthogonal category $\ovr{\CC}$, our concept of 2AQFTs will assign
locally presentable $\bbK$-linear categories, cf.\ \cite{AdamekRosicky,BCJF}. 
\sk

Recall that a {\em $\bbK$-linear category} is a category $\DD$ that is enriched
over the symmetric monoidal category $\Vec_\bbK$ of vector spaces over $\bbK$.
Concretely, this means that we have a vector space $\DD(d,d^\prime)\in\Vec_\bbK$ 
of morphisms,  for every pair of objects $d,d^\prime\in \DD$,
and that the composition maps $\circ : \DD(d^\prime,d^{\prime\prime})\times 
\DD(d,d^\prime)\to \DD(d,d^{\prime\prime})$ are $\bbK$-bilinear, for all 
$d,d^\prime,d^{\prime\prime}\in \DD$. Given two $\bbK$-linear categories
$\DD$ and $\EE$, a {\em $\bbK$-linear functor} $F : \DD\to \EE$ is a functor
such that the maps $F : \DD(d,d^\prime)\to \EE(Fd,Fd^\prime)$ are $\bbK$-linear, for all $d,d^\prime\in \DD$.
\sk

A $\bbK$-linear category $\DD$ is called {\em locally presentable} if it is 1.)~cocomplete, i.e.\ has all small colimits,
and 2.)~generated under small colimits by a set $\Gamma\subset \DD_0$ of objects that are $\lambda$-presentable
for some infinite cardinal $\lambda$, 
see e.g.\ \cite{BCJF} for a recollection of the relevant material on locally presentable categories.
The natural concept of functors $F : \DD\to\EE$ between two locally presentable
$\bbK$-linear categories $\DD$ and $\EE$ is given by {\em co-continuous $\bbK$-linear functors},
i.e.\ $\bbK$-linear functors that preserve all small colimits. Natural transformations in this context
are just ordinary natural transformations.
\begin{defi}\label{def:PrKoperad}
The operad $\Pr_\bbK$ of {\em locally presentable $\bbK$-linear categories} 
is the $\Cat$-enriched colored operad (cf.\ Definition \ref{def:2operad}) defined by the following data:
\begin{enumerate}[(1)]
\item The objects are all locally presentable $\bbK$-linear categories.

\item For $\TT$ and $\und{\DD}=(\DD_1,\dots,\DD_n)$ locally presentable $\bbK$-linear categories,
the category of operations is the full subcategory
\begin{flalign}
\Pr_{\bbK}\big(\substack{\TT\\ \und{\DD}}\big)\,\subseteq\,\Fun\Big(\prod_{i=1}^n \DD_i , \TT\Big)
\end{flalign}
of the functor category that consists of all functors $F : \prod_{i=1}^n \DD_i \to \TT$ that
are $\bbK$-linear and co-continuous in each variable. For the empty tuple $\und{\DD}=\emptyset$,
we set $\Pr_{\bbK}\big(\substack{\TT\\ \emptyset}\big):= \Fun(\1,\TT)$, where 
$\1$ is the category with only one object and its identity morphism.

\item The composition functors 
$\gamma: \Pr_{\bbK}\big(\substack{\TT\\ \und{\DD}}\big)\times
\prod_{i=1}^n\Pr_{\bbK}\big(\substack{\DD_i\\ \und{\EE}_i}\big)\to \Pr_{\bbK}\big(\substack{\TT\\ \und{\und{\EE}}}
\big)$ are given by composition of functors and (horizontal) composition of natural transformations, i.e.\
\begin{subequations}
\begin{flalign}
\gamma\big(F,(G_1,\dots,G_n)\big)\, &:=\, F\,\und{G} \, :=\, F\,\prod_{i=1}^n G_i\quad,\\
 \gamma\big( \alpha, (\beta_1,\dots,\beta_n)\big)\,&:=\, \alpha\ast\und{\beta} \,:=\, \alpha\ast\prod_{i=1}^n\beta_i\quad.
\end{flalign}
\end{subequations}

\item The identity $1$-operations are the identity functors $\oone := \id_{\TT}\in \Pr_{\bbK}\big(\substack{\TT\\ \TT}\big)\subseteq\Fun(\TT,\TT)$.

\item The permutation action functors 
$\Pr_{\bbK}(\sigma) : \Pr_{\bbK}\big(\substack{\TT\\ \und{\DD}}\big)\to 
\Pr_{\bbK}\big(\substack{\TT\\ \und{\DD}\sigma}\big)$ are given by
\begin{flalign}
\Pr_{\bbK}(\sigma) (F)\,:=\, F~\flip_\sigma\quad,\qquad
\Pr_{\bbK}(\sigma) (\alpha)\,:=\, \alpha\ast \Id_{\flip_\sigma}\quad,
\end{flalign}
where $\flip_\sigma : \prod_{i=1}^n \DD_{\sigma(i)}\to \prod_{i=1}^n\DD_i$ is the permutation functor
and $\Id_{\flip_\sigma} : \flip_\sigma \Rightarrow \flip_\sigma$ its identity natural transformation.
\end{enumerate}
\end{defi}

\begin{rem}\label{rem:PrKmonoidal}
With some abuse of notation, we will sometimes 
denote by the same symbol $\Pr_\bbK$ the underlying $2$-category
of $1$-ary operations of the $\Cat$-enriched colored operad from 
Definition \ref{def:PrKoperad}. It should be clear from the context whether
we mean by the symbol $\Pr_\bbK$ a $\Cat$-enriched colored operad or a $2$-category.
The underlying $2$-category $\Pr_\bbK$ is known to be (closed) symmetric monoidal
with respect to the Kelly-Deligne tensor product $\DD\boxtimes\EE$ of locally presentable
$\bbK$-linear categories, whose monoidal unit is given by the $\bbK$-linear category $\Vec_\bbK$
of vector spaces, see \cite{Kelly} and also \cite{BCJF} for a review. This symmetric monoidal
structure is linked as follows to our $\Cat$-enriched colored operad from Definition \ref{def:PrKoperad}:
By the universal property of the Kelly-Deligne tensor product, the categories of operations 
\begin{flalign}\label{eqn:DeligneKelly}
\Pr_{\bbK}\big(\substack{\TT\\ \und{\DD}}\big)\, \simeq \, 
\mathbf{Lin}_{\bbK,\cc}\Big(\bigboxtimes_{i=1}^n \DD_i,\TT\Big)
\end{flalign}
are equivalent to the categories of co-continuous $\bbK$-linear functors
out of the Kelly-Deligne tensor product. Hence, the $\Cat$-enriched colored operad $\Pr_\bbK$
can also be understood as the operad associated with the symmetric monoidal $2$-category 
$(\Pr_\bbK,\boxtimes,\Vec_\bbK)$. This alternative perspective will become useful
in some of our computations in Sections \ref{sec:trunc}, \ref{sec:gauging} and \ref{sec:localtoglobal}.
\end{rem}

Recalling Theorem \ref{theo:AQFTPFAequivalence}, 
ordinary AQFTs on $\ovr{\CC}$ are equivalently $\Alg_{\bbK}$-valued prefactorization 
algebras, i.e.\ $\AQFT(\ovr{\CC}) \cong \Alg_{\P_{\ovr{\CC}}}(\Alg_{\bbK})$. 
Replacing the target $\Alg_{\bbK}$ with $\Pr_{\bbK}$ suggests the following
\begin{defi}\label{def:2AQFT}
The $2$-category of {\em $2$-categorical algebraic quantum field theories} (2AQFTs)
on an orthogonal category $\ovr{\CC}$ is defined as
the $\Hom$-$2$-category (cf.\ Remark \ref{rem:HomCats})
\begin{flalign}
\2AQFT(\ovr{\CC}) \,:=\, \Alg_{\P_{\ovr{\CC}}}(\Pr_\bbK) \,:=\,[\P_{\ovr{\CC}},\Pr_{\bbK}]\quad,
\end{flalign}
where $\P_{\ovr{\CC}}$ is the prefactorization operad from 
Definition \ref{def:PCoperad} and $\Pr_\bbK$ is the $\Cat$-enriched colored 
operad from Definition \ref{def:PrKoperad}.
\end{defi}

\begin{rem}\label{rem:2AQFTunpacked}
Let us unpack this definition by using the definitions from Appendix \ref{app:2operads}.
\sk

An object $\AAAA\in \2AQFT(\ovr{\CC})$ is given by the following data:
\begin{enumerate}[(1)]
\item For each $c\in \CC$, a locally presentable $\bbK$-linear category $\AAAA(c)\in\Pr_\bbK$.

\item For each tuple $\und{f} = (f_1,\dots,f_n)\in\P_{\ovr{\CC}}\big(\substack{t \\ \und{c}}\big)$
of mutually orthogonal $\CC$-morphisms,  a functor (called {\em factorization product})
\begin{flalign}
\AAAA(\und{f})\,:\,\prod_{i=1}^n \AAAA(c_i)~\longrightarrow~\AAAA(t)
\end{flalign}
that is $\bbK$-linear and co-continuous in each variable. For the empty tuple
$\und{c} =\emptyset$, this defines an object $\a_t := \AAAA(\ast_t)\in\AAAA(t)$ (called {\em pointing}, see e.g.\ \cite{BZBJ})
that is associated to the only element $\ast_t\in \P_{\ovr{\CC}}\big(\substack{t \\ \emptyset}\big)$.

\item For each $\und{f}\in\P_{\ovr{\CC}}\big(\substack{t \\ \und{a}}\big)$ and $\und{\und{g}} = (\und{g}_1,\dots,\und{g}_n)\in\prod_{i=1}^n\P_{\ovr{\CC}}\big(\substack{a_i \\ \und{b}_i}\big)$, a natural isomorphism 
\begin{flalign}\label{eqn:A2coherencemaps}
\xymatrix@R=1.2em@C=2.5em{
\ar[ddrr]_-{\AAAA(\und{f}\,\und{\und{g}})} \prod\limits_{i=1}^n\prod\limits_{j=1}^{k_i} \AAAA(b_{ij}) \ar[rr]^-{\AAAA(\und{\und{g}})\,:=\,\prod_i \AAAA(\und{g}_i)}~&~~&~\prod\limits_{i=1}^n\AAAA(a_i)\ar[dd]^-{\AAAA(\und{f})} 
\ar@{=>}[dl]_-{\AAAA^2_{(\und{f},\und{\und{g}})}}\\
~&~\phantom{x}~&~\\
 ~&~~&~\AAAA(t)
}
\end{flalign}

\item For each $t\in \CC$, a natural isomorphism
\begin{flalign}\label{eqn:A0coherencemaps}
\xymatrix@R=0.5em@C=3em{
~& \ar@{=>}[dd]_-{\AAAA^0_t}&~\\
\AAAA(t) \ar@/^1.5pc/[rr]^-{\id_{\AAAA(t)}} \ar@/_1.5pc/[rr]_-{\AAAA(\id_t)}~&&~\AAAA(t)\\
~&&~
}
\end{flalign}

\item For each $\und{f}\in\P_{\ovr{\CC}}\big(\substack{t \\ \und{c}}\big)$ 
and permutation $\sigma \in \Sigma_n$, a natural isomorphism
\begin{flalign}\label{eqn:Asigmacoherencemaps}
\xymatrix@R=1.2em@C=2.5em{
\ar[rrdd]_-{\AAAA(\und{f}\sigma)} \prod\limits_{i=1}^n \AAAA(c_{\sigma(i)}) \ar[rr]^-{\flip_\sigma}~&~~&~\prod\limits_{i=1}^n \AAAA(c_{i})\ar[dd]^-{\AAAA(\und{f})} \ar@{=>}[dl]_-{\AAAA^\sigma_{\und{f}}}\\
~&~\phantom{x}~&~\\
~&~~&~\AAAA(t)
}
\end{flalign}
\end{enumerate}
These data are required to satisfy the axioms from Definition \ref{def:pseudomorphism}.
\sk

A $1$-morphism $\zeta : \AAAA\to \BBBB$ in $\2AQFT(\ovr{\CC})$ is given by the following data:
\begin{enumerate}[(1)]
\item For each $c\in\CC$, a co-continuous $\bbK$-linear functor $\zeta_c : \AAAA(c)\to\BBBB(c)$.

\item For each $\und{f}\in\P_{\ovr{\CC}}\big(\substack{t \\ \und{c}}\big)$, a natural isomorphism
\begin{flalign}\label{eqn:2AQFTmorphismcoherence}
\xymatrix@C=5em{
\ar[d]_-{\AAAA(\und{f})}\prod\limits_{i=1}^n\AAAA(c_i) \ar[r]^-{\prod_i\zeta_{c_i}}~&~\prod\limits_{i=1}^n\BBBB(c_i)\ar[d]^-{\BBBB(\und{f})} \ar@{=>}[dl]_-{\zeta_{\und{f}}~}\\
\AAAA(t)\ar[r]_-{\zeta_t}~&~\BBBB(t)
}
\end{flalign}
Note that, for $\und{f} = \ast_t\in \P_{\ovr{\CC}}\big(\substack{t \\ \emptyset}\big)$, this amounts to an 
isomorphism $\zeta_{\ast_t} : \b_t
\stackrel{\cong}{\longrightarrow} \zeta_t(\a_t)$ in $\BBBB(t)$ 
from the pointing $\b_t = \BBBB(\ast_t)\in\BBBB(t)$
to the image of the pointing $\a_t=\AAAA(\ast_t)\in \AAAA(t)$
under the functor $\zeta_t : \AAAA(t)\to\BBBB(t)$.
\end{enumerate}
These data are required to satisfy the axioms from Definition \ref{def:pseudotransformation}.
\sk

A $2$-morphism $\Gamma : \zeta \Rightarrow \kappa$ between two $1$-morphisms
$\zeta,\kappa  : \AAAA\to\BBBB$ in $\2AQFT(\ovr{\CC})$ is given by the following data:
\begin{enumerate}[(1)]
\item For each $c\in\CC$, a natural transformation
\begin{flalign}\label{eqn:2AQFT2morphism}
\xymatrix@R=0.5em@C=3em{
~&~\ar@{=>}[dd]_-{\Gamma_c}~&~\\
\AAAA(c) \ar@/^1.5pc/[rr]^-{\zeta_c} \ar@/_1.5pc/[rr]_-{\kappa_c}~&~~&~ \BBBB(c)\\
~&~~&~
}
\end{flalign}
\end{enumerate}
These data are required to satisfy the axioms from Definition \ref{def:modification}.
\end{rem}

\begin{rem}
Category-valued prefactorization algebras were studied before in the context of 
factorization homology of $2$-manifolds \cite{BZBJ,BZBJ2}.
Our framework for 2AQFTs  allows us to interpret 
the examples studied in these papers as $2$-dimensional topological AQFTs.
This is achieved by considering the orthogonal category $\ovr{\Man_2}$ of $2$-dimensional 
(oriented or framed) manifolds, with orthogonality relation given by disjointness, and restricting 
to topological theories by considering locally constant prefactorization algebras,
i.e.\  prefactorization algebras that assign to every isotopy equivalence $f : M\to N$ in $\ovr{\Man_2}$
an equivalence $\AAAA(f) : \AAAA(M)\to \AAAA(N)$ in the $2$-category $\Pr_\bbK$.
\end{rem}

%%%%%%%%%%%%%%%%%%%%%%%%%%%%%%%%%%%%%%%%%%%%%%%%
%%%%%%%%%%%%%%%%%%%%%%%%%%%%%%%%%%%%%%%%%%%%%%%%

\section{\label{sec:trunc}Inclusion-truncation biadjunction}
In this section we explore the relationship between ordinary AQFTs and our 
concept of 2AQFTs from Definition \ref{def:2AQFT}. We
shall show that every $\AAAA\in\2AQFT(\ovr{\CC})$ has an underlying ordinary AQFT
$\pi(\AAAA) \in\AQFT(\ovr{\CC})$, which we call the truncation  of
$\AAAA$. Our truncation construction is given by a $2$-functor
$\pi : \2AQFT(\ovr{\CC})\to\AQFT(\ovr{\CC})$. We shall also define,
for every $\AAA\in\AQFT(\ovr{\CC})$,  a 2AQFT $\iota(\AAA)\in\2AQFT(\ovr{\CC})$ that 
assigns to each object $c\in\CC$ the locally presentable $\bbK$-linear category $\iota(\AAA)(c)= \Mod_{\AAA(c)}$
of right $\AAA(c)$-modules. This construction is given by an
inclusion pseudo-functor $\iota : \AQFT(\ovr{\CC})\to\2AQFT(\ovr{\CC})$.
Inclusion and truncation are compatible with each other in the sense
that they determine a biadjunction $\iota \dashv \pi$, 
see e.g.\ \cite{Street,Street2} and also \cite{biadjointtriangles}
for the relevant bicategorical background. We prove that this biadjunction exhibits
$\AQFT(\ovr{\CC})$ as a coreflective full $2$-subcategory of $\2AQFT(\ovr{\CC})$.
The conceptual meaning and relevance of this result is as follows: On the one hand, 
ordinary AQFTs can be studied equally well inside the $2$-category of 2AQFTs
by applying the fully faithful inclusion pseudo-functor $\iota : \AQFT(\ovr{\CC})\to\2AQFT(\ovr{\CC})$.
There is no loss of information in doing so, because the unit
$\eta : \id \Rightarrow \pi \,\iota$ of the biadjunction is a natural isomorphism
and hence one can recover every $\AAA\in\AQFT(\ovr{\CC})$ 
from its corresponding 2AQFT $\iota(\AAA)$ by applying the truncation $2$-functor. 
On the other hand, the $2$-category $\2AQFT(\ovr{\CC})$ has in general also 
objects that {\em do not} lie in the essential image of the inclusion pseudo-functor $\iota$.  
These are the genuine 2AQFTs that are not fully determined by their truncation 
to an ordinary AQFT. We refer to Section \ref{sec:gauging} for concrete examples.

\subsection{\label{subsec:truncation}Truncation}
Given any $\AAAA\in\2AQFT(\ovr{\CC})$, we define its truncation
$\pi(\AAAA)\in\AQFT(\ovr{\CC})$,  which is an ordinary AQFT,  
by providing the required data listed in Remark \ref{rem:PFAunpacked}:
\begin{enumerate}[(1)]
\item For each $c\in\CC$, we set 
\begin{flalign}
\pi(\AAAA)(c)\,:=\, \End(\a_c)\,:=\, \AAAA(c)(\a_c,\a_c)
\end{flalign}
to be the endomorphism algebra of the pointing $\a_c\in\AAAA(c)$. (Note that 
this is an associative and unital $\bbK$-algebra, because $\AAAA(c)$ is 
a $\bbK$-linear category.)

\item For each non-empty tuple $\und{f} = (f_1,\dots,f_n)\in \P_{\ovr{\CC}}\big(\substack{t \\ \und{c}}\big)$
of mutually orthogonal $\CC$-morphisms, the given functor $\AAAA(\und{f}):\prod_{i=1}^n\AAAA(c_i)\to \AAAA(t)$
restricts to endomorphism algebras as $\AAAA(\und{f}) : \prod_{i=1}^n \End(\a_{c_i})\to 
\End\big(\AAAA(\und{f})(\a_{c_1},\dots,\a_{c_n})\big)$. Because 
$\AAAA(\und{f})$ is $\bbK$-linear in each variable,
we obtain an $\Alg_\bbK$-morphism $\AAAA(\und{f}) : 
\bigotimes_{i=1}^n \End(\a_{c_i})\to \End\big(\AAAA(\und{f})(\a_{c_1},\dots,\a_{c_n})\big)$ 
from the tensor product algebra. The coherence map in \eqref{eqn:A2coherencemaps} 
that is associated to $(\und{f},\ast_{\und{c}}) := (\und{f},(\ast_{c_1},\dots,\ast_{c_n}))$
provides an isomorphism $\AAAA^2_{(\und{f},\ast_{\und{c}})} : \AAAA(\und{f})(\a_{c_1},\dots,\a_{c_n})\to\a_t $
in the category $\AAAA(t)$, which we use to define the $\Alg_\bbK$-morphism
\begin{flalign}
\nn \pi(\AAAA)(\und{f})\,:\, \bigotimes_{i=1}^n \pi(\AAAA)(c_i)~&\longrightarrow~\pi(\AAAA)(t)\quad,\\
h_1\otimes \cdots \otimes h_n ~&\longmapsto ~\AAAA^2_{(\und{f},\ast_{\und{c}})} \circ \AAAA(\und{f})(h_1,\dots,h_n)\circ (\AAAA^2_{(\und{f},\ast_{\und{c}})})^{-1}\quad, \label{eqn:piA}
\end{flalign}
where $\circ$ denotes composition in $\AAAA(t)$.
As noted in Remark \ref{rem:PFAunpacked}, the $\Alg_\bbK$-morphism
$\pi(\AAAA)(\ast_t) : \bbK\to \pi(\AAAA)(t)$ associated to the empty tuple 
$\ast_t \in \P_{\ovr{\CC}}\big(\substack{t \\ \emptyset}\big)$ is the unit 
$\id_{\a_t}$ of $\pi(\AAAA)(t)$.
\end{enumerate}

Using the axioms of 2AQFTs from Remark \ref{rem:2AQFTunpacked},
it is easy to check that $\pi(\AAAA)$ satisfies 
the axioms of $\Alg_\bbK$-valued prefactorization algebras from Remark \ref{rem:PFAunpacked}.
Hence, $\pi(\AAAA)\in\AQFT(\ovr{\CC})$ is an AQFT by Theorem \ref{theo:AQFTPFAequivalence}.
\sk

Let us consider now a $1$-morphism $\zeta: \AAAA\to \BBBB$ in $\2AQFT(\ovr{\CC})$.
For each $c\in\CC$, the $\bbK$-linear functor $\zeta_c : \AAAA(c)\to \BBBB(c)$ restricts to
endomorphism algebras as $\zeta_c : \End(\a_c)\to \End(\zeta_c(\a_c))$.
The coherence map in \eqref{eqn:2AQFTmorphismcoherence}
that is associated to $\ast_c$ provides an isomorphism $\zeta_{\ast_c} : \b_c\to \zeta_c(\a_c)$
in the category $\BBBB(c)$, which we use to define the $\Alg_\bbK$-morphism
\begin{flalign}
\nn \pi(\zeta)_c\,:\, \pi(\AAAA)(c)~&\longrightarrow~\pi(\BBBB)(c)\quad,\\
h~&\longmapsto~(\zeta_{\ast_c})^{-1}\circ \zeta_c(h)\circ \zeta_{\ast_c}\quad. \label{eqn:pizeta}
\end{flalign}
Using the axioms of $1$-morphisms of 2AQFTs from Remark \ref{rem:2AQFTunpacked},
it is easy to check that $\pi(\zeta) : \pi(\AAAA)\to\pi(\BBBB)$ is a morphism
of $\Alg_\bbK$-valued prefactorization algebras in the sense of Remark \ref{rem:PFAunpacked},
and hence by Theorem \ref{theo:AQFTPFAequivalence} a morphism of AQFTs.
\sk

Let $\Gamma : \zeta\Rightarrow \kappa$ be a $2$-morphism between
two $1$-morphisms $\zeta,\kappa :\AAAA\to\BBBB$ in $\2AQFT(\ovr{\CC})$.
Using the axioms from Remark \ref{rem:2AQFTunpacked},
we obtain a commutative diagram
\begin{flalign}\label{eqn:zetakappadiagramTMP}
\xymatrix@C=6em{
\ar[d]_-{\zeta_{\ast_c}}^-{\cong}\b_c \ar[r]^-{=} ~&~\b_c\ar[d]_-{\cong}^-{\kappa_{\ast_c}}\\
\zeta_c(\a_c)\ar[r]_-{\Gamma_c}~&~ \kappa_c(\a_c)
}
\end{flalign}
in the category $\BBBB(c)$, where isomorphisms are indicated by $\cong$.
Hence $\Gamma_c$ in this diagram is an isomorphism too. 
From \eqref{eqn:zetakappadiagramTMP} and \eqref{eqn:pizeta}, we
compute
\begin{flalign}
\nn \pi(\kappa)_c (h)&= (\kappa_{\ast_c})^{-1}\circ \kappa_c(h)\circ \kappa_{\ast_c}
=(\zeta_{\ast_c})^{-1}\circ(\Gamma_c)^{-1} \circ \kappa_c(h)\circ\Gamma_c\circ \zeta_{\ast_c}\\
&=(\zeta_{\ast_c})^{-1}\circ \zeta_c(h)\circ \zeta_{\ast_c} = \pi(\zeta)_c(h)\quad,
\end{flalign}
where in the third step we used that \eqref{eqn:2AQFT2morphism} is a natural transformation.
Hence, $\pi(\kappa) = \pi(\zeta) : \pi(\AAAA)\to \pi(\BBBB)$ define the same morphism
in $\AQFT(\ovr{\CC})$ and we can set $\pi(\Gamma) := \Id : \pi(\zeta)\Rightarrow \pi(\kappa)$.
\begin{propo}\label{prop:truncation2functor}
For every orthogonal category $\ovr{\CC}$, the construction above 
defines a $2$-functor
\begin{flalign}
\pi \,:\, \2AQFT(\ovr{\CC})~\longrightarrow~\AQFT(\ovr{\CC})\quad,
\end{flalign}
which we call the  truncation $2$-functor.
\end{propo}

\subsection{\label{subsec:inclusion}Inclusion}
Let $\AAA\in\AQFT(\ovr{\CC})$ be an ordinary AQFT, regarded as an $\Alg_\bbK$-valued
prefactorization algebra via Theorem \ref{theo:AQFTPFAequivalence}. We define
its inclusion $\iota(\AAA)\in\2AQFT(\ovr{\CC})$ by providing
the data listed in Remark \ref{rem:2AQFTunpacked}:
\begin{enumerate}[(1)]
\item For each $c\in\CC$, we set 
\begin{flalign}
\iota(\AAA)(c)\,:=\, \Mod_{\AAA(c)}
\end{flalign} 
to be the $\bbK$-linear category of right $\AAA(c)$-modules. This is a  locally presentable $\bbK$-linear category, 
see e.g.\ \cite{BCJF}.

\item For each non-empty tuple $\und{f} = (f_1,\dots,f_n)\in\P_{\ovr{\CC}}\big(\substack{t \\ \und{c}}\big)$
of mutually orthogonal $\CC$-morphisms, the given $\Alg_\bbK$-morphism 
$\AAA(\und{f}) : \bigotimes_{i=1}^n \AAA(c_i)\to \AAA(t)$ induces a 
restriction functor $\AAA(\und{f})^\ast : \Mod_{\AAA(t)}\to \Mod_{\Motimes_{i=1}^n \AAA(c_i)}$,
which admits a left adjoint functor (called the {\em induced module functor})
\begin{flalign}\label{eqn:inducedmodfunctor}
\AAA(\und{f})_! = (-)\otimes_{\Motimes_{i=1}^n \AAA(c_i)}^{} \AAA(t) \,:\, \Mod_{\Motimes_{i=1}^n \AAA(c_i)}~
\longrightarrow~\Mod_{\AAA(t)}\quad.
\end{flalign}
The latter functor is clearly $\bbK$-linear and co-continuous.
Observe further that the functor $\otimes^{n} : 
\prod_{i=1}^n \Vec_\bbK\to\Vec_\bbK\,,~(V_1,\dots V_n)\mapsto V_1\otimes \cdots \otimes V_n$
taking $n$-ary tensor products of vector spaces induces a functor
\begin{flalign}
\otimes^n\,:\, \prod\limits_{i=1}^n\Mod_{\AAA(c_i)}~\longrightarrow~\Mod_{\Motimes_{i=1}^n \AAA(c_i)}\quad
\end{flalign}
that is $\bbK$-linear and co-continuous in each variable. We define by composition
\begin{flalign}
\iota(\AAA)(\und{f})\,:\, \xymatrix@C=3em{
 \prod\limits_{i=1}^n\Mod_{\AAA(c_i)} \ar[r]^-{\otimes^n}~&~\Mod_{\Motimes_{i=1}^n \AAA(c_i)} \ar[r]^-{\AAA(\und{f})_!}~&~\Mod_{\AAA(t)}
}\quad.
\end{flalign}
For the empty tuple $\und{c}=\emptyset$, we set the pointing 
$\iota(\AAA)(\ast_t) := \AAA(t) \in \Mod_{\AAA(t)}$
to be the rank $1$ free $\AAA(t)$-module.

\item The coherence natural isomorphisms in \eqref{eqn:A2coherencemaps} are given by pasting of
\begin{flalign}
\xymatrix@R=1.2em@C=1.5em{
\ar[ddrr]_-{\otimes^{\sum k_i}}\prod\limits_{i=1}^n \prod\limits_{j=1}^{k_i}\Mod_{\AAA(b_{ij})} \ar[rr]^-{\prod_i \otimes^{k_i}}~&~~&~ \prod\limits_{i=1}^n\Mod_{\Motimes_{j=1}^{k_i}\AAA(b_{ij})} \ar[rr]^-{\prod_i \AAA(\und{g}_i)_!}\ar[dd]_-{\otimes^n} \ar@{=>}[dl]_-{(\star)}~&~~&~\prod\limits_{i=1}^{n} \Mod_{\AAA(a_i)}\ar[dd]^-{\otimes^n} \ar@{=>}[ddll]_-{(\star\star)}\\
~&~\phantom{x}~&~~&~~&~\\
~&~~&~   \ar[ddrr]_-{\AAA(\und{f}\, \und{\und{g}})_!} \Mod_{\Motimes_{i=1}^n\Motimes_{j=1}^{k_i}\AAA(b_{ij})} \ar[rr]^-{(\Motimes_i \AAA(\und{g}_i))_!} ~&~~&~ \Mod_{\Motimes_{i=1}^n \AAA(a_i)}\ar[dd]^-{\AAA(\und{f})_!} \ar@{=>}[dl]_-{(\star\star\star)} \\
~&~~&~~&~~&~\\
~&~~&~~&~~&~ \Mod_{\AAA(t)}
}
\end{flalign}
The natural isomorphisms $(\star)$ and $(\star\star)$ are canonically determined
by the coherence isomorphisms for tensor products. (Recall that
the induced module functor \eqref{eqn:inducedmodfunctor} is given by a relative tensor product.)
The natural isomorphism $({\star}{\star}{\star})$ is canonically determined by 
uniqueness (up to a unique natural isomorphism)
of left adjoint functors and the strict composition property 
$(\bigotimes_i \AAA(\und{g}_i))^\ast~\AAA(\und{f})^\ast = \big(\AAA(\und{f})~
\bigotimes_i \AAA(\und{g}_i)\big)^\ast = \AAA(\und{f}\,\und{\und{g}})^\ast$ of the right adjoints, see also
\eqref{eqn:PFAaxiomcomposition}.

\item The coherence natural isomorphisms in \eqref{eqn:A0coherencemaps} are canonically determined
by uniqueness of left adjoint functors
and the strict identity property $\AAA(\id_t)^\ast =\id_{\AAA(t)}^\ast =\id_{\Mod_{\AAA(t)}}$ of the right adjoints,
see also \eqref{eqn:PFAaxiomunit}.

\item The coherence natural isomorphisms in \eqref{eqn:Asigmacoherencemaps} are given by pasting of
\begin{flalign}
\xymatrix@R=1.2em@C=1.5em{
\ar[dd]_-{\otimes^n} \prod\limits_{i=1}^n \Mod_{\AAA(c_{\sigma(i)})} \ar[rr]^-{\flip_\sigma}~&~~&~\prod\limits_{i=1}^n \Mod_{\AAA(c_i)}\ar[dd]^-{\otimes^n} \ar@{=>}[ddll]_-{(\star)~~}\\
~&~~&~\\
\ar[ddrr]_-{\AAA(\und{f}\sigma)_!}\Mod_{\Motimes_{i=1}^n \AAA(c_{\sigma(i)})} \ar[rr]^-{(\tau_\sigma)_!}~&~~&~ \Mod_{\Motimes_{i=1}^n \AAA(c_i)}\ar[dd]^-{\AAA(\und{f})_!} \ar@{=>}[dl]_-{(\star\star)}\\
~&~~&~\\
~&~~&~\Mod_{\AAA(t)}
}
\end{flalign}
The natural isomorphism $(\star)$ is canonically determined
by the coherence isomorphisms for tensor products
and the natural isomorphism $(\star\star)$ is canonically determined by uniqueness 
of left adjoint functors and the strict permutation property
$(\tau_\sigma)^\ast~\AAA(\und{f})^\ast = \big(\AAA(\und{f})~
\tau_\sigma\big)^\ast = \AAA(\und{f}\sigma)^\ast$ of the right adjoints,
see also \eqref{eqn:PFAaxiompermutation}.
\end{enumerate}

Since the coherences in (3-5) are canonically given by coherence isomorphisms, 
one confirms that $\iota(\AAA)\in\2AQFT(\ovr{\CC})$ satisfies the axioms of 2AQFTs from
Remark \ref{rem:2AQFTunpacked}.
\sk

Let us consider now a morphism $\zeta: \AAA\to\BBB$ in $\AQFT(\ovr{\CC})$.
Then the following data defines a $1$-morphism $\iota(\zeta) : \iota(\AAA)\to\iota(\BBB)$
in $\2AQFT(\ovr{\CC})$, see also Remark \ref{rem:2AQFTunpacked}:
\begin{enumerate}[(1)]
\item For each $c\in\CC$, we set
\begin{flalign}
\iota(\zeta)_c \,:= \, (\zeta_c)_!\, : \,\Mod_{\AAA(c)}~\longrightarrow~\Mod_{\BBB(c)}
\end{flalign}
to be the $\bbK$-linear and co-continuous induced module functor 
along the given $\Alg_\bbK$-morphism $\zeta_c:\AAA(c)\to\BBB(c)$.

\item The coherence natural isomorphisms in \eqref{eqn:2AQFTmorphismcoherence} are given by pasting of
\begin{flalign}
\xymatrix@C=5em{
\ar[d]_-{\otimes^n}\prod\limits_{i=1}^n \Mod_{\AAA(c_i)}\ar[r]^-{\prod_i (\zeta_{c_i})_!}~&~\prod\limits_{i=1}^n \Mod_{\BBB(c_i)} \ar@{=>}[dl]_-{(\star)~~} \ar[d]^-{\otimes^n}\\
\ar[d]_-{\AAA(\und{f})_!}\Mod_{\Motimes_{i=1}^n\AAA(c_i)} \ar[r]^-{(\Motimes_i \zeta_{c_i})_!}~&~\Mod_{\Motimes_{i=1}^n\BBB(c_i)}\ar[d]^-{\BBB(\und{f})_!} \ar@{=>}[dl]_-{(\star\star)~~}\\
\Mod_{\AAA(t)} \ar[r]_-{(\zeta_t)_!}~&~\Mod_{\BBB(t)}
}
\end{flalign}
where $(\star)$ is canonically determined by the coherence isomorphisms for tensor products
and $(\star\star)$ is canonically determined by uniqueness
of left adjoint functors and the strict naturality property
$(\bigotimes_i \zeta_{c_i})^\ast~\BBB(\und{f})^\ast = \big(\BBB(\und{f})~
\bigotimes_i \zeta_{c_i}\big)^\ast = \big(\zeta_t~\AAA(\und{f}) \big)^\ast = \AAA(\und{f})^\ast~(\zeta_t)^\ast$ 
of the right adjoints, see also \eqref{eqn:PFAmapaxiom}.
\end{enumerate}

\begin{propo}\label{prop:inclusionpseudofunctor}
For every orthogonal category $\ovr{\CC}$, the construction above 
defines a pseudo-functor
\begin{flalign}
\iota \,:\, \AQFT(\ovr{\CC})~\longrightarrow~\2AQFT(\ovr{\CC})\quad,
\end{flalign}
which we call the inclusion pseudo-functor.
\end{propo}

\subsection{\label{subsec:biadjunction}Biadjunction}
We now prove that the pseudo-functors in Propositions \ref{prop:truncation2functor} and
\ref{prop:inclusionpseudofunctor} determine a biadjunction, with the inclusion 
$\iota:\AQFT(\ovr{\CC})\to\2AQFT(\ovr{\CC})$ as the left adjoint and the truncation 
$\pi: \2AQFT(\ovr{\CC})\to\AQFT(\ovr{\CC})$ as the right adjoint.
\sk

We describe first the unit $\eta: \id\Rightarrow \pi\,\iota$ of this biadjunction,
which is easier than the counit $\epsilon : \iota\,\pi\Rightarrow \id$ because 
$\AQFT(\ovr{\CC})$ is just a $1$-category, hence $\eta$ is a natural transformation
between ordinary functors.
Let $\AAA\in \AQFT(\ovr{\CC})$ be an ordinary AQFT. 
From the explicit descriptions of $\pi$ and $\iota$ in Sections \ref{subsec:truncation} and
\ref{subsec:inclusion}, we observe that
\begin{flalign}
(\pi\,\iota(\AAA))(c)\,=\,\End(\AAA(c)) \,= \, \Mod_{\AAA(c)}\big(\AAA(c),\AAA(c)\big)
\end{flalign}
is the endomorphism algebra of the rank $1$ free module $\AAA(c)\in\Mod_{\AAA(c)}$,
for every $c\in\CC$. We define the $\AAA$-component $\eta_\AAA : \AAA\to \pi\,\iota(\AAA)$
of the unit $\eta$ as the $\AQFT(\ovr{\CC})$-morphism determined
by the component $\Alg_\bbK$-morphisms
\begin{flalign}\label{eqn:etacomponents}
(\eta_\AAA)_c \,:\, \AAA(c)~\longrightarrow~\End(\AAA(c))~~,\quad a~\longmapsto~\mu_{\AAA(c)}(a\otimes -)\quad,
\end{flalign}
for all $c\in\CC$, where explicitly $\mu_{\AAA(c)}(a\otimes -) :\AAA(c)\to\AAA(c)\,,~a^\prime\mapsto a\,a^\prime$
is the right $\AAA(c)$-module endomorphism given by left multiplication by $a\in\AAA(c)$.
Naturality of $\eta_{\AAA}$ in $\AAA\in\AQFT(\ovr{\CC})$ is obvious, hence we have constructed
the desired natural transformation $\eta:\id \Rightarrow \pi\,\iota$. We further observe that
$\eta$ is a natural isomorphism because each of its components \eqref{eqn:etacomponents} 
is an isomorphism, with inverse given by the $\Alg_\bbK$-morphism
$(\eta_\AAA)_c^{-1} : \End(\AAA(c))\to \AAA(c)\,,~h\mapsto h(1_{\AAA(c)})$ that 
evaluates an endomorphism $h$ on the unit element $1_{\AAA(c)}\in\AAA(c)$.
\sk

Using the natural transformation $\eta: \id \Rightarrow \pi\,\iota$,
we can define, for every $\AAA\in\AQFT(\ovr{\CC})$ and $\BBBB\in\2AQFT(\ovr{\CC})$,
a functor between $\Hom$-categories
\begin{flalign}\label{eqn:biadjunctionfunctor}
\widetilde{(-)}\,:\,
\2AQFT(\ovr{\CC})\big(\iota(\AAA),\BBBB\big)~\longrightarrow~\AQFT(\ovr{\CC})\big(\AAA,\pi(\BBBB)\big)
\quad,
\end{flalign}
where we note that the target is a discrete category,  i.e.\ a category with only identity morphisms.
To a $1$-morphism $\zeta : \iota(\AAA)\to \BBBB$ in $\2AQFT(\ovr{\CC})$,
this functor assigns the $\AQFT(\ovr{\CC})$-morphism
\begin{flalign}
\widetilde{\zeta} \,:\, \xymatrix@C=2em{
\AAA \ar[r]^-{\eta_\AAA}~&~\pi\,\iota(\AAA) \ar[r]^-{\pi(\zeta)}~&~\pi(\BBBB)
}\quad.
\end{flalign}
Given any $2$-morphism $\Gamma : \zeta\Rightarrow\kappa$ between 
$1$-morphism $\zeta, \kappa : \iota(\AAA)\to \BBBB$ in $\2AQFT(\ovr{\CC})$,
we have already seen in Section \ref{subsec:truncation}
that $\pi(\zeta) = \pi(\kappa)$, hence setting $\widetilde{\Gamma} = 
\Id : \widetilde{\zeta}\Rightarrow \widetilde{\zeta}=\widetilde{\kappa} $
consistently defines the functor \eqref{eqn:biadjunctionfunctor}.
\begin{theo}\label{theo:biadjunction}
Let $\ovr{\CC}$ be any orthogonal category. Then
the functor \eqref{eqn:biadjunctionfunctor} is an equivalence of categories,
for every $\AAA\in\AQFT(\ovr{\CC})$ and $\BBBB\in\2AQFT(\ovr{\CC})$.
As a consequence, we obtain a biadjunction
\begin{flalign}
\xymatrix{
\iota \,:\, \AQFT(\ovr{\CC}) ~\ar@<0.5ex>[r]&\ar@<0.5ex>[l]  ~\2AQFT(\ovr{\CC}) \,:\,\pi
}\quad,
\end{flalign}
whose left adjoint is the inclusion pseudo-functor
from Proposition \ref{prop:inclusionpseudofunctor}
and whose right adjoint is the truncation $2$-functor from
Proposition \ref{prop:truncation2functor}.
Because the unit $\eta: \id\Rightarrow \pi\,\iota$ is a natural isomorphism, 
this biadjunction exhibits the category $\AQFT(\ovr{\CC})$ 
of ordinary AQFTs as a coreflective full $2$-subcategory of the $2$-category
$\2AQFT(\ovr{\CC})$.
\end{theo}
\begin{proof}
Let us recall from \cite{BCJF} the following fact: 
For any associative and unital $\bbK$-algebra $A \in \Alg_{\bbK}$, denote by 
$\catEnd(A)$ the full $\bbK$-linear subcategory of $\Mod_A \in \Pr_{\bbK}$ 
on the object $A \in \Mod_A$. (Note that $\catEnd(A)$ is just the endomorphism
algebra $\End(A)$ regarded as a $\bbK$-linear category with only one object.) 
Then, for any locally presentable $\bbK$-linear category $\DD \in \Pr_{\bbK}$, the restriction 
along the inclusion $\catEnd(A) \subseteq \Mod_A$ of $\bbK$-linear 
categories induces an equivalence (i.e.\ a fully faithful and essentially surjective functor) 
\begin{flalign}\label{eqn:EndRestriction}
\xymatrix@C=3em{
\mathbf{Lin}_{\bbK,\cc}\big(\Mod_A,\DD\big)\ar[r]^-{\simeq}~&~
\mathbf{Lin}_{\bbK}\big(\catEnd(A),\DD\big)
}
\end{flalign}
from the full subcategory of $\Fun(\Mod_A,\DD)$ that consists of $\bbK$-linear 
and co-continuous functors to the full subcategory of $\Fun(\catEnd(A),\DD)$ 
that consists of $\bbK$-linear functors. Using this result,
we can check that the functor \eqref{eqn:biadjunctionfunctor} 
is fully faithful and essentially surjective as claimed. 
\sk

\textit{Faithful:} Let $\Gamma, \Delta: \zeta \Rightarrow \kappa$ 
be $2$-morphisms between the $1$-morphisms
$\zeta,\kappa: \iota(\AAA) \to \BBBB$ in $\2AQFT(\ovr{\CC})$. 
(Note that $\widetilde{\Gamma} = \widetilde{\Delta}$ is automatic.) 
From \eqref{eqn:zetakappadiagramTMP} we deduce that, for every 
$c \in \CC$,  the morphisms $\Gamma_c = \Delta_c: \zeta_c(\AAA(c)) \to 
\kappa_c(\AAA(c))$ in $\BBBB(c)$ coincide.
This means that the two natural transformations 
$\Gamma_c, \Delta_c: \zeta_c \Rightarrow \kappa_c$ between the co-continuous 
$\bbK$-linear functors $\zeta_c, \kappa_c: \Mod_{\AAA(c)} \to \BBBB(c)$ 
have the same restriction along the inclusion 
$\catEnd(\AAA(c)) \subseteq \Mod_{\AAA(c)}$. 
Recalling that the restriction functor \eqref{eqn:EndRestriction} 
is faithful, we conclude that $\Gamma_c = \Delta_c: \zeta_c \Rightarrow \kappa_c$ coincide 
as natural transformations,  for all $c \in \CC$, and 
hence that $\Gamma = \Delta: \zeta \Rightarrow \kappa$ 
coincide as $2$-morphisms in $\2AQFT(\ovr{\CC})$. 
This shows that the functor \eqref{eqn:biadjunctionfunctor} is faithful.
\sk

\textit{Full:} Let $\zeta,\kappa: \iota(\AAA) \to \BBBB$ be $1$-morphisms 
in $\2AQFT(\ovr{\CC})$ such that $\widetilde\zeta = \widetilde\kappa: 
\AAA \to \pi(\BBBB)$ in $\AQFT(\ovr{\CC})$. (Recall that
$\AQFT(\ovr{\CC})$ only has identity $2$-morphisms.) 
For each $c \in \CC$, consider the morphism $\kappa_{\ast_c} \circ 
(\zeta_{\ast_c})^{-1}: \zeta_c(\AAA(c)) \to \kappa_c(\AAA(c))$ in $\BBBB(c)$. 
Using $\widetilde\zeta = \widetilde\kappa$, one shows that this defines a 
natural transformation between the restrictions along the inclusion functor 
$\catEnd(\AAA(c)) \subseteq \Mod_{\AAA(c)}$ of the co-continuous 
$\bbK$-linear functors $\zeta_c, \kappa_c: \Mod_{\AAA(c)} \to \BBBB(c)$. 
Recalling that the restriction functor \eqref{eqn:EndRestriction} is full, 
there exists a natural transformation $\Gamma_c: \zeta_c \Rightarrow \kappa_c$ 
whose $\AAA(c)$-component is $\kappa_{\ast_c} \circ (\zeta_{\ast_c})^{-1}$. 
We still have to prove that the collection $\Gamma_c$, for all $c \in \CC$, defines 
a $2$-morphism $\Gamma: \zeta \Rightarrow \kappa$ between the 1-morphisms 
$\zeta,\kappa: \iota(\AAA) \to \BBBB$ in $\2AQFT(\ovr{\CC})$. This amounts to
showing that the diagram
\begin{flalign}\label{eqn:biadjunctionproofdiagram}
\xymatrix@C=5em{
\ar[d]_-{\zeta_{\und{f}}}\BBBB(\und{f})\, \prod_i \zeta_{c_i} \ar[r]^-{\Id\ast \prod_i \Gamma_{c_i}}~&~\BBBB(\und{f})\, \prod_i \kappa_{c_i}\ar[d]^-{\kappa_{\und{f}}}\\
\zeta_t\,\iota(\AAA)(\und{f}) \ar[r]_-{\Gamma_t\ast \Id}~&~ \kappa_t\, \iota(\AAA)(\und{f})
}
\end{flalign}
of natural transformations commutes, for all $\und{f}\in\P_{\ovr{\CC}}\big(\substack{t \\ \und{c}}\big)$.
Since this diagram lives in the category
$\Pr_{\bbK}\big(\substack{\BBBB(t)\\ \iota(\AAA)(\und{c})}\big)$,
i.e.\ all functors are $\bbK$-linear and co-continuous in each variable,
we deduce from the equivalences in \eqref{eqn:DeligneKelly} and \eqref{eqn:EndRestriction} 
that the diagram \eqref{eqn:biadjunctionproofdiagram} of natural transformations
commutes if and only if the corresponding component
on the object $\prod_{i=1}^n \AAA(c_i) \in \prod_{i=1}^n \Mod_{\AAA(c_i)}$ commutes.
This can be checked directly by using that $\zeta,\kappa: \iota(\AAA) \to \BBBB$ 
are $1$-morphisms in $\2AQFT(\ovr{\CC})$. (Here the axioms of 
Definition \ref{def:pseudotransformation} enter explicitly.) 
This shows that the functor \eqref{eqn:biadjunctionfunctor} is full.
\sk

\textit{Essentially surjective:} Let $\zeta: \AAA \to \pi(\BBBB)$ 
be any $\AQFT(\ovr{\CC})$-morphism. We denote its component
$\Alg_{\bbK}$-morphisms by $\zeta_c: \AAA(c) \to \End(\b_c)$, for all $c \in \CC$. Recalling that 
$\AAA(c) \in \Alg_{\bbK}$ is naturally isomorphic via $\eta$ (cf.\ \eqref{eqn:etacomponents})
to the endomorphism algebra $\End(\AAA(c))$ of the object $\AAA(c) \in \Mod_{\AAA(c)}$, 
we define a functor $\widehat{\zeta}_c: \catEnd(\AAA(c)) \to \BBBB(c)$ 
that sends the only object $\AAA(c) \in \catEnd(\AAA(c))$ to $\b_c \in \BBBB(c)$ 
and each $\catEnd(\AAA(c))$-morphism $h \in \End(\AAA(c))$ 
to the $\BBBB(c)$-morphism $\widehat{\zeta}_c(h) := \zeta_c((\eta_\AAA)_c^{-1} (h))$. This functor is by construction 
$\bbK$-linear, i.e.\ $\widehat{\zeta}_c \in \mathbf{Lin}_{\bbK}(\catEnd(\AAA(c)),\BBBB(c))$. 
Since the functor \eqref{eqn:EndRestriction} is essentially surjective, 
there exists a $\bbK$-linear and co-continuous functor 
$\kappa_c \in \mathbf{Lin}_{\bbK,\mathrm{c}}(\Mod_{\AAA(c)},\BBBB(c))$ 
and a natural isomorphism $\kappa_{\ast_c}$ from the functor $\widehat{\zeta}_c$ 
to the restriction along the inclusion 
$\catEnd(\AAA(c)) \subseteq \Mod_{\AAA(c)}$ of the functor $\kappa_c$. 
Because $\AAA(c) \in \catEnd(\AAA(c))$ is the only object, 
the natural isomorphism $\kappa_{\ast_c}$ consists of a single
$\BBBB(c)$-isomorphism $\kappa_{\ast_c}: \b_c \to \kappa_c(\AAA(c))$, with
naturality being encoded in the condition $\kappa_c(h) \circ \kappa_{\ast_c} = 
\kappa_{\ast_c} \circ \widehat{\zeta}_c(h)$, for all $h \in \End(\AAA(c))$. 
Note that we have just constructed part of the data defining a $1$-morphism 
$\kappa: \iota(\AAA) \to \BBBB$ in $\2AQFT(\ovr{\CC})$ (cf.\ Remark \ref{rem:2AQFTunpacked}). 
To complete the data, we have to construct, for each $\und{f}\in\P_{\ovr{\CC}}\big(\substack{t \\ \und{c}}\big)$, 
a natural isomorphism $\kappa_{\und{f}} : \BBBB(\und{f})\, \prod_i  \kappa_{c_i} 
\Rightarrow \kappa_t\, \iota(\AAA)(\und{f})$ 
between functors from $\prod_{i=1}^n \Mod_{\AAA(c_i)}$ to $\BBBB(t)$
that are $\bbK$-linear and co-continuous in each variable.
Using again the equivalences in \eqref{eqn:DeligneKelly} and \eqref{eqn:EndRestriction},
this problem is equivalent to constructing a $\BBBB(t)$-isomorphism, denoted 
with a slight abuse of notation also by
$\kappa_{\und{f}}: \BBBB(\und{f})\left(\prod_i \kappa_{c_i} \big( \prod_i \AAA(c_i) \big)\right) 
\to  \kappa_t \left( \iota(\AAA)(\und{f}) \big( \prod_i \AAA(c_i) \big)\right)$, fulfilling the naturality condition 
$\kappa_t \left( \iota(\AAA)(\und{f}) (\und{h}) \right) \circ \kappa_{\und{f}} 
= \kappa_{\und{f}} \circ \BBBB(\und{f}) \left( \prod_i \kappa_{c_i} (\und{h}) \right)$, 
for all $\und{h} \in \prod_{i=1}^n \End(\AAA(c_i))$.
We define the $\BBBB(t)$-isomorphism $\kappa_{\und{f}}$ according to 
\begin{flalign}
\xymatrix{
\BBBB(\und{f})\left(\prod_i \kappa_{c_i} \big( \prod_i \AAA(c_i) \big)\right) \ar@{-->}[rr]^-{\kappa_{\und{f}}} ~&~ ~&~ \kappa_t \left( \iota(\AAA)(\und{f}) \big( \prod_i \AAA(c_i) \big)\right) \ar[d]_-{\cong}^-{\kappa_t\big( \iota(\AAA)^{2}_{(\und{f},\ast_{\und{c}})}\big)} \\ 
\BBBB(\und{f}) \big( \prod_i \b_{c_i} \big) \ar[r]_-{\BBBB^{2}_{(\und{f},\ast_{\und{c}})}} \ar[u]_-{\cong}^-{\BBBB(\und{f})\left(\prod_i \kappa_{\ast_{c_i}}\right)} ~&~ \b_t \ar[r]_-{\kappa_{\ast_t}} ~&~ \kappa_t(\AAA(t)) 
}
\end{flalign}
and observe that the required naturality condition for $\kappa_{\und{f}}$ 
follows from naturality of $\kappa_{\ast_c}$ and of $\zeta$. 
This provides us with the desired natural isomorphism 
$\kappa_{\und{f}} : \BBBB(\und{f})\, \prod_i \kappa_{c_i} \Rightarrow \kappa_t\, \iota(\AAA)(\und{f})$
and hence completes the data needed to define a $1$-morphism 
$\kappa: \iota(\AAA) \to \BBBB$ in $\2AQFT(\ovr{\CC})$. 
It remains to check that the relevant axioms hold (cf.\ Remark \ref{rem:2AQFTunpacked} and 
Definition \ref{def:pseudotransformation}).
Using once again the equivalences in \eqref{eqn:DeligneKelly} and \eqref{eqn:EndRestriction}, 
confirming these axioms can be reduced to checking that certain diagrams in $\BBBB(t)$ commute.
This can be done directly by using that $\iota(\AAA)$ and $\BBBB$ are objects 
in $\2AQFT(\ovr{\CC})$. (Here the axioms of Definition \ref{def:pseudomorphism} 
enter explicitly.) 
Since by construction the $\AQFT(\ovr{\CC})$-morphisms $\widetilde{\kappa} = \zeta: \AAA \to \pi(\BBBB)$ 
coincide, this shows that the functor \eqref{eqn:biadjunctionfunctor} is essentially surjective. 
\end{proof}

\begin{rem}\label{rem:counit}
The counit $\epsilon : \iota\,\pi \Rightarrow \id$ of the inclusion-truncation biadjunction 
is determined implicitly by Theorem \ref{theo:biadjunction},
see e.g.\ \cite[Definition 2.5 and Remark 2.6]{biadjointtriangles} for further details on biadjunctions. 
Its $\BBBB$-component $\epsilon_{\BBBB} : \iota\,\pi(\BBBB)\to \BBBB$ 
is a $1$-morphism in $\2AQFT(\ovr{\CC})$ which maps under the equivalence of categories 
in  \eqref{eqn:biadjunctionfunctor} to the identity $\widetilde{\epsilon_{\BBBB}} 
= \id_{\pi(\BBBB)} : \pi(\BBBB)\to \pi(\BBBB)$ in $\AQFT(\ovr{\CC})$.
Note that the latter property determines the $1$-morphism $\epsilon_{\BBBB} $ up to invertible 
$2$-morphisms in $\2AQFT(\ovr{\CC})$.
\end{rem}

The counit allows us to detect whether an object $\BBBB\in\2AQFT(\ovr{\CC})$
lies in the essential image of the inclusion pseudo-functor $\iota : \AQFT(\ovr{\CC})\to\2AQFT(\ovr{\CC})$.
\begin{defi}\label{def:truncatedobject}
We say that $\BBBB\in\2AQFT(\ovr{\CC})$ is {\em truncated} if the
corresponding component $\epsilon_{\BBBB} : \iota\,\pi(\BBBB)\to \BBBB$ 
of the counit is an equivalence in $\2AQFT(\ovr{\CC})$. 
\end{defi}

This means that a truncated
2AQFT $\BBBB$ is fully determined by its truncation $\pi(\BBBB)\in\AQFT(\ovr{\CC})$,
which is an ordinary AQFT, as it can be reconstructed (up to equivalence) by applying the inclusion pseudo-functor $\iota$.
Our goal in Section \ref{sec:gauging} is to construct examples 
of 2AQFTs that are {\em not} truncated in the above sense.

%%%%%%%%%%%%%%%%%%%%%%%%%%%%%%%%%%%%%%%%%%%%%%%%
%%%%%%%%%%%%%%%%%%%%%%%%%%%%%%%%%%%%%%%%%%%%%%%%

\section{\label{sec:gauging}Gauging construction and orbifold 2AQFTs}
We present a construction of 2AQFTs from the data
of an ordinary AQFT $\AAA\in\AQFT(\ovr{\CC})$ that is endowed 
with an action of a finite group $G$.  (A generalization 
to infinite groups or also algebraic groups is possible,  however we prefer to 
restrict ourselves to finite groups for which we can prove the main Theorem \ref{theo:gaugingtruncated}
of this section,  see also the related Theorem \ref{theo:Hopfgalois}.)
This construction can be interpreted physically
as a local gauging of $\AAA$ with respect to $G$
and the resulting 2AQFT as the corresponding 
categorified orbifold theory, see  Proposition \ref{prop:orbifoldtruncation} and
Remark \ref{rem:orbifoldinterpretation} below.
Let us start with introducing some relevant terminology.
\begin{defi}\label{def:equivariantAQFT}
Let $\ovr{\CC}$ be an orthogonal category and $G$ a finite group. 
A {\em $G$-equivariant AQFT on $\ovr{\CC}$} is a pair $(\AAA,\rho)$ 
consisting of an object $\AAA\in \AQFT(\ovr{\CC})$
and a representation $\rho : G \to \Aut(\AAA)$ of $G$ as
natural automorphisms of $\AAA$. A morphism $\zeta : (\AAA,\rho)\to (\BBB,\sigma)$
of $G$-equivariant AQFTs is an $\AQFT(\ovr{\CC})$-morphism
$\zeta : \AAA\to \BBB$ that commutes with the $G$-actions,
i.e.\ $\zeta\, \rho(g) = \sigma(g)\, \zeta$, for all $g\in G$.
We denote by $G{\text{-}}\AQFT(\ovr{\CC})$ the category of $G$-equivariant AQFTs on $\ovr{\CC}$.
\end{defi}

\begin{rem}\label{rem:equivariantAQFT}
Our choice of terminology in Definition \ref{def:equivariantAQFT} 
is motivated by the following equivalent perspective on $G$-equivariant AQFTs.
Let us denote by $\Rep_{\bbK}(G)$ the category of
$\bbK$-linear representations of $G$. Recall that this is a (closed) 
symmetric monoidal category with monoidal product given by
the tensor product $V\otimes W$ of representations,
monoidal unit given by the trivial representation $\bbK$ and symmetric braiding given by the flip map.
Hence, we can introduce the category $G{\text{-}}\Alg_{\bbK} := \Alg_{\As}(\Rep_{\bbK}(G))$
of associative and unital algebras in $\Rep_{\bbK}(G)$, which are also called
$G$-equivariant associative and unital $\bbK$-algebras. 
(Note that for the trivial group $G=\{ e\}$, this is just the category $\Alg_{\bbK}$
that we considered in Section \ref{sec:1AQFT}.)
It is then easy to check that a $G$-equivariant AQFT $(\AAA,\rho)$, 
as introduced in Definition \ref{def:equivariantAQFT},
is the same data as a functor $\CC\to G{\text{-}}\Alg_{\bbK}$ to the category
of $G$-equivariant associative and unital $\bbK$-algebras that satisfies
the $\perp$-commutativity property from Definition \ref{def:AQFT}.
From this perspective, morphisms of $G$-equivariant AQFTs are 
just natural transformations of functors from  $\CC$ to $G{\text{-}}\Alg_{\bbK}$.
\end{rem}

Given any $G$-equivariant AQFT $(\AAA,\rho)\in G{\text{-}}\AQFT(\ovr{\CC})$, one can construct 
its associated {\em orbifold theory} $\AAA_0^G\in\AQFT(\ovr{\CC})$ 
by taking the invariants of the action $\rho : G\to\Aut(\AAA)$, see e.g.\ \cite{Xu,Mueger}.
We have added a subscript $0$ to emphasize that, 
as we shall show in Proposition \ref{prop:orbifoldtruncation}, 
the traditional orbifold theory $\AAA_0^G\cong\pi(\AAA^G)$ 
is only the truncation of an in general richer {\em categorified orbifold theory}
$\AAA^G\in \2AQFT(\ovr{\CC})$. 
(We shall illustrate later in Example \ref{ex:RepGFredCat}
that these higher categorical structures are particularly important for local-to-global constructions.)
The latter will be described by a gauging construction
that we develop in this section. We also refer to Remark \ref{rem:orbifoldinterpretation} 
for a physical interpretation of our gauging construction and the resulting categorified orbifold theory.
\sk

As a preparation, let us briefly recall some standard facts and constructions from 
the theory of equivariant algebras and modules. 
As already mentioned in Remark \ref{rem:equivariantAQFT},
the representation category $\Rep_{\bbK}(G)$ of a finite group $G$
is a (closed) symmetric monoidal category, hence we can introduce the category
$G{\text{-}}\Alg_{\bbK}$ of {\em $G$-equivariant associative and unital $\bbK$-algebras}.
Analogously to the non-equivariant case $\Alg_\bbK$ from Section \ref{sec:1AQFT},
this category is symmetric monoidal with monoidal product the tensor
product algebra $A\otimes B$ (endowed with the tensor product $G$-action),
monoidal unit the algebra $\bbK$ (endowed with the trivial $G$-action) and 
symmetric braiding given by the flip map.
For every object $A\in G{\text{-}}\Alg_{\bbK}$, one can introduce
the locally presentable $\bbK$-linear category 
$G{\text{-}}\Mod_A := \Mod_A(\Rep_{\bbK}(G))$ of {\em $G$-equivariant right $A$-modules}.
An object in this category is an object
$V\in \Rep_{\bbK}(G)$ together with a $\Rep_{\bbK}(G)$-morphism
$V\otimes A\to V$ that satisfies the usual axioms of a right $A$-action.
Morphisms are $\Rep_{\bbK}(G)$-morphism that preserve the right $A$-actions.
Similarly to the non-equivariant case, given any morphism $\kappa :A\to B$ 
in $G{\text{-}}\Alg_{\bbK}$, one can define 
a $\bbK$-linear induced module functor $\kappa_! = (-)\otimes_{A}^{}B :  G{\text{-}}\Mod_A \to G{\text{-}}\Mod_B$.
This functor has a right adjoint given by the 
restriction functor $\kappa^\ast : G{\text{-}}\Mod_B \to G{\text{-}}\Mod_A$.
As a consequence, $\kappa_!$ is a co-continuous $\bbK$-linear functor between
locally presentable $\bbK$-linear categories, i.e.\ a $1$-morphism
in the $2$-category $\Pr_\bbK$.
\sk

Let now $(\AAA,\rho)\in G{\text{-}}\AQFT(\ovr{\CC})$ be any $G$-equivariant AQFT on $\ovr{\CC}$.
We define its gauging $\AAA^G\in \2AQFT(\ovr{\CC})$ by a $G$-equivariant generalization
of the inclusion pseudo-functor $\iota$ from Section \ref{subsec:inclusion}. 
Concretely,  $\AAA^G$ is described by the following data as
listed in Remark \ref{rem:2AQFTunpacked}:
\begin{enumerate}[(1)]
\item For each $c\in\CC$, we set
\begin{flalign}
\AAA^G(c)\,:=\, G{\text{-}}\Mod_{\AAA(c)}
\end{flalign}
to be the locally presentable $\bbK$-linear category of $G$-equivariant
right modules over the $G$-equivariant associative and unital $\bbK$-algebra
$\AAA(c)\in G{\text{-}}\Alg_{\bbK}$.

\item For each non-empty tuple $\und{f} = (f_1,\dots,f_n)\in\P_{\ovr{\CC}}\big(\substack{t \\ \und{c}}\big)$
of mutually orthogonal $\CC$-morphisms, we define the functor
\begin{flalign}\label{eqn:AGfactorizationproduct}
\AAA^G(\und{f})\,:\, \xymatrix@C=3em{
 \prod\limits_{i=1}^n G{\text{-}}\Mod_{\AAA(c_i)} \ar[r]^-{\otimes^n}~&~G{\text{-}}\Mod_{\Motimes_{i=1}^n \AAA(c_i)}\ar[r]^-{\AAA(\und{f})_!}~&~G{\text{-}}\Mod_{\AAA(t)}
}\quad,
\end{flalign}
where $\otimes^n: (V_1,\dots,V_n)\mapsto V_1\otimes\cdots\otimes V_n$ 
is the functor assigning the $n$-ary tensor product of representations 
(equipped with the induced structure of a $G$-equivariant module 
over the tensor product of algebras) and $\AAA(\und{f})_!$
is the induced module functor for $G$-equivariant modules along
the $G{\text{-}}\Alg_{\bbK}$-morphism $\AAA(\und{f}) : \bigotimes_{i=1}^n \AAA(c_i)\to \AAA(t)$.
Note that the functor \eqref{eqn:AGfactorizationproduct} is $\bbK$-linear
and co-continuous in each variable, i.e.\ it defines a $1$-operation in $\Pr_\bbK$.
For the empty tuple $\und{c}=\emptyset$, we set the pointing to be
$\AAA^G(\ast_t) := \AAA(t) \in G{\text{-}}\Mod_{\AAA(t)}$.

\item[(3-5)] The coherence isomorphisms for $\AAA^G$ are completely analogous to
the ones for the inclusion $\iota(\AAA)\in \2AQFT(\ovr{\CC})$ from Section \ref{subsec:inclusion}.
\end{enumerate}

Let us consider now a morphism $\zeta: (\AAA,\rho)\to(\BBB,\sigma)$ in $G{\text{-}}\AQFT(\ovr{\CC})$.
Then the following data defines a $1$-morphism $\zeta^G : \AAA^G\to \BBB^G$
in $\2AQFT(\ovr{\CC})$, see also Remark \ref{rem:2AQFTunpacked}:
\begin{enumerate}[(1)]
\item For each $c\in\CC$, we set
\begin{flalign}
(\zeta^G)_c \,:= \, (\zeta_c)_!\, : \,G{\text{-}}\Mod_{\AAA(c)}~\longrightarrow~G{\text{-}}\Mod_{\BBB(c)}
\end{flalign}
to be the $\bbK$-linear and co-continuous induced module functor for $G$-equivariant modules
along the $G{\text{-}}\Alg_{\bbK}$-morphism $\zeta_c:\AAA(c)\to\BBB(c)$.

\item The coherence isomorphisms for $\zeta^G : \AAA^G\to\BBB^G$ are completely analogous to the ones for 
$\iota(\zeta):\iota(\AAA)\to\iota(\BBB)$ from Section \ref{subsec:inclusion}.
\end{enumerate}

\begin{propo}\label{prop:gaugingconstruction}
For every orthogonal category $\ovr{\CC}$ and finite group $G$, 
the construction above defines a pseudo-functor
\begin{flalign}
(-)^G \,:\, G{\text{-}}\AQFT(\ovr{\CC})~\longrightarrow~\2AQFT(\ovr{\CC})\quad,
\end{flalign}
which we call the gauging construction.
\end{propo}

The following result relates our 
gauging construction to the traditional concept of orbi\-fold theories from \cite{Xu,Mueger}. 
\begin{propo}\label{prop:orbifoldtruncation}
For every $G$-equivariant AQFT $(\AAA,\rho)\in G{\text{-}}\AQFT(\ovr{\CC})$,
there exists a natural isomorphism $\AAA_0^G \cong \pi(\AAA^G)$
in $\AQFT(\ovr{\CC})$ between the traditional orbifold theory $\AAA_0^G$
(that assigns subalgebras of $G$-invariants)
and the truncation (cf.\ Section \ref{subsec:truncation}) of
the gauging construction $\AAA^G\in\2AQFT(\ovr{\CC})$ from Proposition \ref{prop:gaugingconstruction}.
\end{propo}
\begin{proof}
From the description of the truncation $2$-functor in Section \ref{subsec:truncation},
we obtain that $\pi(\AAA^G)(c)\,=\,\End(\AAA(c))$ is the endomorphism algebra of the 
$G$-equivariant module $\AAA(c)\in G{\text{-}}\Mod_{\AAA(c)}$, for each $c\in\CC$.
Since morphisms in $G{\text{-}}\Mod_{\AAA(c)}$ preserve by definition the $G$-action,
it follows that $\End(\AAA(c))$ is isomorphic to the subalgebra of $G$-invariants in $\AAA(c)$,
hence $\pi(\AAA^G)(c)\cong \AAA^G_0(c)$ is isomorphic to the algebra that is 
assigned by the traditional orbifold theory $\AAA^G_0$. Using further
the explicit description of the factorization products of $\pi(\AAA^G)$ 
from Section \ref{subsec:truncation}, one checks that this family of 
$\Alg_\bbK$-isomorphisms defines an $\AQFT(\ovr{\CC})$-isomorphism
$\pi(\AAA^G)\cong \AAA^G_0$. Naturality of this isomorphism 
in $(\AAA,\rho)\in G{\text{-}}\AQFT(\ovr{\CC})$ is obvious.
\end{proof}

The previous proposition provides a justification for the following
terminology.
\begin{defi}\label{def:catorbifold}
We call the gauging construction $\AAA^G\in\2AQFT(\ovr{\CC})$ 
the {\em categorified orbifold theory} associated to the $G$-equivariant AQFT 
$(\AAA,\rho)\in G{\text{-}}\AQFT(\ovr{\CC})$.
\end{defi}

\begin{rem}\label{rem:orbifoldinterpretation} 
In addition to our result in Proposition \ref{prop:orbifoldtruncation},
there is further justification for calling the 2AQFT $\AAA^G$ a 
categorified orbifold theory. The presentation in this remark is
intentionally rather informal, which is convenient
to convey our main message.
\sk

Let us recall that the field configurations 
of a classical $\sigma$-model are given by maps $\phi : \Sigma \to X$ 
from a world-sheet $\Sigma$ to a target space $X$. When a finite 
group $G$ acts on the target space $X$, one can form the orbifold
(i.e.\ quotient stack) $X//G$ and consider the corresponding
{\em orbifold $\sigma$-model} whose field configurations
are now maps $\phi : \Sigma\to X//G$ with values in a stack. As a consequence,
the ``space'' of field configurations is a stack too, namely the mapping stack
\begin{flalign}\label{eqn:fieldstack}
\mathrm{Fields}(\Sigma)\,:=\, \mathrm{Map}\big(\Sigma,X//G \big)\quad.
\end{flalign}
In order to study local aspects of this field theory,
let us introduce the orthogonal category $\ovr{\Disk(\Sigma)}$, whose underlying category $\mathbf{Disk}(\Sigma)\subseteq \Open(\Sigma)$ 
consists of open subsets $U\subseteq \Sigma$ that are diffeomorphic to a Cartesian
space $U\cong \bbR^m$ and whose orthogonal category structure is the restriction to $\mathbf{Disk}(\Sigma)$ of
the orthogonal category structure from Example \ref{ex:opens} on $\ovr{\Open(\Sigma)}$. 
Considering such subsets $U\subseteq \Sigma$,
the stack of fields is equivalent to the stacky quotient
\begin{flalign}\label{eqn:fieldstackyquotient}
\mathrm{Fields}(U) \, \simeq \, \mathrm{Map}(U,X)//G
\end{flalign}
of the {\em ordinary} mapping space $\mathrm{Map}(U,X)$ by the finite group $G$.
If we ignore for the moment the stacky quotient by $G$,
we are in the familiar scenario where the space of fields
$\mathrm{Map}(U,X)$ is just an {\em ordinary} space and not a stack.
Formal deformation quantization of such a field theory
leads to an {\em ordinary} AQFT  $\AAA\in \AQFT(\ovr{\mathbf{Disk}(\Sigma)})$, 
which in the case there are no anomalies will carry a $G$-action, i.e.\ 
$(\AAA,\rho)\in  G{\text{-}}\AQFT(\ovr{\mathbf{Disk}(\Sigma)})$.
By construction, $\AAA(U)$ is a deformation quantization of
a suitable $G$-equivariant function algebra $\O(\mathrm{Map}(U,X))$.
\sk

Things get more interesting when we consider the stacky 
quotient by $G$ in \eqref{eqn:fieldstackyquotient}.
From the perspective of \cite{Brandenburg,LurieTannaka}, which we recalled in Section \ref{sec:intro},
it is better to assign to the stack $\mathrm{Fields}(U)$ in \eqref{eqn:fieldstackyquotient} 
its category of quasi-coherent sheaves 
\begin{flalign}
\QCoh\big(\mathrm{Fields}(U)\big)\,\simeq\, \QCoh\big(\mathrm{Map}(U,X)//G\big) 
\,\simeq \, G{\text{-}}\Mod_{\O( \mathrm{Map}(U,X))}\quad,
\end{flalign}
which is the symmetric monoidal category of $G$-equivariant modules over the classical $G$-equivariant
function algebra $\O( \mathrm{Map}(U,X))$.
The $G$-equivariant quantization $\AAA(U)$ of $\O( \mathrm{Map}(U,X))$ from the previous paragraph 
then allows us to introduce the pointed category $\AAA^G(U) = G{\text{-}}\Mod_{\AAA(U)}$, which we
interpret as a pointed category quantizing the symmetric monoidal category $\QCoh(\mathrm{Fields}(U))$ (cf.\ Section \ref{sec:intro}) and recognize as our gauging construction.  
Hence, our gauging construction encodes local aspects of orbifold $\sigma$-models.
It is well-known that orbifold theories exhibit rich global phenomena,  such as the so-called twisted sectors
and anomalies, see e.g.\ \cite{JohnsonFreydOrbifold} for an excellent review.  Simple examples of some of these phenomena will be
discussed later in Section \ref{sec:localtoglobal}.
\end{rem}

We still have to address the important question whether or not it is possible
to obtain genuine non-truncated $\AAA^G\in \2AQFT(\ovr{\CC})$
from our gauging construction. This will of course depend on details
of the group $G$ and its action $\rho : G\to\Aut(\AAA)$ on $\AAA\in\AQFT(\ovr{\CC})$.
For example, if $G=\{e\}$ is the trivial group, then the gauging construction from Proposition
\ref{prop:gaugingconstruction} coincides with the inclusion pseudo-functor $\iota$
from Proposition \ref{prop:inclusionpseudofunctor}, hence gauging the trivial
group $G=\{e\}$ always leads to truncated 2AQFTs in the sense of Definition \ref{def:truncatedobject}. 
On the other hand, it is very easy to give simple examples of gaugings that define
non-truncated 2AQFTs.
\begin{ex}\label{ex:trivialAQFTgauging}
Let $G$ be a finite group and consider the trivial AQFT $\AAA=\bbK\in\AQFT(\ovr{\CC})$
that assigns $\AAA(c)=\bbK\in\Alg_\bbK$ to every $c\in\CC$. When endowed with the trivial $G$-action 
$\rho : G\to\Aut(\bbK)\,,~g\mapsto \id_{\bbK}$, this defines a $G$-equivariant AQFT
$(\bbK,\rho)\in G{\text{-}}\AQFT(\ovr{\CC})$. Note that the corresponding 
traditional orbifold theory $\bbK^G_0 =\bbK$ is of course the trivial AQFT too.
In contrast, the categorified orbifold theory $\bbK^G \in \2AQFT(\ovr{\CC})$
that is obtained from our gauging construction is much more interesting.
It assigns to every $c\in\CC$ the representation category
$\bbK^G(c) = G{\text{-}}\Mod_{\bbK} = \Rep_{\bbK}(G)$ of the group $G$
and its factorization products $\bbK^G(\und{f})= \otimes^n : \prod_{i=1}^n\Rep_{\bbK}(G) \to\Rep_{\bbK}(G)$
are given by the $n$-ary tensor products of representations. The pointing
$\bbK^G(\ast_t)=\bbK\in\Rep_{\bbK}(G)$ is given by the trivial representation.
By Proposition \ref{prop:orbifoldtruncation}, the truncation $\pi(\bbK^G)\cong \bbK^G_0 = \bbK$
of this 2AQFT is the trivial theory
\sk

Our claim is that the categorified orbifold theory $\bbK^G \in \2AQFT(\ovr{\CC})$ is not truncated, 
whenever $G\neq \{e\}$ is non-trivial. To prove this claim, we consider as explained
in Remark \ref{rem:counit} the corresponding component
$\epsilon_{\bbK^G} : \iota\,\pi(\bbK^G) \to \bbK^G$ of the counit of the inclusion-truncation biadjunction.
This is a $1$-morphism in $\2AQFT(\ovr{\CC})$ whose components
$(\epsilon_{\bbK^G})_c : \iota\,\pi(\bbK^G)(c)\simeq \Vec_\bbK \to \bbK^G(c)=\Rep_\bbK(G)$
are given by co-continuous $\bbK$-linear functors from the category of vector spaces
to the representation category of $G$. Because $1$-morphisms in $\2AQFT(\ovr{\CC})$
preserve the pointings (up to coherence isomorphisms), we know that the 
$1$-dimensional vector space $\bbK\in\Vec_\bbK$ is mapped to a {\em trivial}
representation $(\epsilon_{\bbK^G})_c (\bbK) \cong \bbK \in \Rep_{\bbK}(G)$.
Using further that every vector space $V\cong \bigoplus_{b\in B} \bbK \in\Vec_\bbK$ 
is isomorphic to a coproduct of the $1$-dimensional vector space $\bbK$ (by choosing
a basis $B$) and co-continuity of the functor $\epsilon_{\bbK^G}$, we observe that the 
essential image of $(\epsilon_{\bbK^G})_c : \Vec_\bbK \to \Rep_\bbK(G)$
lies in the full subcategory of trivial $G$-representations, hence it cannot be an equivalence
of categories as every finite group $G\neq \{e\}$ admits 
non-trivial $\bbK$-linear representations. As a consequence,  
the component $\epsilon_{\bbK^G} : \iota\,\pi(\bbK^G) \to \bbK^G$
of the counit is not an equivalence in $\2AQFT(\ovr{\CC})$ and hence
the categorified orbifold theory $\bbK^G\in\2AQFT(\ovr{\CC})$ is not truncated.
\end{ex}

Quite remarkably, it is possible to characterize precisely those $G$-equivariant
AQFTs $(\AAA,\rho)\in  G{\text{-}}\AQFT(\ovr{\CC})$ whose associated
gauging construction $\AAA^G\in \2AQFT(\ovr{\CC})$ is truncated.
Our arguments below make use of some standard terminology and results
from Hopf-Galois theory, see e.g.\ \cite{Doi} and also the review article 
\cite{Montgomery}. Let $H$ be a Hopf algebra over $\bbK$. (In our applications below,
$H = \O(G) = \mathrm{Map}(G,\bbK)$ is the function Hopf algebra of a finite group $G$.)
A {\em right $H$-comodule algebra} is an algebra $A\in\Alg_\bbK$
endowed with a right $H$-coaction $\delta : A\to A\otimes H$ 
that is an $\Alg_\bbK$-morphism. We denote by $B :=A^{\mathrm{co} H} := 
\{a\in A\,:\, \delta(a) = a\otimes 1_H\}\subseteq A$ the subalgebra of $H$-coaction invariants.
\begin{defi}\label{def:HopfGalois}
The algebra extension $B= A^{\mathrm{co} H}\subseteq A$ is called {\em $H$-Hopf-Galois}
if the canonical map
\begin{flalign}\label{eqn:canonicalmap}
\beta\,:\, A\otimes_B^{} A~\longrightarrow~A\otimes H~,~~a\otimes_B^{} a^\prime \longmapsto (a\otimes 1_H)\, \delta(a^\prime)
\end{flalign}
is bijective. 
\end{defi}
Associated to any right $H$-comodule algebra $A$ are two 
$\bbK$-linear categories of interest: First, we have the 
category $\Mod_A^H$ of {\em right $(H,A)$-Hopf modules}.
An object in this category is a right $A$-module $V\in \Mod_A$ that is endowed
with a compatible right $H$-comodule structure $\delta^V : V\to V\otimes H$,
i.e.\ $\delta^V(v\,a) = \delta^V(v)\, \delta(a)$, for all $v\in V$ and $a\in A$. The morphisms
in $\Mod_A^H$ are $\bbK$-linear maps that preserve both the $A$-actions and the $H$-coactions.
Second, we have the $\bbK$-linear category $\Mod_B$ of right modules over 
the subalgebra $B= A^{\mathrm{co} H}\subseteq A$
of $H$-coaction invariants.  These two categories are
related by a $\bbK$-linear adjunction
\begin{flalign}\label{eqn:Hopfadjunction}
\xymatrix{
\Phi  \,:\, \Mod_B ~\ar@<0.5ex>[r]&\ar@<0.5ex>[l]  ~\Mod_A^H\,:\, \Psi
}\quad,
\end{flalign}
whose left adjoint $\Phi = (-)\otimes_B^{} A$ is the induced module functor,
where $W\otimes_B^{} A$ is endowed with the right $H$-coaction $\id\otimes_B^{} \delta$, for all $W\in\Mod_B$, 
and whose right adjoint $\Psi = (-)^{\mathrm{co}H}$  is the functor
taking $H$-coaction invariants $V^{\mathrm{co} H} := \{v\in V\,:\, \delta^V(v) = v\otimes 1_H\}$,
for all $V\in\Mod_{A}^{H}$.
We shall need the following result (cf.\ \cite[Theorem 5.6]{Montgomery}), which is originally
due to Doi and Takeuchi \cite{Doi}.
\begin{theo}\label{theo:Hopfgalois}
Let $H$ be finite dimensional. Then $B= A^{\mathrm{co} H}\subseteq A$
is $H$-Hopf-Galois if and only if the counit $\epsilon : \Phi \, \Psi\Rightarrow \id$
of the adjunction \eqref{eqn:Hopfadjunction} is a natural isomorphism.
\end{theo}

In the case of interest to us, the Hopf algebra $H = \O(G)=\mathrm{Map}(G,\bbK)$ 
is given by the function Hopf algebra of a finite group $G$.
In particular, $H$ is finite dimensional. One easily observes that a
right $\O(G)$-coaction $\delta^V : V\to V\otimes \O(G)$
is the same datum as a group action $\rho : G\to \Aut(V)$ by the relationship
$\rho(g)(v) = v_{(0)}^{} \, \langle v_{(1)}^{},g\rangle$,
where we used Sweedler notation $\delta^V(v) = v_{(0)}^{}\otimes v_{(1)}^{}$
and the duality pairing $\langle\cdot,\cdot\rangle : \O(G)\otimes \bbK[G]\to\bbK$.
In particular, right $\O(G)$-comodule algebras are precisely $G$-equivariant 
algebras $A\in G{\text{-}}\Alg_\bbK$ and the $\bbK$-linear
category $\Mod_A^{\O(G)}$ of right $(\O(G),A)$-Hopf modules is
the locally presentable $\bbK$-linear category $G{\text{-}}\Mod_A$ of
$G$-equivariant right $A$-modules. Moreover,
the subalgebra $B=A^{\mathrm{co}H} = A^G_0\subseteq A$ of $\O(G)$-coaction 
invariants is precisely the subalgebra of $G$-invariants.
\begin{cor}\label{cor:HopfGalois}
Let $G$ be a finite group and $H = \O(G)$ the function Hopf algebra of $G$.
In this case, the adjunction in \eqref{eqn:Hopfadjunction} reads as
\begin{flalign}\label{eqn:Hopfadjunction2}
\xymatrix{
\Phi \,:\, \Mod_B ~\ar@<0.5ex>[r]&\ar@<0.5ex>[l]  ~G{\text{-}}\Mod_A\,:\,\Psi
}\quad.
\end{flalign}
This is an (adjoint) equivalence in the $2$-category
$\Pr_\bbK$ of locally presentable $\bbK$-linear categories
if and only if $B=A^{\mathrm{co}H} = A^G_0\subseteq A$ is $\O(G)$-Hopf-Galois.
\end{cor}
\begin{proof}
The left adjoint functor  $\Phi = (-)\otimes_B^{} A$ is clearly $\bbK$-linear  and co-continuous,
i.e.\ a $1$-morphism in $\Pr_\bbK$. The right adjoint functor $\Psi = (-)^{\mathrm{co}H} = (-)^G_0$
assigns the $G$-invariants (given by a categorical limit), which for actions of
finite groups $G$ and $\mathrm{char}(\bbK) =0$ coincides with the
$G$-{\em co}invariants (i.e.\ a categorical colimit). Hence, the right adjoint $\Psi$
is a $\bbK$-linear and co-continuous functor too
and the adjunction \eqref{eqn:Hopfadjunction2} is in the $2$-category $\Pr_\bbK$.
\sk

The unit $\eta : \id \Rightarrow \Psi\,\Phi$ 
of the adjunction \eqref{eqn:Hopfadjunction2} is given
by the components $\eta_W : W\to (W\otimes_B^{} A)^{G}_0\,,~w\mapsto w\otimes_B^{} 1_A$,
for all $W\in \Mod_B$. Using again that forming $G$-invariants coincides with forming
$G$-{\em co}invariants, we find that  $\eta : \id \Rightarrow \Psi\,\Phi$  is a natural isomorphism.
Our claim then follows from Theorem \ref{theo:Hopfgalois}.
\end{proof}

\begin{theo}\label{theo:gaugingtruncated}
Let $G$ be a finite group and $(\AAA,\rho)\in G{\text{-}}\AQFT(\ovr{\CC})$
a $G$-equivariant AQFT. Then the categorified orbifold theory $\AAA^G\in\2AQFT(\ovr{\CC})$
is truncated if and only if the algebra extension $\AAA^G_0(c) \subseteq \AAA(c)$ 
is $\O(G)$-Hopf-Galois, for all $c\in\CC$.
\end{theo}
\begin{proof}
Recalling Definition \ref{def:truncatedobject}, 
the 2AQFT $\AAA^G\in\2AQFT(\ovr{\CC})$ is by definition truncated 
if the corresponding component $\epsilon_{\AAA^G} : \iota\,\pi(\AAA^G)\to \AAA^G$
of the counit of the inclusion-truncation biadjunction from Theorem \ref{theo:biadjunction} 
is an equivalence in $\2AQFT(\ovr{\CC})$. The component $\epsilon_{\AAA^G}$ of the counit is determined
uniquely (up to invertible $2$-morphisms in $\2AQFT(\ovr{\CC})$) by the condition
$\widetilde{\epsilon_{\AAA^G}} = \id_{\pi(\AAA^G)} : \pi(\AAA^G)\to \pi(\AAA^G)$
on its adjunct under \eqref{eqn:biadjunctionfunctor}.
Using the explicit description of the inclusion and truncation pseudo-functors
from Section \ref{sec:trunc} and the one of the gauging construction
from the present section, one observes that the induced module functors
$\Phi_c = (-)\otimes_{\AAA^G_0(c)} \AAA(c) : \iota\,\pi(\AAA^G)(c)\cong 
\Mod_{\AAA^G_0(c)}\to \AAA^G(c) = G{\text{-}}\Mod_{\AAA(c)} $ 
(together with the obvious coherence isomorphisms) define a $1$-morphism 
$\Phi : \iota\,\pi(\AAA^G)\to \AAA^G$
in $\2AQFT(\ovr{\CC})$ that satisfies $\widetilde{\Phi} = \id_{\pi(\AAA^G)} : \pi(\AAA^G)\to \pi(\AAA^G)$.
Hence, $\Phi \cong \epsilon_{\AAA^G}$ and we can equivalently investigate if $\Phi$ 
is an equivalence in $\2AQFT(\ovr{\CC})$.
\sk

By a straightforward but slightly lengthy calculation, one proves that a $1$-morphism in 
$\2AQFT(\ovr{\CC})$ is an equivalence if and only if all its components are 
equivalences in the $2$-category $\Pr_\bbK$.
(In this proof one uses that every equivalence in any $2$-category (here $\Pr_\bbK$) 
can be upgraded to an adjoint equivalence in order to define the quasi-inverse 
$1$-morphism in $\2AQFT(\ovr{\CC})$.) Thus, to prove that $\AAA^G\in\2AQFT(\ovr{\CC})$ 
is truncated we can equivalently study the components 
$\Phi_c = (-)\otimes_{\AAA^G_0(c)} \AAA(c) :
\Mod_{\AAA^G_0(c)}\to G{\text{-}}\Mod_{\AAA(c)} $, for all $c\in \CC$.
By Corollary \ref{cor:HopfGalois}, these components are equivalences in $\Pr_\bbK$
if and only if the algebra extension $\AAA^G_0(c) \subseteq \AAA(c)$ 
is $\O(G)$-Hopf-Galois, for all $c\in\CC$. This completes the proof.
\end{proof}

\begin{rem}\label{rem:HopfGalois}
We would like to emphasize that our result in Theorem \ref{theo:gaugingtruncated}
matches perfectly our physical interpretation of the gauging construction 
$\AAA^G\in\2AQFT(\ovr{\CC})$ in terms of orbifold $\sigma$-models from Remark \ref{rem:orbifoldinterpretation}. The $H=\O(G)$-Hopf-Galois
condition from Definition \ref{def:HopfGalois} should be interpreted
as a non-commutative algebraic generalization of a free $G$-action on a space,
see e.g.\ \cite[Examples 2.11 and 2.12]{Montgomery}. Because the
quotient stack $X//G\simeq X/G$ corresponding to a free $G$-action
is equivalent to the ordinary quotient space, the resulting ``orbifold'' $\sigma$-model in this case
is just an ordinary $\sigma$-model with target space $X/G$. In particular, for free $G$-actions
one does not expect higher categorical features in the corresponding ``orbifold'' $\sigma$-model.
This is precisely what we have proven in Theorem \ref{theo:gaugingtruncated} 
for orbifold quantum field theories.
\end{rem}

We conclude this section by presenting more examples
of  non-truncated and also truncated categorified orbifold theories
$\AAA^G\in\2AQFT(\ovr{\CC})$. 
\begin{ex}\label{ex:CCR}
Let us denote by $\Disk(\bbS^1)\subset \Open(\bbS^1)$ the full subcategory of all
non-empty open intervals $I\subset \bbS^1$ in the circle $\bbS^1$. Restricting 
the orthogonality relation $\perp_{\bbS^1}$ from Example \ref{ex:opens}, we obtain
a full orthogonal subcategory $\ovr{\Disk(\bbS^1)} \subset \ovr{ \Open(\bbS^1)}$.
Objects $\AAA\in \AQFT(\ovr{\Disk(\bbS^1)})$ are interpreted as chiral conformal AQFTs
\cite{Kawahigashi}. In this example we set $\bbK=\bbC$ to be the field of complex numbers.
Let us consider the following specific theory, which is called the {\em chiral free boson}. 
To each interval
$I\subset \bbS^1$, we assign the canonical commutation relations (CCR) algebra
\begin{flalign}\label{eqn:scalarintervalalgebra}
\AAA(I) \,:=\, T_\bbC^\otimes C^\infty_\cc(I)\Big/ \Big\langle\varphi_1\otimes \varphi_2- \varphi_2\otimes\varphi_1 - i \hbar \,\int_I \varphi_1\,\dd\varphi_2~1\Big\rangle\, \in\, \Alg_\bbC \quad,
\end{flalign}
where $\hbar\in\bbR$ is the deformation parameter,  i.e.\ Planck's constant
(treated here as a number and not as a formal parameter),  
$C^\infty_\cc(I)$ denotes the vector space of compactly supported real-valued functions
on $I\subset\bbS^1$ and $T_\bbC^\otimes C^\infty_\cc(I) := \bigoplus_{n=0}^\infty 
(C^\infty_\cc(I)\otimes_\bbR^{} \bbC)^{\otimes n} \in\Alg_\bbC$ is the complexified free algebra.
To each interval inclusion $\iota_I^J : I\to J$, we assign the $\Alg_\bbK$-morphism
$\AAA(\iota_I^J) : \AAA(I)\to\AAA(J)$ that is defined on the generators
by pushforward (i.e.\ extension by zero) of compactly supported functions.
This defines an AQFT $\AAA\in \AQFT(\ovr{\Disk(\bbS^1)})$ in the sense of Definition \ref{def:AQFT}.
\sk

Let us consider the representation $\rho : G=\bbZ_2\to \Aut(\AAA)$
of the cyclic group of order $2$ that is defined  on the generators of $\AAA(I)$
by multiplication with $\pm 1$, i.e.\ $\rho(\pm 1) (\varphi) = \pm \varphi$,
for all $\varphi\in \AAA(I)$. This defines a $\bbZ_2$-equivariant AQFT $(\AAA,\rho)$
and we can form the corresponding categorified orbifold theory
$\AAA^{\bbZ_2}\in\2AQFT(\ovr{\Disk(\bbS^1)})$ from Definition \ref{def:catorbifold}.
To find out whether this theory is truncated or not, we use our results
from Theorem \ref{theo:gaugingtruncated}. Let us consider
an arbitrary interval $I\subset \bbS^1$ and set $A:=\AAA(I)$.
Observe that the subalgebra $B:=A^{\bbZ_2}_0\subset A$ of $\bbZ_2$-invariants is the 
even part of the algebra \eqref{eqn:scalarintervalalgebra}. Regarding $A=\AAA(I)$
as a $B$-bimodule, we obtain a direct sum decomposition $A=B\oplus V$,
where $V$ is the odd part of \eqref{eqn:scalarintervalalgebra}. Hence,
the source of the canonical map  \eqref{eqn:canonicalmap} is isomorphic
to $A\otimes_B^{} A \cong B\oplus (V\otimes_B^{} V)\oplus V \oplus V$.
Using further that $A \otimes \O(G) \cong \prod_{g\in G} A$, 
the canonical map \eqref{eqn:canonicalmap} explicitly reads as
\begin{flalign}
\nn \beta\,:\,  B\oplus (V\otimes_B^{} V)\oplus V \oplus V&~\longrightarrow~
\prod_{g\in\bbZ_2} A\quad,\\
b+v\otimes_B^{} v^\prime + v_1 + v_2 &~\longmapsto~\begin{pmatrix}
b + v\,v^\prime + v_1 + v_2\\
b - v\,v^\prime + v_1 - v_2
\end{pmatrix}\quad.\label{eqn:betaexplicit}
\end{flalign}
Note that the canonical map $\beta$ is bijective if and only if the map
$\mu : V\otimes_B^{} V\to B\,,~v\otimes_B^{} v^\prime\mapsto v\,v^\prime$ 
that is induced by the multiplication $\mu: A\otimes A\to A$ on 
\eqref{eqn:scalarintervalalgebra} is bijective. 
\sk

Let us consider first the case where the deformation parameter
$\hbar = 0$ is zero, which describes a classical (i.e.\ not quantized) field theory. 
In this case \eqref{eqn:scalarintervalalgebra} is a complexified symmetric algebra
over $C^\infty_\cc(I)$ and the map $\mu : V\otimes_B^{} V\to B$
is {\em not} surjective because its image is at least quadratic in the generators.
This implies that the canonical map $\beta$ in \eqref{eqn:betaexplicit} is not bijective, hence
by Theorem \ref{theo:gaugingtruncated} the categorified orbifold theory 
$\AAA^{\bbZ_2}\in \2AQFT(\ovr{\Disk(\bbS^1)})$ for $\hbar = 0$ is non-truncated.
\sk

The situation changes drastically in the quantum case where  $ 0\neq \hbar\in\bbR$ is a {\em non-formal} 
parameter.  From the canonical 
commutation relations in \eqref{eqn:scalarintervalalgebra}, we deduce that
one can always find two generators $\varphi_1,\varphi_2 \in C^\infty_\cc(I)\subseteq V\subseteq A$ 
that satisfy $[\varphi_1,\varphi_2]=i\hbar\, 1$. Dividing by
$\hbar$, which is possible because we assumed that $\bbR\ni \hbar \neq 0$, 
we can now prove that the map $\mu : V\otimes_B^{} V\to B$ is bijective.
For surjectivity, consider an arbitrary $b\in B$ and observe that
\begin{flalign}
\mu\Big(b\, \frac{1}{i\hbar} (\varphi_1\otimes_B^{} \varphi_2 - \varphi_2\otimes_B^{}\varphi_1)\Big) = b\,\frac{1}{i\hbar}[\varphi_1,\varphi_2]=b\quad.
\end{flalign} 
For injectivity, consider $\sum_j v_j\otimes_B^{} v_j^\prime\in V\otimes_B^{} V$ 
such that $\sum_j v_j\,v^\prime_j =0$ and observe that
\begin{flalign}
\nn \sum_j v_j\otimes_B^{} v_j^\prime &= \frac{1}{i\hbar} \sum_j [\varphi_1,\varphi_2] \,v_j\otimes_B^{} v_j^\prime \\
&=\frac{1}{i\hbar} \,\varphi_1 \otimes_B^{} \varphi_2\sum_{j}v_j\,v_j^\prime - \frac{1}{i\hbar} \,\varphi_2 \otimes_B^{} \varphi_1\sum_{j}v_j\,v_j^\prime  =0\quad,
\end{flalign}
where we also used that $\varphi_1\, v_j\in B$ and $\varphi_2\, v_j\in B$.
Theorem \ref{theo:gaugingtruncated} then implies that the categorified orbifold 
theory $\AAA^{\bbZ_2}\in \2AQFT(\ovr{\Disk(\bbS^1)})$ for $\hbar \neq 0$ is truncated.
\sk

Summing up, we have seen an example of a {\em non-truncated} classical orbifold field theory
that is quantized to a {\em truncated} orbifold quantum field theory. 
We would like to emphasize that this result crucially relies on inverting the deformation 
parameter $0\neq \hbar\in \bbR$ and hence it does not arise in formal deformation
quantization. (In fact, treating $\hbar$ in  \eqref{eqn:scalarintervalalgebra} as a formal parameter,
the categorified orbifold theory $\AAA^{\bbZ_2}\in \2AQFT(\ovr{\Disk(\bbS^1)})$
is non-truncated as in the classical case.)
A similar interplay between quantization and orbifold singularities
was observed before within a different framework \cite{pillow1,pillow2}. 
\end{ex}

\begin{rem}
We would like to emphasize that the results of Example \ref{ex:CCR} hold true in much greater generality.
Let $\ovr{\CC}$ be any orthogonal category and $\AAA\in \AQFT(\ovr{\CC})$ any AQFT that assigns, to every $c\in\CC$,
a CCR-algebra $\AAA(c) = \CCR(\mathcal{L}(c),\sigma_c)$ of a symplectic vector space $(\mathcal{L}(c),\sigma_c)$,
i.e.\ $\sigma_c$ is non-degenerate.
Using similar arguments as in Example \ref{ex:CCR}, one shows that the categorified
orbifold theory corresponding to the $\bbZ_2$-action
$\rho(\pm 1) : (\mathcal{L}(c),\sigma_c)\to (\mathcal{L}(c),\sigma_c)\, ,~\varphi \mapsto \pm\varphi$ is truncated,
provided that $\bbR\ni \hbar \neq 0$. The same holds true for AQFTs assigning canonical 
anticommutation relation (CAR) algebras of non-degenerate inner product spaces.
\end{rem}

%%%%%%%%%%%%%%%%%%%%%%%%%%%%%%%%%%%%%%%%%%%%%%%%
%%%%%%%%%%%%%%%%%%%%%%%%%%%%%%%%%%%%%%%%%%%%%%%%

\section{\label{sec:localtoglobal}Fredenhagen's universal category}
The goal of this section is to present a categorified 
version of Fredenhagen's universal algebra,  which plays the role of 
a local-to-global construction in AQFT that is 
analogous to factorization homology in topological QFTs \cite{AyalaFrancis,BZBJ,BZBJ2}.
Let us briefly recall the original $1$-categorical construction
for ordinary AQFTs from \cite{Fre1,Fre2,Fre3}, 
see also \cite{Lang,BSWoperad} for more details.
Given a full orthogonal subcategory embedding $J: \ovr{\CC} \to \ovr{\DD}$ 
and any ordinary AQFT $\AAA \in \AQFT(\ovr{\CC})$ on $\ovr{\CC}$, 
operadic left Kan extension along the induced operad morphism $J : \P_{\ovr{\CC}}\to\P_{\ovr{\DD}}$ determines a canonical extension
$J_!(\AAA) \in \AQFT(\ovr{\DD})$ of $\AAA$ to the larger orthogonal category
$\ovr{\DD}$. The algebra $J_!(\AAA)(d) \in \Alg_\bbK$ that is assigned by the extended AQFT $J_!(\AAA)$ 
to an object $d \in \DD$ is usually referred to as Fredenhagen's universal algebra.
This extension prescription is canonical in the sense that it is part of an adjunction 
\begin{flalign}\label{eqn:Fredenhagenadjunction}
\xymatrix{
J_! \,:\, \AQFT(\ovr{\CC}) ~\ar@<0.5ex>[r]&\ar@<0.5ex>[l]  ~\AQFT(\ovr{\DD}) \,:\,J^\ast
}\quad,
\end{flalign} 
where the right adjoint functor $J^\ast$ is given by restriction of AQFTs along $J$. 
The following two examples of full 
orthogonal subcategory embeddings $J: \ovr{\CC} \to \ovr{\DD}$
are typically considered in applications to physics.
\begin{ex}\label{ex:diskextension}
Recall from Example \ref{ex:opens} 
the orthogonal category $\ovr{\Open(M)} = (\Open(M),{\perp}_M)$ 
of non-empty open subsets of a manifold $M$. 
Consider the full subcategory $\Disk(M) \subseteq \Open(M)$ 
of all disks in $M$, i.e.\ all open subsets $U \subseteq M$ such that 
$U\cong \bbR^m$ is diffeomorphic to a Cartesian space,
and endow it with the restricted orthogonality relation.
This defines an orthogonal category $\ovr{\Disk(M)}$
together with a full orthogonal subcategory embedding
$J: \ovr{\Disk(M)}\to \ovr{\Open(M)} $. For the case of the 
circle $M=\bbS^1$, the corresponding extension functor
$J_! :  \AQFT(\overline{\Disk(\bbS^1)})\to\AQFT(\overline{\Open(\bbS^1)})$ 
is studied in the context of chiral conformal AQFT \cite{Fre1,Fre2,Fre3}. 
\end{ex}

\begin{ex}\label{ex:locextension}
Recall from Example \ref{ex:Loc} the orthogonal category $\ovr{\Loc} = (\Loc,{\perp}_{\Loc})$ 
of oriented and time-oriented globally hyperbolic Lorentzian manifolds.
Consider the full subcategory $\Locc \subseteq \Loc$ 
of all objects $M\in\Loc$ whose underlying manifold is diffeomorphic
to a Cartesian space and endow it with the restricted orthogonality relation.
This defines an orthogonal category $\ovr{\Locc}$
together with a full orthogonal subcategory embedding
$J : \ovr{\Locc}\to \ovr{\Loc} $. The corresponding extension functor
$J_! :  \AQFT(\overline{\Locc})\to\AQFT(\overline{\Loc})$ 
is studied in the context of locally covariant AQFT \cite{Lang,BSWoperad}.
\end{ex}

We will study a generalization of this extension construction to 2AQFTs, 
which is based on the biadjunction
\begin{flalign}\label{eqn:J!J*biadjunction}
\xymatrix{
J_! \,:\, \2AQFT(\ovr{\CC}) ~\ar@<0.5ex>[r]&\ar@<0.5ex>[l]  ~\2AQFT(\ovr{\DD}) \,:\,J^\ast
}\quad,
\end{flalign}
where the right adjoint $2$-functor $J^\ast$ is given by restriction of 2AQFTs along $J$.
Hence, the left adjoint pseudo-functor $J_!$ is a $2$-categorical generalization of operadic left Kan extension.
(With a mild abuse of notation, we denote both the $1$-categorical adjunction \eqref{eqn:Fredenhagenadjunction}
and the $2$-categorical adjunction \eqref{eqn:J!J*biadjunction} by the same symbols.
Below it will be clear from the context,  and from our different notations for ordinary AQFTs and 2AQFTs,  
whether $J_!\dashv J^\ast$ refers to the $1$- or $2$-categorical adjunction.)
Given any $\AAAA \in \2AQFT(\ovr{\CC})$ on $\ovr{\CC}$, 
this determines a canonical extension $J_!(\AAAA) \in \2AQFT(\ovr{\DD})$ 
to the larger orthogonal category $\ovr{\DD}$. Following AQFT terminology, 
we shall refer to the locally presentable $\bbK$-linear category $J_!(\AAAA)(d) \in \Pr_{\bbK}$ 
that is assigned by the extended 2AQFT $J_!(\AAAA)$ to an object $d \in \DD$ 
as {\it Fredenhagen's universal category}. In the context of Example \ref{ex:diskextension}, we 
will provide examples of such categories for 2AQFTs on the circle $M = \bbS^1$.

\subsection{Preliminaries}
Our construction of the extension pseudo-functor
$J_!: \2AQFT(\ovr{\CC}) \to \2AQFT(\ovr{\DD})$ associated to
a full orthogonal subcategory embedding $J:\ovr{\CC}\to\ovr{\DD}$
uses the {\em monoidal envelope} $\P^\otimes_{\ovr{\CC}}$ of the prefactorization operad
$\P_{\ovr{\CC}}$ from Definition \ref{def:PCoperad}.
We refer the reader to \cite[Theorem 4.2]{Mandell}
for details on monoidal envelopes for colored operads.
In our case of interest, $\P_{\ovr{\CC}}^\otimes$
is given by the following symmetric monoidal category: 
\begin{description}
\item[Objects:] (Possibly empty) tuples 
$\und{c} =(c_1,\dots,c_n)\in \CC^n$ of objects in $\P_{\ovr{\CC}}$.

\item[Morphisms:] Pairs $(\alpha,\und{\und{f}}): 
\und{c} = (c_1,\ldots,c_n) \to \und{t} = (t_1,\ldots,t_m)$
with $\alpha: \{1,\dots,n\} \to \{1,\dots,m\}$ a map of sets
and $\und{\und{f}} = (\und{f}_1,\dots,\und{f}_m)$ a tuple of operations
$\und{f}_j = (f_{j1},\dots,f_{jk_j}) \in \P_{\ovr{\CC}}\big(\substack{ t_j \\ \und{c}_{\alpha,j}}\big)$, 
for $j=1,\dots,m$, where $\und{c}_{\alpha,j}$ denotes 
the (possibly empty) sub-tuple of $\und{c}$ containing only 
the $c_i$'s satisfying $\alpha(i)=j$ 
and $k_j$ denotes the length of $\und{c}_{\alpha,j}$. 

\item[Identities and composition:] The identity morphism
for $\und{c} = (c_1,\ldots,c_n)\in \P_{\ovr{\CC}}^\otimes$
is given by $\id_{\und{c}} := (\id,(\id_{c_1},\dots,\id_{c_n})): \und{c} \to \und{c}$.
The composition of two morphisms $(\alpha,\und{\und{f}}) : \und{b} \to \und{a}$
and $(\beta,\und{\und{g}}) : \und{a} \to \und{t}$  in $\P_{\ovr{\CC}}^\otimes$
is given by $(\beta,\und{\und{g}}) \circ (\alpha,\und{\und{f}})  := 
(\beta \alpha, \und{\und{h}}): \und{b} \to \und{t}$, 
where $\beta\alpha$ is the usual composition of maps of sets
and $\und{\und{h}} := (\und{h}_1,\ldots,\und{h}_\ell)$ 
is the tuple of operations $\und{h}_k := \und{g}_k\, \und{\und{f}}_{\beta,k}
\in \P_{\ovr{\CC}}\big(\substack{ t_k \\ \und{b}_{\beta\alpha,k}}\big)$ determined
by operadic composition, for $k=1,\dots,\ell$, where $\und{\und{f}}_{\beta,k}$ 
is the sub-tuple of $\und{\und{f}} = (\und{f}_1,\dots,\und{f}_m)$ 
containing only the $\und{f}_j$'s satisfying $\beta(j) = k$.

\item[Symmetric monoidal structure:] The tensor product 
$\und{c} \otimes \und{c}^\prime := (\und{c},\und{c^\prime})$ 
is defined by concatenation of tuples, the monoidal unit is 
the empty tuple $\emptyset$ and the symmetric braiding  
is given by the $\P_{\ovr{\CC}}^\otimes$-morphisms 
$(\alpha_{n,n^\prime},(\id_{c_1},\dots,\id_{c_n},\id_{c_1^\prime},\dots,\id_{c^\prime_{n^\prime}})): 
\und{c} \otimes \und{c}^\prime \to \und{c}^\prime \otimes \und{c}$, where 
$\alpha_{n,n^\prime}: \{1,\ldots,n+n^\prime\} \to \{1,\dots,n+n^\prime\}$ 
is defined by $\alpha_{n,n^\prime}(i) = n^\prime+i$, 
for $i=1,\dots,n$, and $\alpha_{n,n^\prime}(n+i) = i$, 
for $i=1,\dots,n^\prime$. 
\end{description}

Applying the same construction to $\P_{\ovr{\DD}}$ 
defines a symmetric monoidal category $\P_{\ovr{\DD}}^\otimes$.
Furthermore, the orthogonal functor $J: \ovr{\CC} \to \ovr{\DD}$
induces an operad morphism $J :\P_{\ovr{\CC}}\to \P_{\ovr{\DD}}$ and hence
a symmetric monoidal functor $J^\otimes: \P_{\ovr{\CC}}^\otimes \to \P_{\ovr{\DD}}^\otimes$ 
between the monoidal envelopes. The latter reads explicitly as follows:
\begin{description}
\item[On objects:] For $\und{c} = (c_1,\ldots,c_n) \in 
\P_{\ovr{\CC}}^\otimes$, we set $J^\otimes (\und{c})
:= (J(c_1),\dots,J(c_n)) \in \P_{\ovr{\DD}}^\otimes$.

\item[On morphisms:] For $(\alpha,\und{\und{f}}): \und{c} \to \und{t}$ 
in $\P_{\ovr{\CC}}^\otimes$, 
we set $J^\otimes(\alpha,\und{\und{f}}) := (\alpha,J(\und{\und{f}})): 
J^\otimes (\und{c}) \to J^\otimes (\und{t})$ in $\P_{\ovr{\DD}}^\otimes$, 
where $J(\und{\und{f}}) := ((J(f_{11}),\dots,J(f_{1k_1})),\dots,(J(f_{m1}),\dots,J(f_{mk_m})))$.

\item[Symmetric monoidal structure:] Since 
$J^\otimes(\und{c}) \otimes J^\otimes(\und{c}^\prime) 
= J^\otimes(\und{c} \otimes \und{c}^\prime)$ 
and $\emptyset = J^\otimes(\emptyset)$, it is straightforward 
to equip $J^\otimes$ with a symmetric monoidal structure. 
\end{description}

Recall from Definition \ref{def:2AQFT} that 2AQFTs on $\ovr{\CC}$
are by definition $\P_{\ovr{\CC}}$-algebras.
Hence, by the universal property of monoidal envelopes, we can associate
to every $\AAAA \in \2AQFT(\ovr{\CC})$ a symmetric monoidal pseudo-functor 
\begin{subequations}\label{eqn:undAAAA}
\begin{flalign}
\und{\AAAA}\,:\, \P_{\ovr{\CC}}^\otimes ~\longrightarrow~ \Pr_{\bbK}
\end{flalign} 
from the monoidal envelope of $\P_{\ovr{\CC}}$.
This pseudo-functor acts on objects $\und{c} = (c_1,\dots,c_n) \in \P_{\ovr{\CC}}^\otimes$ as
the  $n$-ary Kelly-Deligne tensor product
\begin{flalign}
\und{\AAAA}(\und{c}) \,:=\, \bigboxtimes_{i=1}^n \AAAA(c_i)
\end{flalign}
of the locally presentable $\bbK$-linear categories $\AAAA(c_i)\in\Pr_\bbK$, cf.\ Remark \ref{rem:PrKmonoidal}.
(By convention, we set $\und{\AAAA}(\emptyset) := \Vec_\bbK$ to be the monoidal unit of $\Pr_{\bbK}$.)
On morphisms $(\alpha,\und{\und{f}}): \und{c} \to \und{t}$ in $\P_{\ovr{\CC}}^\otimes$, 
this pseudo-functor acts as
\begin{flalign}
\und{\AAAA}(\alpha,\und{\und{f}})\,:\, 
\xymatrix@C=2em{
\und\AAAA(\und{c})= \bigboxtimes\limits_{i=1}^n \AAAA(c_i) \ar[r]^-{\simeq_\alpha} ~&~ 
\bigboxtimes\limits_{j=1}^m \und{\AAAA}(\und{c}_{\alpha,j})
\ar[rr]^-{\bigboxtimes_j \AAAA(\und{f}_j)} ~&&~ \bigboxtimes\limits_{j=1}^m \AAAA(t_j)=\und{\AAAA}(\und{t})
}\quad,
\end{flalign}
\end{subequations}
where $\simeq_\alpha$ is the equivalence in the symmetric monoidal $2$-category 
$\Pr_\bbK$ that is associated to the displayed permutation determined by $\alpha$. 
The coherence data for the symmetric monoidal pseudo-functor 
$\und{\AAAA} : \P_{\ovr{\CC}}^\otimes \to \Pr_{\bbK}$ are canonically given
by the coherence data for $\AAAA\in\2AQFT(\ovr{\CC})$ and the symmetric 
monoidal structure on $\Pr_\bbK$.

\subsection{\label{subsec:extension}Extension}
The extension pseudo-functor 
$J_!: \2AQFT(\ovr{\CC}) \longrightarrow \2AQFT(\ovr{\DD})$
in the biadjunction \eqref{eqn:J!J*biadjunction} is obtained 
canonically via operadic left pseudo-Kan extension along 
$J: \P_{\ovr{\CC}} \to \P_{\ovr{\DD}}$. 
Passing from colored operads to their monoidal envelopes, 
$J_!$ can be obtained via (categorical) left pseudo-Kan extension 
along $J^\otimes: \P_{\ovr{\CC}}^\otimes \to \P_{\ovr{\DD}}^\otimes$, 
cf.\ \cite{Horel}. Furthermore, the latter left pseudo-Kan 
extension can be computed in terms of suitable bicolimits 
\cite{Lack, biadjointtriangles}. Using this approach, 
we can now describe the extension $J_!(\AAAA) \in \2AQFT(\ovr{\DD})$ 
of a 2AQFT $\AAAA \in \2AQFT(\ovr{\CC})$. 
For each $d \in \DD$, Fredenhagen's universal category 
is the locally presentable $\bbK$-linear category 
\begin{flalign}\label{eqn:J!1}
J_!(\AAAA)(d) \,:=\, \bicolim \left( 
\xymatrix{ 
J^\otimes / (d) \ar[r]^-{\text{forget}} 
~&~ \P_{\ovr{\CC}}^\otimes \ar[r]^-{\und\AAAA} ~&~ \Pr_{\bbK} 
} 
\right)
\end{flalign}
obtained as a bicolimit in $\Pr_{\bbK}$, 
where $J^\otimes / (d)$ denotes the slice category 
for the functor $J^\otimes: \P_{\ovr{\CC}}^\otimes \to 
\P_{\ovr{\DD}}^\otimes$ over the object 
$(d) \in \P_{\ovr{\DD}}^\otimes$. Recall also \eqref{eqn:undAAAA} 
for the construction of the pseudo-functor 
$\und\AAAA: \P_{\ovr{\CC}}^\otimes \to \Pr_{\bbK}$.
(To avoid confusion, let us stress that the symbol $(d)$ stands for the tuple of length one 
that is defined by the object $d\in\DD$, i.e.\  $(d)\in\P^\otimes_{\ovr{\DD}}$
is an object in the monoidal envelope.)
This bicolimit always exists because $\Pr_\bbK$ is 
bicategorically cocomplete, see e.g.\ \cite[Lemma 2.5]{BCJF}.
For each tuple $\und{g} = (g_1,\dots,g_n)\in\P_{\ovr{\DD}}\big(\substack{s \\ \und{d}}\big)$
of mutually orthogonal $\DD$-morphisms, 
we set the factorization product 
\begin{subequations}
\begin{flalign}\label{eqn:J!2}
J_!(\AAAA)(\und{g})\,:\,\prod_{i=1}^n J_!(\AAAA)(d_i)~\longrightarrow~J_!(\AAAA)(s) 
\end{flalign}
to be the functor that is defined below, 
which is co-continuous and $\bbK$-linear in each entry: 
Consider the diagram 
\begin{flalign}\label{eqn:J!3}
\xymatrix@C=5em{
\prod\limits_{i=1}^n J^\otimes / (d_i) \ar[r]^-{\prod_i \text{forget}} \ar[d]_-{\und{g}_\ast} ~&~ \prod\limits_{i=1}^n \P_{\ovr{\CC}}^\otimes \ar[r]^-{\prod_i \und{\AAAA}} \ar[d]_-{\otimes^n} ~&~ \prod\limits_{i=1}^n \Pr_{\bbK} \ar[d]^-{\boxtimes^n} \ar@{=>}[dl]_-{(\star)} \\
J^\otimes / (s) \ar[r]_-{\text{forget}} ~&~ \P_{\ovr{\CC}}^\otimes \ar[r]_-{\und{\AAAA}} ~&~ \Pr_{\bbK}
}
\end{flalign}
\end{subequations}
where $\und{g}_\ast: \prod_{i=1}^n J^\otimes/(d_i) \to J^\otimes/(s)$ 
is the functor induced by post-composition with $\und{g}$ in 
the colored operad $\P_{\ovr{\DD}}$. By direct inspection, 
the left square commutes. In the right square, instead, 
the clockwise and counter-clockwise paths give functors 
that are related by the natural isomorphism $(\star)$ determined 
by the symmetric monoidal structure on the pseudo-functor $\und\AAAA$. 
Passing to bicolimits and recalling that the Kelly-Deligne 
tensor product $\boxtimes$ commutes with $\bicolim$ (in each entry) 
provides a co-continuous  $\bbK$-linear functor 
$\bigboxtimes_{i=1}^n J_!(\AAAA)(d_i) \to J_!(\AAAA)(s)$. 
Pre-composition with the canonical functor 
$\prod_{i=1}^n J_!(\AAAA)(d_i) \to \bigboxtimes_{i=1}^n 
J_!(\AAAA)(d_i)$, which is co-continuous and $\bbK$-linear 
in each entry, completes the construction of \eqref{eqn:J!2}. 
For the empty tuple $\und{d} = \emptyset$, the pointing 
$J_!(\AAAA)(\ast_s) \in J_!(\AAAA)(s)$ 
of Fredenhagen's universal category $J_!(\AAAA)(s)$ 
is obtained in the same fashion from \eqref{eqn:J!3}.
(Notice that empty products are initial categories, 
while $\otimes^0$ and $\boxtimes^0$ assign the respective monoidal units.) 
The coherence data, cf.\ Remark \ref{rem:2AQFTunpacked}, 
for the extended 2AQFT $J_!(\AAAA) \in \2AQFT(\ovr{\DD})$ 
are obtained canonically from the construction above 
and the symmetric monoidal pseudo-functor 
$\und\AAAA: \P_{\ovr{\CC}}^\otimes \to \Pr_{\bbK}$. 
\sk

For an arbitrary $d \in \DD$, we shall now describe Fredenhagen's 
universal category $J_!(\AAAA)(d)$ in fully explicit terms, using the 
prescription in \cite[Lemma 2.5]{BCJF} to compute the relevant bicolimit \eqref{eqn:J!1}. 
This is a two-step procedure: 
\begin{itemize}
\item[1.)] Every co-continuous $\bbK$-linear functor between two
locally presentable $\bbK$-linear categories
admits a right adjoint by the special adjoint functor theorem, 
cf.\ \cite{AdamekRosicky, BCJF}. 
Hence, from the pseudo-functor 
$\und\AAAA: \P_{\ovr{\CC}}^\otimes \to \Pr_{\bbK}$, 
we obtain a new pseudo-functor 
$\und\AAAA^\mathrm{R} : (\P_{\ovr{\CC}}^\otimes)^\op \to \Cat$ that acts 
on objects as $\und\AAAA$, i.e.\ $\und\AAAA^\mathrm{R}(\und{c}) := \und\AAAA(\und{c})$ for 
all $\und{c} \in \P_{\ovr{\CC}}^\otimes$,  and that assigns to a morphism $(\alpha,\und{\und{f}}): \und{c} \to \und{t}$ 
in $\P_{\ovr{\CC}}^\otimes$ the right adjoint of 
the co-continuous $\bbK$-linear functor assigned by $\und{\AAAA}$, i.e.\ $\und\AAAA(\alpha,\und{\und{f}}) \dashv 
\und\AAAA^\mathrm{R}(\alpha,\und{\und{f}}): \und\AAAA^\mathrm{R}(\und{t}) \to \und\AAAA^\mathrm{R}(\und{c})$. 
(Note that $\und\AAAA^\mathrm{R}$ is just a pseudo-functor to $\Cat$ and not necessarily to $\Pr_{\bbK}$
because the right adjoint functors $\und\AAAA^{\mathrm{R}}(\alpha,\und{\und{f}})$ may fail to be co-continuous.) 

\item[2.)] The category underlying the bicolimit \eqref{eqn:J!1} of 
$\und\AAAA \circ \text{forget}: J^\otimes / (d) \to \Pr_{\bbK}$ 
can be computed as a bilimit of the pseudo-functor $\und\AAAA^\mathrm{R} 
\circ \text{forget}: (J^\otimes / (d))^\op \to \Cat$. 
The outcome is a locally presentable $\bbK$-linear category 
in a canonical way, cf.\ \cite{BCJF}.
\end{itemize}

Using the explicit model \cite{Street,biadjointtriangles} for computing bilimits of pseudo-functors to $\Cat$,
we obtain the following description of Fredenhagen's universal category $J_!(\AAAA)(d)$ 
in terms of explicit data and conditions: 
\begin{description}
\item[Objects:] An object
\begin{flalign}
(V,\xi^V)\,:=\, \big(\{V_{\und{h}}\},\{\xi^V_{(\alpha,\und{\und{f}})}\}\big)\,\in\,J_!(\AAAA)(d)
\end{flalign}
consists of the following data:
\begin{subequations}
\begin{enumerate}[(1)]
\item For each object 
$\big(\und{h} := (\ast,\und{h}): \und{c}\to (d)\big) \in J^\otimes / (d)$,  
where $\ast: \{1, \dots,n\} \to \{1\}$ denotes the unique map of sets to the singleton $\{1\}$,
an object 
\begin{flalign}
V_{\und{h}} \, \in \, \und{\AAAA}(\und{c}) = \bigboxtimes_{i=1}^n \AAAA(c_i) \quad.
\end{flalign}

\item For each morphism $(\alpha,\und{\und{f}}) : \und{h}\to \und{h}^\prime$ in $J^\otimes / (d)$, an isomorphism 
\begin{flalign} 
\xi^V_{(\alpha,\und{\und{f}})}\,:\,  \und{\AAAA}^{\mathrm{R}}(\alpha,\und{\und{f}})\big(V_{\und{h}^\prime}\big) ~\longrightarrow~V_{\und{h}}
\end{flalign} 
in the category $\und\AAAA(\und{c})$.
\end{enumerate}
\end{subequations}

These data have to satisfy the following cocycle conditions:
\begin{subequations}\label{eqn:Vxicoherences}
\begin{itemize}
\item[(i)] For all objects $\big(\und{h}: \und{c} \to (d)\big) \in J^\otimes / (d)$, 
the diagram 
\begin{flalign}
\xymatrix@C=5em{
\ar[d]_-{\und{\AAAA}^{\mathrm{R} \,0}_{\und{c}}}^-{\cong} 
\und{\AAAA}^\mathrm{R}(\id_{\und{h}})\big(V_{\und{h}}\big)\ar[r]^-{\xi^V_{\id_{\und{h}}}}~&~V_{\und{h}}\\
V_{\und{h}}\ar[ru]_-{\id_{V_{\und{h}}}} ~&~
}
\end{flalign}
in $\und\AAAA(\und{c})$ commutes, where $\und{\AAAA}^{\mathrm{R}\, 0}_{\und{c}}$ denotes
the coherence isomorphisms for identities that are associated with the pseudo-functor $\und{\AAAA}^{\mathrm{R}}$.

\item[(ii)]
For all composable pairs of morphisms 
$(\alpha,\und{\und{f}}) : \und{h} \to \und{h}^\prime$ and
$(\beta,\und{\und{g}}) : \und{h}^\prime \to \und{h}^{\prime\prime}$ in $J^\otimes / (d)$, 
the diagram
\begin{flalign}
\xymatrix@C=6.5em{
\ar[d]_-{\und{\AAAA}^{\mathrm{R}\,2}_{( (\beta,\und{\und{g}}) ,(\alpha,\und{\und{f}}) )}} ^-{\cong} 
\und{\AAAA}^\mathrm{R}(\alpha,\und{\und{f}})\, \und{\AAAA}^\mathrm{R}(\beta,\und{\und{g}})\big(V_{\und{h}^{\prime\prime}}\big) 
\ar[r]^-{ \und{\AAAA}^\mathrm{R}(\alpha,\und{\und{f}})(\xi^V_{(\beta,\und{\und{g}}) })}
 ~&~ 
 \und{\AAAA}^\mathrm{R}(\alpha,\und{\und{f}})\big(V_{\und{h}^\prime}\big)\ar[d]^-{\xi^V_{(\alpha,\und{\und{f}})}}\\
 \und{\AAAA}^\mathrm{R}\big((\beta,\und{\und{g}}) \circ (\alpha,\und{\und{f}}) \big)\big(V_{\und{h}^{\prime\prime}}\big)
 \ar[r]_-{\xi^V_{(\beta,\und{\und{g}}) \circ (\alpha,\und{\und{f}}) }}
~&~
V_{\und{h}}
}
\end{flalign}
in $\und\AAAA(\und{c})$ commutes, where $\und{\AAAA}^{\mathrm{R}\,2}_{( (\beta,\und{\und{g}}) ,(\alpha,\und{\und{f}}) )}$
denotes the coherence isomorphisms for compositions that are associated
with the pseudo-functor $\und{\AAAA}^\mathrm{R}$.
\end{itemize}
\end{subequations}

\item[Morphisms:] A morphism 
\begin{flalign}
\Gamma \,:=\, \{\Gamma_{\und{h}}\}\, : \, (V,\xi^V) ~\longrightarrow~ (W,\xi^W)
\end{flalign}
in $J_!(\AAAA)(d)$ consists of a family of $\und\AAAA(\und{c})$-morphisms
\begin{subequations}\label{eqn:Gammacoherences}
\begin{flalign}
\Gamma_{\und{h}}\,:\, V_{\und{h}} ~\longrightarrow~W_{\und{h}}\quad,
\end{flalign}
for all $\big(\und{h} : \und{c}\to (d) \big)\in J^\otimes/(d)$,
such that the diagrams
\begin{flalign}
\xymatrix@C=6.5em{
\ar[d]_-{\xi^V_{(\alpha,\und{\und{f}})}} \und{\AAAA}^\mathrm{R}(\alpha,\und{\und{f}}) \big(V_{\und{h}^\prime}\big) 
\ar[r]^-{ \und{\AAAA}^\mathrm{R}(\alpha,\und{\und{f}})(\Gamma_{\und{h}^\prime})}
~&~\und{\AAAA}^\mathrm{R}(\alpha,\und{\und{f}})\big(W_{\und{h}^\prime}\big)\ar[d]^-{\xi^W_{(\alpha,\und{\und{f}})}}\\
V_{\und{h}} \ar[r]_-{\Gamma_{\und{h}}}~&~W_{\und{h}}
}
\end{flalign}
\end{subequations}
in $\und\AAAA(\und{c})$ commute, for all morphisms 
$(\alpha,\und{\und{f}}) : \und{h}\to\und{h}^\prime$ in $J^\otimes / (d) $.

\item[Identities and composition:] Identities and composition 
are defined component-wise. 
\end{description}

\subsection{Examples on $M=\bbS^1$}
The aim of this subsection is to study examples 
of Fredenhagen's universal category for the full
orthogonal subcategory embedding $J : \ovr{\Disk(\bbS^1)}\to \ovr{\Open(\bbS^1)}$
introduced in Example \ref{ex:diskextension}. Given any
$\AAAA\in \2AQFT(\ovr{\Disk(\bbS^1)})$, which is by definition only defined on open intervals in $\bbS^1$,
we are particularly interested in the locally presentable $\bbK$-linear category
\begin{flalign}\label{eqn:extS1}
J_!(\AAAA)(\bbS^1)\,=\, \bicolim \left( 
\xymatrix{ 
J^\otimes / (\bbS^1) \ar[r]^-{\text{forget}} 
~&~ \P_{\ovr{\Disk(\bbS^1)}}^\otimes \ar[r]^-{\und{\AAAA}} ~&~ \Pr_{\bbK} 
} 
\right)
\end{flalign}
that is assigned to the whole circle $\bbS^1$.
The slice category $J^\otimes / (\bbS^1)$ in the present case 
admits the following simple description:
An object is a tuple $\und{I} = (I_1,\dots,I_n)$ of mutually disjoint open 
intervals $I_i\subset \bbS^1 $, i.e.\ $I_i\cap I_j =\emptyset$ for all $i\neq j$.
A morphism $\alpha : \und{I} = (I_1,\dots,I_n)\to \und{J}=(J_1,\dots,J_m)$
is a map of sets $\alpha : \{1,\dots,n\}\to \{1,\dots,m\}$ such that
$I_i \subseteq J_{\alpha(i)}$, for all $i=1,\dots,n$.

\begin{ex}\label{ex:TruncFredCat}
Let us consider first the simplest case where  $\AAAA$ is truncated, i.e.\ $\AAAA = \iota(\AAA)$ 
with $\AAA\in \AQFT(\ovr{\Disk(\bbS^1)})$ an ordinary AQFT.
We have the following square of biadjunctions
\begin{flalign}\label{eqn:bigggsquare}
\xymatrix@R=3em@C=5em{
\ar@<-0.7ex>[d]_-{J_!} \AQFT(\ovr{\Disk(\bbS^1)}) \ar@<0.7ex>[r]^-{\iota} ~&~ \2AQFT(\ovr{\Disk(\bbS^1)}) \ar@<-0.7ex>[d]_-{J_!} \ar@<0.7ex>[l]^-{\pi}\\
\AQFT(\ovr{\Open(\bbS^1)}) \ar@<-0.7ex>[u]_-{J^\ast} \ar@<0.7ex>[r]^-{\iota}~&~\2AQFT(\ovr{\Open(\bbS^1)})\ar@<0.7ex>[l]^-{\pi} \ar@<-0.7ex>[u]_-{J^\ast}
}
\end{flalign}
where we recall that,  with a mild abuse of notation,  the extension-restriction adjunctions
for both ordinary AQFTs \eqref{eqn:Fredenhagenadjunction} and 2AQFTs \eqref{eqn:J!J*biadjunction}
are denoted by the same symbol $J_!\dashv J^\ast$.
The horizontal biadjunctions in this diagram are the inclusion-truncation biadjunctions from Theorem \ref{theo:biadjunction}.
By direct inspection, one confirms that the square formed by the right adjoint $2$-functors commutes, i.e.\
$\pi \,J^\ast = J^\ast\,\pi$, hence the square formed by the left adjoint pseudo-functors
commutes up to an equivalence, i.e.\ $\iota\,J_! \simeq J_!\,\iota$. 
We would like to stress that this observation implies the (expected) 
result that the extension of a truncated 2AQFT is truncated too.
More concretely,  we observe that
Fredenhagen's universal category for a truncated 2AQFT 
\begin{flalign}\label{eqn:truncequivalence}
J_!(\iota(\AAA))(\bbS^1)\, \simeq\,\Mod_{J_!(\AAA)(\bbS^1)}
\end{flalign}
is equivalent to the category of right modules over Fredenhagen's universal 
algebra $J_!(\AAA)(\bbS^1)\in \Alg_\bbK$. 
The latter is given by the ordinary colimit 
\begin{flalign}\label{eqn:ordinaryFredenhagenS1}
J_!(\AAA)(\bbS^1)\,=\, \colim \left( 
\xymatrix{ 
J^\otimes / (\bbS^1) \ar[r]^-{\text{forget}} 
~&~ \P_{\ovr{\Disk(\bbS^1)}}^\otimes \ar[r]^-{\und{\AAA}} ~&~ \Alg_{\bbK} 
} 
\right) \quad,
\end{flalign}
where $\und{\AAA}: \P_{\ovr{\Disk(\bbS^1)}}^{\otimes} \to \Alg_\bbK$
is the symmetric monoidal functor from the monoidal envelope that is 
determined by $\AAA\in\AQFT(\ovr{\Disk(\bbS^1)})$.
\sk

To obtain a better understanding of the objects and morphisms 
in our general presentation of Fredenhagen's universal category 
$J_!(\iota(\AAA))(\bbS^1)$, we construct explicitly a functor
$\Mod_{J_!(\AAA)(\bbS^1)} \to J_!(\iota(\AAA))(\bbS^1)$ that 
implements the equivalence \eqref{eqn:truncequivalence}.
Let us first describe this functor on objects.
Given any right module $V\in \Mod_{J_!(\AAA)(\bbS^1)}$ over Fredenhagen's 
universal algebra, we use the canonical $\Alg_\bbK$-morphisms
$\chi_{\und{I}}^{} : \und{\AAA}(\und{I})=\bigotimes_{i=1}^n \AAA(I_i)\to J_!(\AAA)(\bbS^1)$
to the colimit \eqref{eqn:ordinaryFredenhagenS1} to define
\begin{flalign}\label{eqn:VItruncated}
V_{\und{I}}^{} \,:=\, \chi_{\und{I}}^{\ast}\big(V\big)\,\in\,\Mod_{\und{\AAA}(\und{I})}
\end{flalign}
by restriction of modules, for each $\und{I}\in J^\otimes / (\bbS^1) $.
Given any morphism $\alpha : \und{I}\to\und{J}$ in $J^\otimes / (\bbS^1)$,
the functor $\iota(\AAA)^\mathrm{R}(\alpha) = \und{\AAA}(\alpha)^\ast : 
\Mod_{\und{\AAA}(\und{J})} \to \Mod_{\und{\AAA}(\und{I})}$ is given by restriction
of modules along the $\Alg_\bbK$-morphism $ \und{\AAA}(\alpha) :  \und{\AAA}(\und{I})\to  \und{\AAA}(\und{J})$.
Because $\{\chi_{\und{I}}^{}\}_{\und{I}}^{}$ is a co-cone, we obtain
\begin{subequations}\label{eqn:xiatruncated}
\begin{flalign}
\iota(\AAA)^\mathrm{R}(\alpha) \big(V_{\und{J}}^{}\big) =  \und{\AAA}(\alpha)^\ast\,\chi_{\und{J}}^{\ast}\big(V\big)
= \big(\chi_{\und{J}}^{}\,\und{\AAA}(\alpha)\big)^{\ast}\big(V\big)= \chi_{\und{I}}^\ast\big(V\big) = V_{\und{I}}^{}
\end{flalign}
and therefore we can set
\begin{flalign}
\xi^V_\alpha \,:=\, \id_{V_{\und{I}}^{}}^{}\,:\, \iota(\AAA)^\mathrm{R}(\alpha) \big(V_{\und{J}}^{}\big) ~\longrightarrow~V_{\und{I}}^{}
\end{flalign}
\end{subequations}
to be the identity morphism. One easily checks that the coherence conditions \eqref{eqn:Vxicoherences} are satisfied,
hence we have defined an object $(V,\xi^V)\in J_!(\iota(\AAA))(\bbS^1)$.
Let us now define the functor $\Mod_{J_!(\AAA)(\bbS^1)} \to J_!(\iota(\AAA))(\bbS^1)$ on morphisms.
Given any morphism $L : V\to W$ in $\Mod_{J_!(\AAA)(\bbS^1)}$, consider the restrictions
\begin{flalign}\label{eqn:Ltruncated}
L_{\und{I}}^{} \,:=\, \chi_{\und{I}}^\ast\big(L\big)\,:\, V_{\und{I}}^{} = \chi_{\und{I}}^\ast\big(V\big)~\longrightarrow~\chi_{\und{I}}^\ast\big(W\big)=W_{\und{I}}^{} \quad,
\end{flalign}
for all $\und{I}\in J^\otimes / (\bbS^1) $. One easily checks that the coherence conditions \eqref{eqn:Gammacoherences} are satisfied,
hence we have defined a morphism $L:(V,\xi^V)\to (W,\xi^W)$ in $J_!(\iota(\AAA))(\bbS^1)$. 
Using the universal property of the colimit \eqref{eqn:ordinaryFredenhagenS1}, one checks that the
resulting functor $\Mod_{J_!(\AAA)(\bbS^1)} \to J_!(\iota(\AAA))(\bbS^1)$ implements the equivalence \eqref{eqn:truncequivalence}.
\sk

Summing up, we have found that, in the case of a truncated 2AQFT $\AAAA=\iota(\AAA)$,
the objects of Fredenhagen's universal category can be described as
families of right modules \eqref{eqn:VItruncated} over the local algebras 
$\und{\AAA}(\und{I}) = \bigotimes_{i=1}^n\AAA(I_i)$  on disjoint unions of intervals, 
whose restrictions along inclusions $\alpha : \und{I}\to\und{J}$ coincide \eqref{eqn:xiatruncated}.
Morphisms in Fredenhagen's universal category can be described by 
locally defined module morphisms \eqref{eqn:Ltruncated}, whose restrictions along inclusions $\alpha : \und{I}\to\und{J}$ coincide.
\end{ex}

\begin{ex}\label{ex:RepGFredCat}
Let us consider the gauging  $\bbK^G\in\2AQFT(\ovr{\Disk(\bbS^1)})$ 
of the trivial AQFT $\bbK\in \AQFT(\ovr{\Disk(\bbS^1)})$ with respect to the trivial action of
a finite group $G$, which is a non-truncated 2AQFT for every non-trivial group $G \neq \{ e\}$,
see Example \ref{ex:trivialAQFTgauging}. Because 2AQFTs are by definition
prefactorization algebras with values in $\Pr_\bbK$ (cf.\ Definition \ref{def:2AQFT})
and $\bbK^G\in\2AQFT(\ovr{\Disk(\bbS^1)})$ is also {\em locally constant},
i.e.\ $\bbK^G(\iota_I^J) : \bbK^G(I)\to\bbK^G(J)$ is an equivalence in $\Pr_\bbK$ 
for every interval inclusion $\iota_I^J: I\to J$,
we can compute Fredenhagen's universal category  $J_!(\bbK^G)(\bbS^1)$
for this particular example by factorization homology \cite{AyalaFrancis}.
Using in particular \cite[Theorem 3.19]{AyalaFrancis}, we obtain that
\begin{flalign}
J_!(\bbK^G)(\bbS^1) \,\simeq\, \mathrm{HH}_\bullet\big(\Rep_\bbK(G)\big)
\end{flalign}
is equivalent to the Hochschild homology
of the associative and unital algebra 
$(\Rep_\bbK(G),\otimes,\bbK)\in \Alg_\mathsf{As}(\Pr_\bbK)$
in $\Pr_{\bbK}$. (The latter is just the usual monoidal category
of $\bbK$-linear representations of $G$, regarded internally in the 
symmetric monoidal $2$-category  $\Pr_{\bbK}$.) 
Hochschild homology can be computed as a bicolimit (in $\Pr_\bbK$)
\begin{flalign}\label{eqn:HHsimplicial}
\mathrm{HH}_\bullet\big(\Rep_\bbK(G)\big)\,=\,
\bicolim\left(
\xymatrix@C=2em{
\Rep_\bbK(G) ~&~ \ar@<-0.5ex>[l]\ar@<0.5ex>[l]\Rep_\bbK(G)^{\boxtimes 2} ~&~\ar@<0ex>[l] \ar@<-1ex>[l]\ar@<1ex>[l]\Rep_\bbK(G)^{\boxtimes 3} ~&~
\ar@<-0.5ex>[l]\ar@<-1.5ex>[l]\ar@<0.5ex>[l]\ar@<1.5ex>[l]\cdots
}
\right)
\end{flalign}
of the simplicial diagram associated with $(\Rep_\bbK(G),\otimes,\bbK)\in \Alg_\mathsf{As}(\Pr_\bbK)$,
see e.g.\ \cite[Section 5.1]{DerivedCenters}. (As usual, we suppress the degeneracy maps in \eqref{eqn:HHsimplicial}.)
Since we are working in a $2$-categorical context, this simplicial diagram may be truncated
after $\Rep_\bbK(G)^{\boxtimes 3}$. 
\sk

We will now compute the bicolimit \eqref{eqn:HHsimplicial} explicitly 
by using the techniques of \cite{BCJF}, see also the end of Section \ref{subsec:extension} for a short summary.
A more conceptual explanation of the obtained result is given in Remark \ref{rem:centercocenter} below.
By  \cite[Lemma 2.5]{BCJF}, we can compute this bicolimit
in terms of the bilimit (in the $2$-category $\Cat$ of categories)
\begin{flalign}\label{eqn:HHsimplicialLIMIT}
\mathrm{HH}_\bullet\big(\Rep_\bbK(G)\big)\,=\,
\mathrm{bilim}\left(
\xymatrix@C=2em{
\Rep_\bbK(G)  \ar@<-0.5ex>[r]\ar@<0.5ex>[r] ~&~\Rep_\bbK(G^2) \ar@<0ex>[r] \ar@<-1ex>[r]\ar@<1ex>[r] ~&~ \Rep_\bbK(G^{3})
}
\right)
\end{flalign}
of the truncated cosimplicial diagram obtained by taking right adjoints of the face and degeneracy maps in \eqref{eqn:HHsimplicial}.
In this expression we have also used that $\Rep_\bbK(G)^{\boxtimes n}\simeq \Rep_\bbK(G^n)$ is equivalent 
to the representation category of the product group $G^n$. The coface and codegeneracy maps
in \eqref{eqn:HHsimplicialLIMIT} are given by coinduced representation functors
$\phi_\ast : \Rep_\bbK(G^\prime)\to\Rep_\bbK(G^{\prime\prime})$ for
suitable group maps $\phi : G^\prime\to G^{\prime\prime}$. Concretely,
we have that 
\begin{flalign}
\delta^0 = \delta^1 = {\Delta}_\ast \, :\, \Rep_\bbK(G)~\longrightarrow~ \Rep_\bbK(G^2)
\end{flalign}
for the diagonal map $\Delta : G\to G^2\,,~ g\mapsto (g,g)$, and that 
\begin{subequations}
\begin{flalign}
\delta^i = \phi^i_\ast \, :\, \Rep_\bbK(G^2)~\longrightarrow~ \Rep_\bbK(G^3)
\end{flalign}
for
\begin{flalign}
\phi^i \,:\, G^2~\longrightarrow~G^3 ~~,\quad (g_1,g_2) ~\longmapsto~\begin{cases}
(g_1,g_1,g_2)&~,~~\text{for }i=0~~,\\
(g_1,g_2,g_2)&~,~~\text{for }i=1~~,\\
(g_1,g_2,g_1)&~,~~\text{for }i=2~~.
\end{cases}
\end{flalign}
\end{subequations}
The codegeneracy map $\epsilon^0 : \Rep_\bbK(G^2)\to \Rep_\bbK(G)$
is given by the coinduced representation functor for $G^2\to G\,,~(g_1,g_2)\mapsto g_1$.
\sk

We are now ready to describe the bilimit \eqref{eqn:HHsimplicialLIMIT} and 
hence the category $\mathrm{HH}_\bullet\big(\Rep_\bbK(G)\big)$ in more explicit terms:
\begin{itemize}
\item An object is a tuple $(V,\theta^V)$, where $V\in\Rep_\bbK(G)$
and $\theta^V : \delta^1(V)\to\delta^0(V) $ is an isomorphism in $\Rep_\bbK(G^2)$,
such that $\epsilon^0(\theta^V) =\id_V$ and $\delta^0(\theta^V)\circ \delta^2(\theta^V) = \delta^1(\theta^V)$
in $\Rep_\bbK(G^3)$.

\item  A morphism $L : (V,\theta^V)\to (W,\theta^W)$ is a morphism $L:V\to W$
in $\Rep_\bbK(G)$, such that the diagram
\begin{flalign}
\xymatrix@C=5em{
\ar[d]_-{\theta^V}\delta^1(V) \ar[r]^-{\delta^1(L)}~&~\delta^1(W)\ar[d]^-{\theta^W}\\
\delta^0(V)\ar[r]_-{\delta^0(L)}~&~\delta^0(W)
}
\end{flalign}
in $\Rep_\bbK(G^2)$ commutes.
\end{itemize}
We can simplify this description further by using  explicit models for the coinduced 
representation functors $\phi_\ast : \Rep_\bbK(G^\prime)\to \Rep_{\bbK}(G^{\prime\prime})$
for group maps $\phi : G^{\prime}\to G^{\prime\prime}$. 
Since we consider only {\em finite} groups and a base field $\bbK$ of characteristic $0$,
there exists a natural isomorphism between the coinduced and the induced 
representation functors $\phi_\ast\cong \phi_! : \Rep_\bbK(G^\prime)\to \Rep_{\bbK}(G^{\prime\prime}) $.
The latter is easy to describe: For $V\in \Rep_\bbK(G^\prime)$, we set
$\phi_!(V) := \bbK[G^{\prime\prime}] \otimes_{\bbK[G^{\prime}]} V\in  \Rep_\bbK(G^{\prime\prime})$ 
to be the relative tensor product, where $\bbK[G^{\prime}]$ and $\bbK[G^{\prime\prime}]$
denote the group Hopf algebras associated with the finite groups $G^{\prime}$ and $G^{\prime\prime}$.
(Recall that $ \Rep_\bbK(G^\prime) = {}_{\bbK[G^\prime]}\Mod$ is the category of left $\bbK[G^\prime]$-modules,
and similar for $G^{\prime\prime}$.) Given any object $(V,\theta^V)\in \mathrm{HH}_\bullet\big(\Rep_\bbK(G)\big)$, 
we use this explicit description to deduce that $\theta^V : \bbK[G^2]\otimes_{\bbK[G]} V\to \bbK[G^2]\otimes_{\bbK[G]} V$
is completely determined by a $\bbK$-linear map $\vartheta^V: V \to \bbK[G]\otimes V$ via
$\theta^V(1\otimes 1 \otimes v) = 1\otimes \vartheta^V(v)$, which is $G$-equivariant
with respect to the adjoint action on $\bbK[G]$ and satisfies the axioms of a left $\bbK[G]$-coaction.
Moreover, we deduce that a morphism in $ \mathrm{HH}_\bullet\big(\Rep_\bbK(G)\big)$ is a $G$-equivariant
map that preserves these $\bbK[G]$-coactions.
In summary, we have obtained the following chain of equivalences
\begin{flalign}\label{eqn:RepGFredenhagen}
J_!(\bbK^G)(\bbS^1) \,\simeq\, \mathrm{HH}_\bullet\big(\Rep_\bbK(G)\big)
\,\simeq\, G{\text{-}}^{\bbK[G]}\Mod \,\simeq\, G{\text{-}}\Mod_{\O(G)}\quad,
\end{flalign}
where in the last step we have used that $\bbK[G]$-comodules are equivalent to
modules over the dual Hopf algebra $\O(G)$ of functions on $G$. (The $G$-action
on $\O(G)$ is again the adjoint action.)
\sk

Let us briefly explain the physical interpretation of this result.
By Remark \ref{rem:orbifoldinterpretation}, we can interpret
$\bbK^G\in\2AQFT(\ovr{\Disk(\bbS^1)})$ as an orbifold $\sigma$-model 
that is defined on intervals and whose target is 
the classifying stack $\mathbf{B}G = \{\ast\}//G$ of $G$.
Indeed, the stack of fields on an interval $I\subset \bbS^1$ 
is $\mathrm{Fields}(I) = \mathrm{Map}(I,\mathbf{B}G) \simeq \{\ast\}//G$
and its category of quasi-coherent sheaves is
$\QCoh(\mathrm{Fields}(I)) \simeq  \Rep_\bbK(G)$, which coincides with the category
that the 2AQFT $\bbK^G$ assigns to intervals. On the whole circle $\bbS^1$,
the stack of fields of this orbifold $\sigma$-model is given by the loop stack
$\mathrm{Fields}(\bbS^1) = \mathrm{Map}(\bbS^1,\mathbf{B}G) \simeq \mathrm{Bun}_G(\bbS^1)$,
which is equivalent to the stack of principal $G$-bundles on $\bbS^1$. The non-trivial
bundles can be interpreted physically as ``twisted sectors'' of this orbifold $\sigma$-model, 
see e.g.\ \cite{Dijkgraaf} and also \cite{JohnsonFreydOrbifold}. The category of quasi-coherent sheaves
on this stack is given by $\QCoh(\mathrm{Fields}(\bbS^1) ) \simeq G{\text{-}}\Mod_{\O(G)}$,
which coincides with our result for Fredenhagen's universal category \eqref{eqn:RepGFredenhagen}.
Hence, Fredenhagen's universal category successfully detects all  ``twisted sectors''
for this simple example of an orbifold $\sigma$-model.
\end{ex}

\begin{rem}\label{rem:centercocenter}
The category \eqref{eqn:RepGFredenhagen} that we obtain for the circle is the 
representation category of the groupoid of principal $G$-bundles over the circle, i.e.\ the representation category of 
the loop groupoid $G//G$ of $G$ (the action groupoid of the action of $G$ on itself by conjugation). 
This category is also the Drinfeld center of $\Rep_\bbK(G)$, i.e.\ the Hochschild \emph{co}homology.
As a consequence, the Hochschild homology and Hochschild cohomology for $\Rep_\bbK(G)$ are equivalent.
More general conditions under which one finds such an equivalence are given in \cite[Corollary~3.1.5]{douglasetal}
within the framework of finite tensor categories and in \cite[Theorem~1.7]{DerivedCenters} within the framework of derived algebraic geometry.
\end{rem}

\begin{ex}
As a last example, we discuss briefly the gauging $\AAA^G\in\2AQFT(\ovr{\Disk(\bbS^1)})$
of an arbitrary $G$-equivariant AQFT $(\AAA,\rho)\in G{\text{-}}\AQFT(\ovr{\Disk(\bbS^1)})$,
which includes Examples \ref{ex:TruncFredCat} and \ref{ex:RepGFredCat} as very special cases.
Unfortunately, it seems to be very hard to simplify our explicit description of 
Fredenhagen's universal category $J_!(\AAA^G)(\bbS^1)$ in this general case.
(Note that computing this category as in Example \ref{ex:RepGFredCat} 
by importing techniques from factorization homology is in general not possible,
because we are also interested in 2AQFTs that are {\em not} locally constant.)
In order to develop a better understanding of the category $J_!(\AAA^G)(\bbS^1)$,
we shall specialize our general description of 
Fredenhagen's universal category from the end of Section \ref{subsec:extension} to our example at hand.
Concretely, an object $(V,\xi^V)\in J_!(\AAA^G)(\bbS^1)$ consists of the following data:
\begin{enumerate}[(1)]
\item For each tuple $\und{I}=(I_1,\dots,I_n)\in J^\otimes/(\bbS^1)$ of mutually disjoint intervals, 
a $G^n$-equivariant module 
\begin{flalign}\label{eqn:Vgeneralorbifold}
V_{\und{I}}\,\in\, G^n{\text{-}}\Mod_{\und{\AAA}(\und{I})}
\end{flalign} over the 
tensor product algebra $\und{\AAA}(\und{I}) = \bigotimes_{i=1}^n\AAA(I_i)$.
(The $G^n$-action on the tensor product algebra is given by the component-wise $G$-actions.)

\item For each morphism $\alpha : \und{I}=(I_1,\dots,I_n) \to\und{J}=(J_1,\dots,J_m)$ in $J^\otimes/(\bbS^1)$,
a $G^n{\text{-}}\Mod_{\und{\AAA}(\und{I})}$-isomorphism
\begin{flalign}\label{eqn:xiVgeneralorbifold}
\xi_{\alpha}^V \,:\, (\AAA^G)^{\mathrm{R}}(\alpha)\big(V_{\und{J}}\big)~\longrightarrow~V_{\und{I}}\quad.
\end{flalign}
Here $(\AAA^G)^{\mathrm{R}} (\alpha): G^m{\text{-}}\Mod_{\und{\AAA}(\und{J})}\to G^n{\text{-}}\Mod_{\und{\AAA}(\und{I})}$
is the right adjoint of the functor
\begin{flalign}
\xymatrix@C=2em{
\ar[dr]_-{\Delta_\alpha^\ast} G^n{\text{-}}\Mod_{\und{\AAA}(\und{I})} \ar[rr]^-{\AAA^G(\alpha)}~&~~&~G^m{\text{-}}\Mod_{\und{\AAA}(\und{J})}\\
~&~G^m{\text{-}}\Mod_{\Delta_\alpha^\ast(\und{\AAA}(\und{I}))}\ar[ru]_-{\und{\AAA}(\alpha)_!} ~&~
}
\end{flalign}
where $\Delta_\alpha : G^m\to G^n\,,~(g_1,\dots,g_m)\mapsto (g_{\alpha(1)},\dots,g_{\alpha(n)})$
is the group map determined by $\alpha : \{1,\dots,n\}\to\{1,\dots,m\}$,
$\Delta_\alpha^\ast : \Rep_\bbK(G^n)\to\Rep_\bbK(G^m)$ denotes the corresponding restricted representation functor
and $\und{\AAA}(\alpha)_!$  is the induced module functor for the $G^m$-equivariant
algebra morphism $\und{\AAA}(\alpha) : \Delta_\alpha^\ast(\und{\AAA}(\und{I})) \to\und{\AAA}(\und{J})$.
Explicitly, one finds that $(\AAA^G)^{\mathrm{R}} (\alpha)$ is given by the composition
\begin{flalign}
\xymatrix@C=5em{
G^m{\text{-}}\Mod_{\und{\AAA}(\und{J})}\ar[d]_-{\und{\AAA}(\alpha)^\ast} \ar[r]^-{(\AAA^G)^{\mathrm{R}} (\alpha)}~&~
G^n{\text{-}}\Mod_{\und{\AAA}(\und{I})} \\
G^m{\text{-}}\Mod_{\Delta_\alpha^\ast(\und{\AAA}(\und{I}))} \ar[r]_-{{\Delta_\alpha}_\ast} ~&~ G^n{\text{-}}\Mod_{{\Delta_\alpha}_\ast\Delta_\alpha^\ast(\und{\AAA}(\und{I}))} \ar[u]_-{\eta_{\und{\AAA}(\und{I})}^\ast}
}
\end{flalign}
where $\eta$ denotes the unit of the adjunction
$\Delta_\alpha^\ast : \Rep_\bbK(G^n) \rightleftarrows \Rep_\bbK(G^m) : {\Delta_\alpha}_\ast$.
\end{enumerate}
These data have to satisfy the coherence conditions \eqref{eqn:Vxicoherences}.
\sk

Observe from \eqref{eqn:Vgeneralorbifold} that $V_{\und{I}}$ is a module
over the tensor product algebra $\und{\AAA}(\und{I}) = \bigotimes_{i=1}^n\AAA(I_i)$
associated to a tuple of mutually disjoint intervals together with a {\em separate
$G$-action for each connected component}. In other words, the group $G$ is allowed to 
act differently on different intervals, which is a characteristic feature of a {\em local} gauge symmetry.
To understand better the coherence maps \eqref{eqn:xiVgeneralorbifold},
let us consider the case where we include two intervals into a single bigger interval,
i.e.\ $\alpha : \und{I} = (I_1,I_2)\to J$. In this case $\Delta_\alpha = \Delta : G\to G^2$
is the diagonal map and \eqref{eqn:xiVgeneralorbifold} is given by a
$G^2{\text{-}}\Mod_{\und{\AAA}(\und{I})}$-isomorphism
\begin{flalign}
\xi^V_\alpha \,:\, \eta_{\und{\AAA}(\und{I})}^\ast\,\Delta_\ast\,\und{\AAA}(\alpha)^\ast\big(V_J\big)~\longrightarrow~
V_{\und{I}}\quad.
\end{flalign}
Using as in Example \ref{ex:RepGFredCat} that $\Delta_\ast\cong\Delta_! : \Rep_\bbK(G)\to\Rep_\bbK(G^2)$ 
is naturally isomorphic to the induced representation functor, we obtain that $\xi^V_\alpha$ is completely determined
by a $\bbK$-linear map $\kappa_{\alpha}^V : V_J\to V_{\und{I}}$ via
$\xi^V_\alpha(1\otimes 1\otimes v) =  \kappa_\alpha^V(v)$, for all $v\in V_J$. 
This $\bbK$-linear map has to satisfy the following conditions:
1.)~$G$-equivariance: $\kappa_\alpha^V(g\,v) = (g,g)\,\kappa^V_\alpha(v)$, for all $v\in V_J$
and $g\in G$. 2.)~Preservation of the $\und{\AAA}(\und{I})$-actions:
\begin{flalign}\label{eqn:TMP2intervals}
\kappa_{\alpha}^V(v) \cdot (a_1\otimes a_2) = 
\sum_{(g_1,g_2)\in G^2} (g_1^{-1},g_2^{-1})\,\kappa_\alpha^V\Big( v\cdot \Big(\AAA(\iota_{I_1}^J)(g_1\, a_1)\,\AAA(\iota_{I_2}^J)(g_2\,a_2)\Big)\Big)\quad,
\end{flalign}
for all $a_1\otimes a_2\in \AAA(I_1)\otimes \AAA(I_2)$ and $v\in V_J$, where 
$\iota_{I_i}^J : I_i\to J$ denote the interval inclusions.
(The sum over $G^2$ comes from the unit $\eta$ of the adjunction $\Delta^\ast \dashv \Delta_\ast$
when we use $\Delta_!$ as a model for $\Delta_\ast$.) Comparing  \eqref{eqn:TMP2intervals} 
with the truncated case from Example \ref{ex:TruncFredCat}, we observe that 
there is a component-wise $G^2$-action on the algebra element
$a_1\otimes a_2\in \AAA(I_1)\otimes \AAA(I_2)$ on a pair of intervals 
before it acts on the module element $v\in V_J$ on the single bigger interval.
From a superficial point of view, this behavior resembles the twisted representations of
$G$-equivariant AQFTs by M\"uger \cite{Mueger}. Unfortunately, we do not
understand at the moment if there exists a precise relationship between Fredenhagen's universal category
$J_!(\AAA^G)(\bbS^1)$ for categorified orbifold theories and the results in  \cite{Mueger}.
\end{ex}

%%%%%%%%%%%%%%%%%%%%%%%%%%%%%%%%%%%%%%%%%%%%%%%%
%%%%%%%%%%%%%%%%%%%%%%%%%%%%%%%%%%%%%%%%%%%%%%%%

\section*{Acknowledgments}
We would like to thank Simen Bruinsma, Chris Fewster, Ignacio Lopez Franco, 
Klaus Fredenhagen, Owen Gwilliam, Theo Johnson-Freyd and Robert Laugwitz
for useful discussions about this work.
We also would like to thank the referees for their detailed comments 
and suggestions that helped us to improve this manuscript.
M.B.\ gratefully acknowledges the financial support of the 
National Group of Mathematical Physics GNFM-INdAM (Italy). 
M.P.\ is supported by a PhD scholarship (RGF\textbackslash EA\textbackslash 180270)
of the Royal Society (UK).
A.S.\ gratefully acknowledges the financial support of 
the Royal Society (UK) through a Royal Society University 
Research Fellowship (UF150099), a Research Grant (RG160517) 
and two Enhancement Awards (RGF\textbackslash EA\textbackslash 180270 and RGF\textbackslash EA\textbackslash 201051). 
L.W.\ is supported by the RTG 1670 ``Mathematics inspired 
by String Theory and Quantum Field Theory''.

\section*{Conflict of interest statement}
On behalf of all authors, the corresponding author states that there is no conflict of interest. 

%%%%%%%%%%%%%%%%%%%%%%%%%%%%%%%%%%%%%%%%%%%%%%%%
%%%%%%%%%%%%%%%%%%%%%%%%%%%%%%%%%%%%%%%%%%%%%%%%

\appendix

\section{\label{app:2operads}Basic theory of $\Cat$-enriched colored operads}
The aim of this appendix is to set up a suitable framework for $\Cat$-enriched colored operads
(which one could also call $2$-operads or $2$-multicategories) that will be used in this work. 
Our definitions of pseudo-morphisms, 
pseudo-transformations and modifications are a relatively
straightforward generalization of the analogous concepts from 
$2$-category theory (see e.g.\ \cite{Leinster,Lack,SchommerPries})
to the theory of colored operads (see e.g.\ \cite{Yau,BSWoperad}). 
We would like to note that our approach is slightly more flexible than the earlier one by
Corner and Gurski \cite{CornerGurski}, because we allow our pseudo-morphisms
to preserve permutation actions only up to coherent isomorphisms.
This generalization is necessary to capture the quantum field theoretical examples that we study in this work.
See also Remark \ref{rem:CornerGurski} for precise comment on the relationship to \cite{CornerGurski}.

\begin{defi}\label{def:2operad}
A  {\em $\Cat$-enriched colored operad} $\O$ consists of the following data:
\begin{enumerate}[(1)]
\item A collection $\O_0$. Elements are called objects and are denoted by symbols like $a,b,c\in\O$.

\item Categories $\O\big(\substack{t \\ \und{c}}\big)$, for each $t\in \O$ and each 
tuple $\und{c} := (c_1,\dots,c_n)\in \O^n$. Objects of $\O\big(\substack{t \\ \und{c}}\big)$ are 
called $1$-operations and are denoted by symbols like $\phi,\psi$. Morphisms of 
$\O\big(\substack{t \\ \und{c}}\big)$ are called $2$-operations and are denoted by 
symbols like $\alpha,\beta$. We write $\Id$ for the identity $2$-operations and $\alpha\,\beta$
for the (vertical) composition of $2$-operations.

\item Composition functors
$\gamma : \O\big(\substack{t \\ \und{a}}\big) \times \prod_{i=1}^n \O\big(\substack{a_i \\ \und{b}_i }\big)\to 
\O\big(\substack{t \\ \und{\und{b}}}\big)$,
for each $t\in \O$, $\und{a}\in\O^n$ and $\und{b}_i\in \O^{k_i}$, for $i=1,\dots,n$, where 
$\und{\und{b}} := (\und{b}_1,\dots,\und{b}_n)$. We write
$\phi\,\und{\psi} := \gamma(\phi,(\psi_1,\dots,\psi_n))$ for the composition of $1$-operations
and $\alpha\ast\und{\beta} := \gamma(\alpha,(\beta_1,\dots,\beta_n))$ for the (horizontal) 
composition of $2$-operations.

\item Functors $\oone : \1 \to \O\big(\substack{t\\t}\big)$, for each $t\in \O$, where $\1$ 
is the category with only one object and its identity morphism. We also write 
$\oone\in \O\big(\substack{t\\t}\big)$ for the corresponding identity $1$-operation.

\item Permutation functors $\O(\sigma) : \O\big(\substack{t \\ \und{c}}\big)\to 
\O\big(\substack{t \\ \und{c}\sigma}\big)$, for each $t\in \O$, $\und{c}\in\O^n$ and permutation
$\sigma\in\Sigma_n$, where $\und{c}\sigma := (c_{\sigma(1)},\dots,c_{\sigma(n)})$. 
We write $\phi\cdot \sigma := \O(\sigma) (\phi)$ and $\alpha\cdot\sigma :=\O(\sigma) (\alpha)$
for the permutation action on $1$- and $2$-operations.
\end{enumerate}
These data are required to satisfy the usual permutation action, associativity, unitality and equivariance
axioms, see e.g.\ \cite[Definition 11.2.1]{Yau}.
\end{defi}

\begin{defi}\label{def:pseudomorphism}
Let $\O$ and $\P$ be $\Cat$-enriched colored operads. A {\em pseudo-morphism} 
$F : \O\to \P$ consists of the following data:
\begin{enumerate}[(1)]
\item A function $F : \O_0\to \P_0$.
\item Functors $F : \O\big(\substack{t\\ \und{c}}\big) \to \P\big(\substack{Ft\\F\und{c}}\big)$,
for each $t\in\O$ and $\und{c}\in\O^n$, where $F\und{c} := (Fc_1,\dots,F c_n)$.
\item Natural isomorphisms 
\begin{flalign}
\xymatrix@C=5em{
\ar[d]_-{\gamma^\O}\O\big(\substack{t \\ \und{a}}\big) \times \prod\limits_{i=1}^n \O\big(\substack{a_i \\ \und{b}_i }\big)\ar[r]^-{F\times\prod_i F}~&~\P\big(\substack{Ft \\ F\und{a}}\big) \times \prod\limits_{i=1}^n \P\big(\substack{F a_i \\ F\und{b}_i }\big)\ar[d]^-{\gamma^\P}\ar@{=>}[dl]_-{F^2~}\\
\O\big(\substack{t \\ \und{\und{b}}}\big)\ar[r]_-{F}~&~\P\big(\substack{Ft \\ F\und{\und{b}}}\big)
}
\end{flalign}
for each $t\in \O$, $\und{a}\in\O^n$ and $\und{b}_i\in \O^{k_i}$, for $i=1,\dots,n$.

\item Natural isomorphisms
\begin{flalign}
\xymatrix@R=1.5em@C=2.5em{
\ar[dd]_-{\oone^\O} \1 \ar[ddrr]^-{\oone^P} ~&&~\\
&\ar@{=>}[dl]_-{F^0}&\\
\O\big(\substack{t\\t}\big)\ar[rr]_-{F}~&&~\P\big(\substack{Ft\\Ft}\big)
}
\end{flalign}
for each $t\in\O$.

\item Natural isomorphisms
\begin{flalign}\label{eqn:sigmafiller}
\xymatrix@C=5em{
\ar[d]_-{\O(\sigma)}\O\big(\substack{t \\ \und{c}}\big) \ar[r]^-{F} ~&~ \P\big(\substack{Ft \\ F\und{c}}\big)\ar[d]^-{\P(\sigma)} \ar@{=>}[dl]_-{F^\sigma~}\\
\O\big(\substack{t \\ \und{c}\sigma}\big)\ar[r]_-{F}~&~ \P\big(\substack{Ft \\ F\und{c}\sigma}\big)
}
\end{flalign}
for each $t\in \O$, $\und{c}\in\O^n$ and $\sigma\in\Sigma_n$.
\end{enumerate}
These data are required to satisfy the following axioms:
\begin{subequations}
\begin{flalign}
\xymatrix@C=5em{
\ar[d]_-{\Id\ast \prod F^2} (F\phi)\,(F\und{\psi})\, (F\und{\und{\rho}}) \ar[r]^-{F^2\ast \prod \Id}~&~F(\phi\,\und{\psi}) \,(F\und{\und{\rho}})\ar[d]^-{F^2}\\
(F\phi) \,F(\und{\psi}\,\und{\und{\rho}})\ar[r]_-{F^2}~&~F(\phi\,\und{\psi}\,\und{\und{\rho}})
}
\end{flalign}
\begin{flalign}
\xymatrix@C=5em{
\ar[d]_-{F^0 \ast\Id}\oone^\P \, (F\phi) \ar[rd]^-{\Id}~&~ & \ar[d]_-{\Id\ast\prod F^0}(F\phi)\,\prod\oone^\P 
\ar[rd]^-{\Id} ~&~\\
(F\oone^\O)\,(F\phi)\ar[r]_-{F^2}~&~F(\oone^\O\,\phi) & (F\phi)\prod F\oone^\O \ar[r]_-{F^2} ~&~F(\phi\,\prod\oone^\O)
}
\end{flalign}
\begin{flalign}
\xymatrix@C=5em{
\ar[d]_-{\Id} ((F\phi)\cdot \sigma)\cdot\sigma^\prime \ar[r]^-{F^\sigma \cdot \sigma^\prime} ~&~ (F(\phi\cdot\sigma))\cdot \sigma^\prime\ar[d]^-{F^{\sigma^\prime}} & \ar[d]_-{\Id}F\phi \ar[dr]^-{\Id} ~&~\\
(F\phi)\cdot(\sigma\sigma^\prime)\ar[r]_-{F^{\sigma\sigma^\prime}}~&~F(\phi\cdot(\sigma\sigma^\prime)) & (F\phi)\cdot e \ar[r]_-{F^e} ~&~F(\phi\cdot e)
}
\end{flalign}
\begin{flalign}
\xymatrix@C=8em{
\ar[d]_-{\Id}((F\phi)\,(F\und{\psi}))\cdot\sigma\langle k_1,\dots,k_n\rangle \ar[r]^-{F^2\cdot \sigma\langle k_1,\dots,k_n\rangle } ~&~(F(\phi\,\und{\psi}))\cdot \sigma\langle k_1,\dots,k_n\rangle\ar[d]^-{F^{\sigma\langle k_1,\dots,k_n\rangle}}\\
\ar[d]_-{F^\sigma\ast\prod \Id}((F\phi)\cdot\sigma)\,(F\und{\psi}\sigma)~&~F\big((\phi\,\und{\psi})\cdot \sigma\langle k_1,\dots,k_n\rangle\big)\ar[d]^-{\Id}\\
F(\phi\cdot\sigma)\,(F\und{\psi}\sigma)\ar[r]_-{F^2}~&~F\big((\phi\cdot \sigma)\,(\und{\psi}\sigma)\big)
}
\end{flalign}
\begin{flalign}
\xymatrix@C=8em{
\ar[d]_-{\Id}((F\phi)\,(F\und{\psi}))\cdot(\sigma_1\oplus\cdots\oplus\sigma_n)\ar[r]^-{F^2\cdot(\sigma_1\oplus\cdots\oplus\sigma_n) }  ~&~(F(\phi\,\und{\psi}))\cdot (\sigma_1\oplus\cdots\oplus\sigma_n)\ar[d]^-{F^{\sigma_1\oplus\cdots\oplus\sigma_n}}\\
\ar[d]_-{\Id\ast\prod F^{\sigma_i}}(F\phi)\,\big((F\und{\psi})\cdot(\sigma_1\oplus\cdots\oplus\sigma_n)\big)~&~F\big((\phi\,\und{\psi})\cdot (\sigma_1\oplus\cdots\oplus\sigma_n)\big)\ar[d]^-{\Id}\\
(F\phi)\, \big(F(\und{\psi}\cdot (\sigma_1\oplus \cdots \oplus \sigma_n))\big)\ar[r]_-{F^2}~&~ 
F\big(\phi\,(\und{\psi}\cdot(\sigma_1\oplus \cdots \oplus \sigma_n))\big)
}
\end{flalign}
\end{subequations}
\end{defi}
\begin{rem}\label{rem:CornerGurski}
In the case all coherences $F^\sigma$ for permutation actions in \eqref{eqn:sigmafiller} 
are identities, our concept of pseudo-morphisms specializes to \cite[Definition 2.2]{CornerGurski}. 
We however require the more flexible Definition \ref{def:pseudomorphism} in the present paper, 
because our quantum field theoretical examples of interest generically come with non-trivial coherences
$F^\sigma$.
\end{rem}

\begin{defi}\label{def:pseudotransformation}
Let $\O$ and $\P$ be $\Cat$-enriched colored operads
and $F,G : \O\to\P$ pseudo-morphisms. A {\em pseudo-transformation}
$\zeta: F\Rightarrow G$ consists of the following data:
\begin{enumerate}[(1)]
\item Functors $\zeta_c : \1\to \P\big(\substack{Gc \\ Fc}\big)$, for each $c\in \O$. We also
write  $\zeta_c \in\P\big(\substack{Gc \\ Fc}\big)$ for the corresponding $1$-operation.

\item Natural isomorphisms
\begin{flalign}
\xymatrix@R=1.5em@C=5em{
\ar[d]_-{\cong}\O\big(\substack{t \\ \und{c}}\big)\times\prod\limits_{i=1}^n \1\ar[r]^-{G\times\prod_i\zeta_{c_i}}~&~
\P\big(\substack{Gt \\ G\und{c}}\big) \times\prod\limits_{i=1}^n\P\big(\substack{Gc_i\\ Fc_i}\big)\ar[dd]^-{\gamma^\P}\ar@{=>}[ddl]_-{\zeta^\bullet~}\\
\1 \times \O\big(\substack{t \\ \und{c}}\big)\ar[d]_-{\zeta_t\times F}~&~\\
\P\big(\substack{Gt\\ Ft}\big) \times \P\big(\substack{Ft\\ F\und{c}}\big)\ar[r]_-{\gamma^\P}~&~\P\big(\substack{Gt \\ F\und{c}}\big)
}
\end{flalign}
for each $t\in\O$ and $\und{c}\in\O^n$.
\end{enumerate}
These data are required to satisfy the following axioms:
\begin{subequations}
\begin{flalign}
\xymatrix@C=5em{
\ar[d]_-{G^2\ast \prod\Id}(G\phi)\,(G\und{\psi})\,\prod\zeta_{b_{ij}} \ar[r]^-{\Id\ast \prod\zeta^\bullet} ~&~(G\phi) \,\prod \zeta_{a_i} \, (F\und{\psi}) \ar[r]^-{\zeta^\bullet\ast\prod\Id} ~&~\zeta_t\,(F\phi)\,(F\und{\psi})\ar[d]^-{\Id\ast F^2}\\
G(\phi\,\und{\psi})\,\prod\zeta_{b_{ij}}\ar[rr]_-{\zeta^\bullet}~&~ ~&~\zeta_t\,F(\phi\,\und{\psi})
}
\end{flalign}
\begin{flalign}
\xymatrix@C=5em{
\ar[d]_-{\Id}\oone^\P\, \zeta_t \ar[r]^-{G^0\ast\Id}~&~(G\oone^\O)\,\zeta_t\ar[d]^-{\zeta^\bullet}\\
\zeta_t\,\oone^\P \ar[r]_-{\Id\ast F^0}~&~\zeta_t\,(F\oone^\O)
}
\end{flalign}
\begin{flalign}
\xymatrix@C=5em{
\ar[d]_-{\Id}((G\phi)\,\prod\zeta_{c_i})\cdot \sigma \ar[r]^-{\zeta^\bullet\cdot\sigma} ~&~(\zeta_t\,(F\phi))\cdot \sigma\ar[d]^-{\Id}\\
\ar[d]_-{G^\sigma\ast\prod\Id}((G\phi)\cdot\sigma)\,\prod\zeta_{c_{\sigma(i)}}~&~\zeta_t\,((F\phi)\cdot \sigma)\ar[d]^-{\Id\ast F^\sigma}\\
G(\phi\cdot\sigma)\,\prod\zeta_{c_{\sigma(i)}}\ar[r]_-{\zeta^\bullet}~&~\zeta_t\,F(\phi\cdot\sigma)
}
\end{flalign}
\end{subequations}
\end{defi}

\begin{defi}\label{def:modification}
Let $\O$ and $\P$ be $\Cat$-enriched colored operads,
$F,G : \O\to\P$ pseudo-morphisms and $\zeta,\kappa : F\Rightarrow G$ pseudo-transformations.
A {\em modification} $\Gamma : \zeta \Rrightarrow\kappa$ consists of the following data:
\begin{enumerate}[(1)]
\item Natural transformations
\begin{flalign}
\xymatrix@R=0.5em@C=2.5em{
&\ar@{=>}[dd]^-{\Gamma_c}&\\
\1\vphantom{\big(\substack{Gc \\ Fc}\big)} \,   \ar@/^1.6pc/[rr]^-{\zeta_c} \ar@/_1.6pc/[rr]_-{\kappa_c}&~~&~ \P\big(\substack{Gc \\ Fc}\big)\\
&&\\
}
\end{flalign}
for each $c\in\O$.
\end{enumerate}
These data are required to satisfy the following axioms:
\begin{flalign}
\xymatrix@C=5em{
\ar[d]_-{\zeta^\bullet} (G\phi)\prod\zeta_{c_i} \ar[r]^-{\Id\ast\prod\Gamma_{c_i}} ~&~(G\phi)\prod\kappa_{c_i}\ar[d]^-{\kappa^\bullet}\\
\zeta_t\,(F\phi)\ar[r]_-{\Gamma_t\ast \Id}~&~\kappa_t\,(F\phi)
}
\end{flalign}
\end{defi}

\begin{rem}\label{rem:HomCats}
$\Cat$-enriched colored operads, pseudo-morphisms, pseudo-transformations and modifications
assemble into a tricategory. The various compositions are similar to the case of the tricategory of 
bicategories and hence will not be displayed in full detail here. We refer the reader to
\cite[Appendix A.1]{SchommerPries} for a brief review of the tricategory of bicategories
and to \cite{GPS95} for the details.
\sk

Let us nevertheless fix the relevant notations that will appear in the bulk of this paper.
Given two $\Cat$-enriched colored operads $\O$ and $\P$, the tricategory structure implies that
there exists a $\Hom$-$2$-category
\begin{flalign}
\Alg_\O(\P) \,:=\, [\O,\P]\,\in\, \mathbf{2Cat}\quad,
\end{flalign}
whose objects are pseudo-morphisms from $\O$ to $\P$, $1$-morphisms are pseudo-transformations
and $2$-morphisms are modifications. Following the usual terminology of operad theory,
we shall call $\Alg_\O(\P)$ the {\em $2$-category of $\O$-algebras with values in $\P$}.
Given pseudo-morphisms $F : \O\to\O^\prime$
and $G : \P\to \P^\prime$, there exist pseudo-functors
\begin{flalign}\label{eqn:pullpushApp}
F^\ast\,:\, [\O^\prime,\P]~\longrightarrow~[\O,\P]\quad,\qquad G_\ast \,:\, [\O,\P] ~\longrightarrow~[\O,\P^\prime]\quad,
\end{flalign}
which we call {\em pullback} and {\em pushforward}.
\sk

Let us note that in the case $\O$ and $\P$ are $\Set$-valued colored operads, 
i.e.\ all categories of operations in Definition \ref{def:2operad} are discrete, 
then $\Alg_\O(\P) = [\O,\P]\in\Cat$ is an ordinary category that coincides with
the usual category of $\O$-algebras with values in $\P$, see e.g.\ \cite{Yau,BSWoperad}.
\end{rem}

%%%%%%%%%%%%%%%%%%%%%%%%

\end{document}